\documentclass[journal]{IEEEtran}
\ifCLASSINFOpdf
\else
\fi
\usepackage{multirow}
\usepackage{booktabs}
\usepackage{graphicx}
\usepackage{epsfig}
\graphicspath{{image/}}
\usepackage{subfigure}
\usepackage{graphicx}
\usepackage{multirow}
\usepackage{booktabs}
\usepackage{amsfonts}
\usepackage{amsopn}
\usepackage{amsmath,amsthm}
\usepackage{amssymb,amsfonts}
\usepackage{algorithm}
\usepackage{algorithmic}
\usepackage{tabularx}
\usepackage{epstopdf}
\usepackage{epsfig}
\usepackage{graphicx}
\usepackage{float}
\usepackage{bbding}
\usepackage{tabularx}
\usepackage{float}
\graphicspath{{image/}}
\usepackage{color,xcolor}
\usepackage[numbers,sort&compress]{natbib}
\usepackage{footnote}
\usepackage{url}
\usepackage{multirow}
\usepackage{booktabs}
\usepackage{soul}
\newtheorem{lemma}{Lemma}
\newtheorem{thm}{\bf Theorem}

\begin{document}
%
\title{ Multi-modal and frequency-weighted tensor nuclear norm for hyperspectral image denoising }
%
%
%

\author{Xiaozhen Xie, and Sheng Liu
\thanks{Xiaozhen Xie, and Sheng Liu are with College of Science, Northwest A\&F University, Yangling 712100, China (e-mail: xiexzh@nwafu.edu.cn (Corresponding author); e-mail: liu\_sheng@nwafu.edu.cn).}
}

%
%

\markboth{Journal of \LaTeX\ Class Files,~Vol.~14, No.~8, August~2015}%
{Shell \MakeLowercase{\textit{et al.}}: Bare Demo of IEEEtran.cls for IEEE Journals}
%



\maketitle

\begin{abstract}
Low-rankness is important in the hyperspectral image (HSI) denoising tasks.
The tensor nuclear norm (TNN), defined based on the tensor singular value decomposition, is a state-of-the-art method to describe the low-rankness of HSI.
However, TNN ignores some physical meanings of HSI in tackling denoising tasks, leading to suboptimal denoising performance.
In this paper, we propose the multi-modal and frequency-weighted tensor nuclear norm (MFWTNN) and the non-convex MFWTNN for HSI denoising tasks.
Firstly, we investigate the physical meaning of frequency slices and reconsider their weights to improve the low-rank representation ability of TNN.
Secondly, we consider the correlation among two spatial dimensions and the spectral dimension of HSI and combine the above improvements to TNN to propose MFWTNN.
Thirdly, we use non-convex functions to approximate the rank function of the frequency tensor and propose the NonMFWTNN to relax the MFWTNN better.
Besides, we adaptively choose bigger weights for slices mainly containing noise information and smaller weights for slices containing profile information.
Finally, we develop the efficient alternating direction method of multiplier (ADMM) based algorithm to solve the proposed models,
and the effectiveness of our models are substantiated in simulated and real HSI datasets.

\end{abstract}



%
\IEEEpeerreviewmaketitle

\section{Introduction}
\IEEEPARstart{H}{yperspectral}
 image (HSI) has been widely used in many fields due to its wealthy spectral information of real scenes, e.g., target detection \cite{Target}, unmixing \cite{unmixing_zhou}, classification \cite{classification01} and so on.
Unfortunately, due to imaging conditions and the influence of weather, observed HSIs are often accompanied by the Gaussian noise, salt and pepper noise, and stripe noise.
This not only reduces the quality of HSIs, but also hinders subsequent applications.
Therefore, to improve the accuracy of subsequent applications, an important research topic is to recover the clear HSI from the observed HSI in the hyperspectral processing tasks \cite{zengTGRS,LLRSSTV}.

Since the spectral bands are the imaging results of the same scene under different wavelengths, they have global correlation or low-rankness.
Based on this prior, HSI can be unfolded into matrices along the spectral mode and then processed by matrix restoration methods.
The low-rank matrix recovery-based method restores the clean HSI by minimizing the ranks of unfolding matrices, which is the most popular way \cite{LRTA,LRMR,LRTV,zengTGRS}.
Zhang et al. \cite{LRMR} propose to disassemble the noisy HSI into the clean HSI and noise information and use the matrix nuclear norm as the relaxation of the rank function.
To improve the low-rank performance of the model, some non-convex methods are introduced into the HSI denoising task.
Xie et al. \cite{logsum} use the $\log sum$ function, Chen et al. \cite{Non_LRMA} use the $\gamma$ function, Lou et al. \cite{L1_L2} use the $L_1-L_2$ function and Gao et al.\cite{SCAD_MCP} use the SCAD and MCP function as the non-convex approximation of the rank function.
Low-rank matrix decomposition is also an effective way to solve the low-rank minimization problem \cite{BFTV,NLMC}.

However, the unfolding operator destroys the high-dimensional structure of HSI.
Treating HSI as a third-order tensor can protect this structure, and tensor rank minimization is the most effective way to overcome the above shortcoming.
Due to the non-unique definitions of the tensor rank, different tensor decompositions and the corresponding tensor ranks are proposed,
e.g.,  the CANDECOMP/PARAFAC (CP) decomposition \cite{CP,CP_r1}, Tucker decomposition \cite{tucker_1,HaLRTC} and tensor singular value decomposition (t-SVD) \cite{t_SVD,TNN}.
CP rank is defined as the minimum of rank-1 tensors by CP decomposition.
Whereas, it is difficult to estimate when HSI is unknown.
By Tucker decomposition, a tensor can be decomposed into core tensor and modal matrices.
Then, Tucker rank is approximated by the sum of modal matrix nuclear norms (SNN) \cite{HaLRTC}.
However, SNN has shortcomings in preserving low-rank structure, and its results are materially sub-optimal \cite{TRPCA,SNN_shortcoming}.

Due to the powerful representation ability of neural networks, the models based-on deep learning have achieved great achievements in HSI denoising tasks.
There are many models that implicitly learn the prior information of HSI from the data, which cannot be represented by traditional regular expressions, e.g., HSID-CNN \cite{HSID-CNN}, HSI-SDeCNN \cite{HSI-SDeCNN}, HSI-DeNet \cite{HSI-DeNet}, QRNN3D \cite{QRNN3D} and so on.
The above models explore the prior information of HSI from the perspective of deep network. 
However, the above models require abundant labeled data for training, which limits their denoising performance by the training data’s adversity and quantity \cite{DS2DP}. 
A more suitable strategy is to combine the advantages of deep learning and traditional regularization \cite{ZENG_HSI_tensor, zhang2020deep, luo2021hyperspectral,wang2021hyperspectral,DS2DP, zhuang2021hyperspectral, rui2021learning}.

Based on t-SVD, the tensor tubal rank has attracted extensive attention recently \cite{zengTCI,LRTRfan, WSTNN, FWTNN}.
Its convex relaxation is the tensor nuclear norm (TNN), which is defined as the sum of the matrix nuclear norms of each frontal slice in the Fourier transformed tensor \cite{TRPCA,t_SVD}.
The minimization of TNN can be quickly solved by convex optimization algorithms, and it is effective in characterizing the low-rankness of tensor.
Therefore, it is widely used in HSI denoising task \cite{LRTRfan,SSTVLRTF, WSTNNSSTV}.
In the definition process of TNN, the information in the image domain of HSI is transformed into frequency information in the Fourier domain.
In the Fourier transformed tensor, the low-frequency slices mainly carry the profile information, while the high-frequency slices mainly carry detail and noise information.
Besides, in each frontal slice of the Fourier transformed tensor, bigger singular values mainly contain information of clean data and smaller singular values mainly contain information of noise.
Meanwhile, HSIs have correlations among different modes, but TNN lacks the flexibility to explore different correlations along with different modes \cite{WSTNN}.
Thus, it is advantageous to improve the flexibility of TNN when considering the prior information about the physical meaning of frequency, correlation among different modes, and singular value.

In this paper, to take full advantage of the above prior information and improve the capability and flexibility of model,
we propose the multi-modal and frequency-weighted tensor nuclear norm (MFWTNN) and non-convex MFWTNN. 
In the proposed MFWTNN, we explore the physical meaning of frequency to explore the low-rankness of HSI in the Fourier domain and give an adaptive calculation method of frequency weights.
In addition, we consider the correlation among the spatial modes and the spectral mode.
Furthermore, based on MFWTNN, we take into account the physical meaning of singular values inside frequency slices and propose the non-convex approximation of their nuclear norms.
Finally, we apply them to the HSI denoising task.
The main contributions of this paper are summarized as follows:

\begin{itemize}
	\item Motivated by Tucker decomposition, we consider the correlation among spectral and spatial modes of HSI and use the mode-$p$ permutation of tensor instead of the mode-$p$ unfolding matrix.
	And then we use their weighted sum to explore the low-rankness of HSI.
	Furthermore, we combine the advantages of multi modes with the improvement of TNN in the frequency domain and propose MFWTNN in HSI denoising.
	\item In each frontal slice of the Fourier transformed tensor,
	we choose the $\log sum$ function as a non-convex approximation to the rank function, which can provide a better approximation to rank function compared with nuclear norms in MFWTNN.
	Different from the one-dimensional weights of the classical non-convex models, NonMFWTNN considers the weight of frequency direction and singular value direction, which is the two-dimensional weight.
	\item According to information types in different frequency slices in the Fourier transformed tensor,
	we adaptively choose bigger weights for slices mainly containing noise information and smaller weights for slices containing profile information,
	which can depress noise more and simultaneously preserve the profile information of clean HSI better.
	\item We develop the efficient alternating direction method of multiplier based algorithm to solve the proposed models,
	and obtain the best restorative performance both on the simulated and real HSI dataset in comparison to all competing HSI denoised methods.
	
\end{itemize}
\section{Notations and Preparations }
\subsection{Notations}
In this section, we follow the work in \cite{tensor_basic,3DTNN} and give some basic notations used about matrix and tensor.
We use bold upper-case letter $\mathbf{X}$ and calligraphic letter $\mathcal{X}$ to denote matrix and tensor, respectively.
For a 3rd-order tensor $\mathcal{X}\in\mathbb{R}^{n_1 \times n_2 \times n_3}$, its $(i,j,k)$-th component is denoted as $\mathcal{X}(i,j,k)$.
For $\mathcal{X}$, $\mathcal{Y} \in\mathbb{R}^{n_1 \times n_2 \times n_3}$, their inner product is defined as $<\mathcal{X},\mathcal{Y}>=\sum_{i=1}^{n_1}\sum_{j=1}^{n_2}\sum_{k=1}^{n_3}x_{ijk}y_{ijk}$.
Then the Frobenius norm of a tensor $\mathcal{X}$ is computed as $\|\mathcal{X}\|_{F}=\sqrt{<\mathcal{X},\mathcal{X}>}$.
The $\ell_1$-norm of tensor $\mathcal{X}$ was calculated as $\|\mathcal{X}\|_1=\sum_{i=1}^{n_1}$ $\sum_{j=1}^{n_2}\sum_{k=1}^{n_3}|x_{ijk}|$.
The $k$-th frontal slice of $\mathcal{X}$ is represented as $\mathbf{X}^{(k)}=\mathcal{X}(:, :,k )$.
For one-dimensional vector $v$, its fast Fourier transform is denoted as $\bar{v}=\texttt{fft}(v)\in\mathbb{C}^{N}$, and inverse transformation is represented as $v=\texttt{ifft}(\bar{v})$.
For three-dimensional tensor $\mathcal{X}$, its fast Fourier transform is represented as $\bar{\mathcal{X}}=\texttt{fft}(\mathcal{X},[],3)$
and its inverse operation is $\mathcal{X}=\texttt{ifft}(\bar{\mathcal{X}},[],3)$.
The mode-$p$ permutation of $\mathcal{X}$ is defined as $\mathcal{X}_{p}=\texttt{permute}(\mathcal{X},p)$, $p=1,2,3$,
where the $m$-th mode-3 slice of $\mathcal{X}_{p}$ is the $m$-th mode-$p$ slice of $\mathcal{X}$, i.e.,
$\mathcal{X}(i,j,k)=\mathcal{X}_{1}(j,k,i)=\mathcal{X}_{2}(k,i,j)=\mathcal{X}_{3}(i,j,k)$.
Also, its inverse operation is  $\mathcal{X}=\texttt{ipermute}(\mathcal{X}_{p},p)$.

\subsection{Problem Formulation}
A clean HSI can be treated as a third-order tensor $\mathcal{X}\in\mathbb{R}^{n_1 \times n_2 \times n_3}$
and is usually assumed to be low-rank.
Corrupted by mixed noise, its observed version can be modeled as
\begin{equation}
\begin{aligned}
\label{eq_1}
\mathcal{Y} = \mathcal{X} + \mathcal{S} + \mathcal{N},
\end{aligned}
\end{equation}
where $\mathcal{Y}, \mathcal{S}, \mathcal{N} \in \mathbb{R}^{n_1 \times n_2 \times n_3}$;
$\mathcal{S}$ denotes the sparse noise;
$\mathcal{N}$ denotes the Gaussian white noise.

HSI denoising aims to recover the clean HSI $\mathcal{X}$ from the observed HSI $\mathcal{Y}$ in (\ref{eq_1}).
Under the framework of regularization theory, it can briefly be formulated as
\begin{equation}
\label{eq_2}
\begin{aligned}
\arg\min_{\mathcal{X}, \mathcal{S}, \mathcal{N}}  \texttt{Rank}(\mathcal{X}) + \lambda\|\mathcal{S}\|_1 + \tau\|\mathcal{N}\|_F^2,
s.t.  \mathcal{Y} = \mathcal{X} + \mathcal{S} + \mathcal{N},
\end{aligned}
\end{equation}
where $\|\cdot\|_1$ describes the sparse noise;
$\|\cdot\|_F$ describes the Gaussian noise;
$\texttt{Rank}(\cdot)$ represents the rank of unknown ideal HSI;
$\lambda$ and $\tau$  are non-negative parameters.

\subsection{Tensor Nuclear Norm  }
Since the tensor rank function is non-convex and discrete, its rank minimization problem is NP-hard \cite{Tensor_NP_hard}.
To overcome computational difficulties, a usual approach is to approximate the non-convex rank function to the convex function and then substitute it into the minimization problem.
Fazel \cite{fazel2002matrix}  proves that the matrix nuclear norm is the convex envelope rank function on the unit sphere of the spectral norm.
Thereafter, the nuclear norm is used as a convex replacement of the rank function.
The rank minimization problem is transformed into nuclear norm minimization, and it is widely used in low-rank matrix recover based HSI denoising tasks \cite{LRMR}.

Furthermore, HSI can be regarded as a third-order tensor.
Then, based on t-SVD \cite{t_SVD}, Lu et al. \cite{TRPCA} propose the tensor nuclear norm (TNN) as the relaxation function of the tensor rank function.
TNN is defined as
\begin{equation}
\label{TNN}
\|\mathcal{X}\|_{\mathrm{TNN}} := \frac{1}{n_3}\sum_{k=1}^{n_3}\|\bar{\textbf{X}}^{(k)}\|_* ,
\end{equation}
where $\bar{\textbf{X}}^{(k)}$is the $k$ th frontal slice of $\bar{\mathcal{X}}$.
Due to the property of $\texttt{fft}$, we have
\begin{equation}
\begin{aligned} \label{TNN_conj}
\left\{ \begin{array}{l}
\boldsymbol{\bar{\textbf{X}}}^{\left( 1 \right)}\in \mathbb{R}^{n_1\times n_2}\\
\texttt{conj}\left( \boldsymbol{\bar{\textbf{X}}}^{\left( i \right)} \right) =\boldsymbol{\bar{\textbf{X}}}^{\left( n_3-i+2 \right)},i=2,\cdots ,\lceil \frac{n_3+1}{2} \rfloor\\
\end{array} \right. .
\end{aligned}
\end{equation}
Based on \eqref{TNN_conj}, frequency component \cite{FWTNN} is defined as
\begin{equation}
\label{tensor_frequency_component}
\begin{aligned}
\left\{\overline{\mathcal{X}}_{i-1}\right\}=\left\{\begin{array}{ll}
\overline{\boldsymbol{\textbf{X}}}^{(1)} & \text { if } i=1 \\
\overline{\boldsymbol{\textbf{X}}}^{(i)}+\overline{\boldsymbol{\textbf{X}}}^{\left(n_{3}-i+2\right)} & \text { if } i=2, \cdots,\left\lceil\frac{n_{3}+1}{2}\right\rceil
\end{array}\right.
\end{aligned}
\end{equation}
Fig. \ref{fig:fwtnn01} shows the relationship between the frontal slice and the frequency component.
\begin{figure}[htbp]
	\centering
	\includegraphics[width=0.75\linewidth]{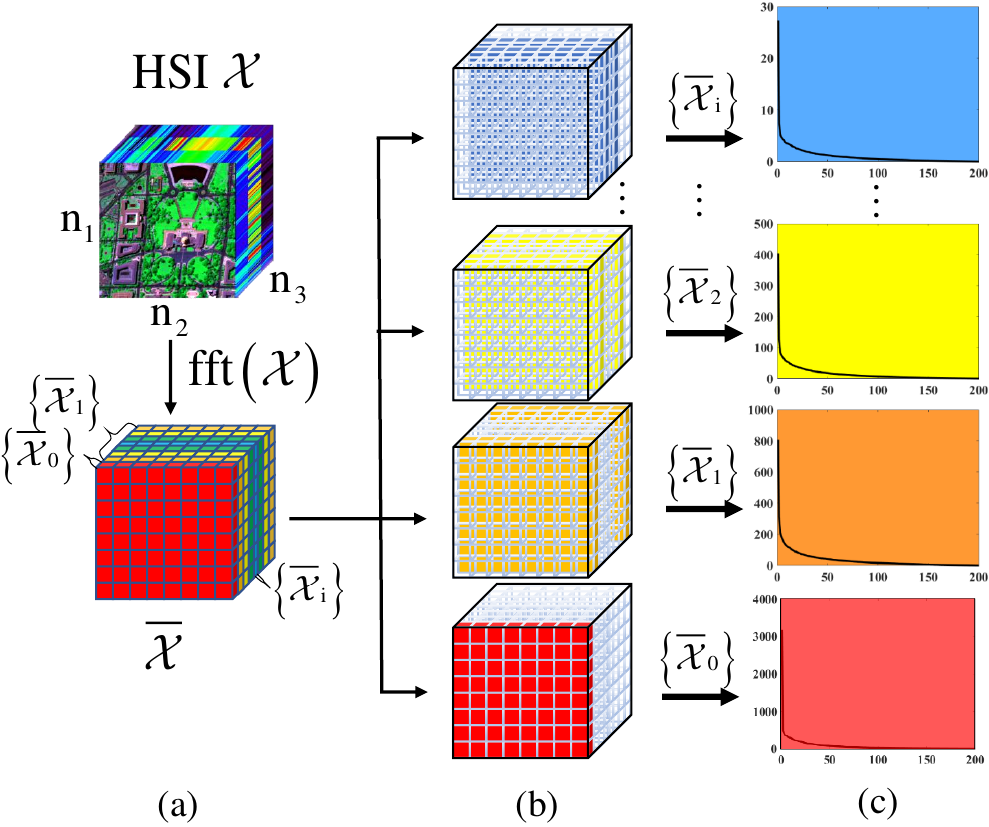}
	\caption{Illustration of the frequency component. (a) $\bar{\mathcal{X}}$ is the Fourier transform of HSI $\mathcal{X}$,(b)Frequency component of $\mathcal{X}$, (c) The singular value curve of the frequency component $\{\bar{\mathcal{X}}_i\}$ of $\bar{\mathcal{X}}$. }
	\label{fig:fwtnn01}
\end{figure}

Due to TNN can be effective in expressing the low-rank characteristics, it has received extensive attention in the HSI denoising task.
However, it focuses on the low-rankness in the spectral dimension.
For HSI, its low-rank structure exists with multiple dimensions \cite{LRTDTV,3DTNN}.
To connect the correlation along different dimensions, Liu et al. \cite{HaLRTC} use modal matrices of the tensor in different directions to explore low-rank structures, and then use SNN as the convex relaxation of the rank function.
SNN is defined as
\begin{equation}
\label{SNN}
\lVert \mathcal{X} \rVert _{\mathrm{SNN}}:=\sum_{p=1}^3{\alpha _p\lVert \boldsymbol{\textbf{X}}_{p} \rVert _*} ,
\end{equation}
where $\boldsymbol{\textbf{X}}_{p}$ is the modal matrix that unfold $\mathcal{X}$ along the direction of the mode-$p$, $\alpha _p $ is weighting parameter of modes and $\sum_{p=1}^3\alpha_p=1$.
However, this way will inevitably destroy the internal structure of HSI.
Zheng et al. \cite{3DTNN} propose the three-directional tensor nuclear norm (3DTNN). They consider the low-rankness of the three modes, which are two spatial modes and one spectral mode.
3DTNN is defined as
\begin{equation}
\label{3DTNN}
\lVert \mathcal{X} \rVert _{\mathrm{3DTNN}}:=\sum_{p=1}^3{\alpha _p\lVert \mathcal{X}_p \rVert _{\mathrm{TNN}}} ,
\end{equation}
where $\mathcal{X}_{p}$ is the mode-$p$ permutation of $\mathcal{X}$.
\begin{figure}[htbp]
	\centering
	\includegraphics[width=1\linewidth]{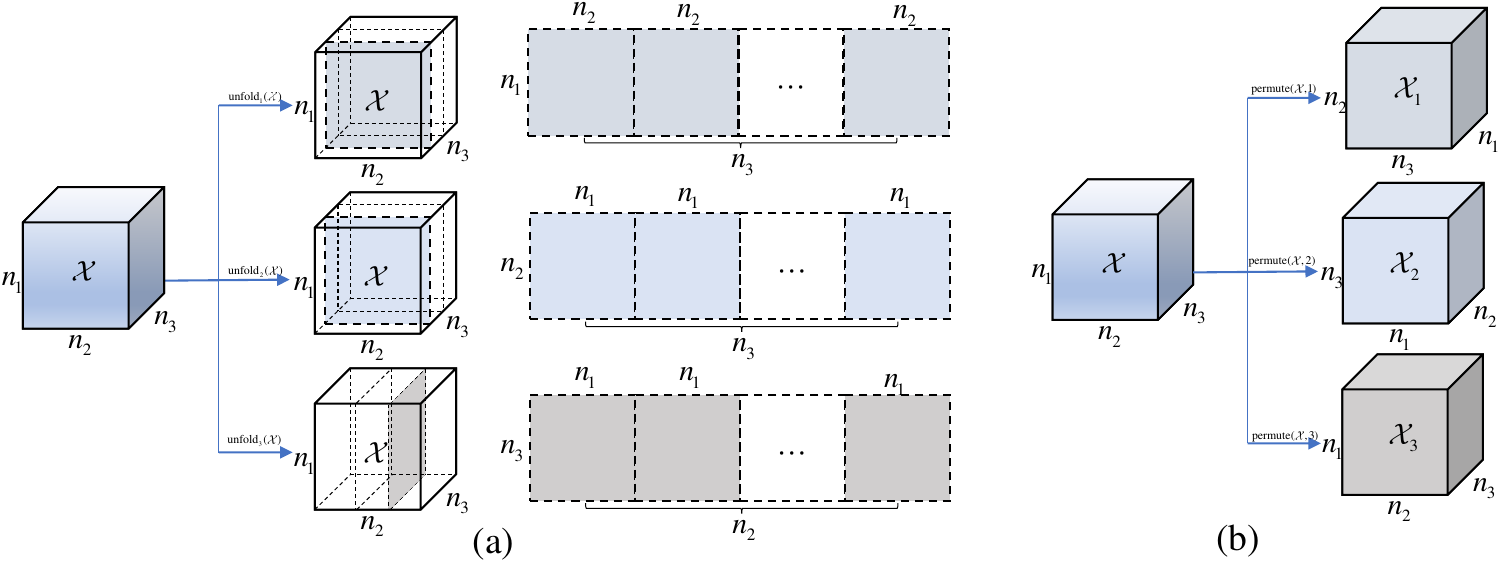}
	\caption{Mode-p unfolding and mode-p permutation of  an $n_1 \times n_2 \times n_3$ tensor. }
	\label{SNN_3DTNN}
\end{figure}
From Fig.\ref{SNN_3DTNN}, the mode-p permutation is more excellent to explore low-rank structures than the mode-p unfolding for third-order tensor, which avoids the unfolding operator.

\section{Multi-modal and Frequency-weighted Tensor Nuclear Norm and Its Non-convex Approximation. }
\subsection{Multi-modal and  Frequency-weighted Tensor Nuclear Norm }
In the definition of TNN \eqref{TNN}, $\texttt{fft}$ plays a core important role in the process of transforming t-products into matrix multiplication in the Fourier domain.
It makes TNN to have fast calculation capabilities.
However, the role of $\texttt{fft}$ is not just for fast computation in the framework of t-SVD.
When Fourier matrix $\boldsymbol{F}_n$ is treated as a transform matrix,  $\bar{\mathcal{X}}$ is the feature tensor of $\mathcal{X}$.
There are some similar work which supports this viewpoint \cite{Framelet_TNN, TTNN_DCT, TTNN_XU2019, TTNN_data,TTNN_linear_transforms,TTNN_Dictionary,TTNN_Self_Supervised,tensor_Qrank}.
To improve the low-rank representation ability of TNN,
Wang et al. \cite{FWTNN} propose the frequency-filtered tensor nuclear norm (FTNN), which is defined as
\begin{equation}
\label{FWTNN111}
\|\mathcal{X}\|_{\mathrm{FTNN}}=\frac{1}{I} \sum_{k=1}^{I} \alpha_{k}\left\|\{\overline{\mathcal{X}}\}_{k}\right\|_{*}, \quad I=\left\lceil\frac{n_{3}+1}{2}\right\rceil
\end{equation}
where  $\left\|\{\overline{\mathcal{X}}\}_{k}\right\|_{*}=\left\|\overline{\mathcal{X}}^{(k)}\right\|_{*}+\left\|\overline{\mathcal{X}}^{\left(n_{3}-k+2\right)}\right\|_{*} $ is the sum of two matrix nuclear norms,  $I$  denotes the number of frequency bands, and  $\alpha_{k} $ is a pre-defined parameter assigned to the  $k$-th frequency band, whose value depends on the prior knowledge.
This work treats the FTNN as a frequency filtering method. $\alpha_k$ characterizes the filtering coefficient for the $k$-th frequency component.
However, these parameters should be data-dependent, and the pre-weight parameters cannot well reflect the real physical meaning between filter coefficient and frequency component.
For the HSI, it is difficult to give $\left\lceil\frac{n_{3}+1}{2}\right\rceil$ pre-weight parameters in advance.
Therefore, we reconsider these weight parameters based on the physical properties of HSI.
For the clean Washington DC Mall dataset, sparse noise is added to form clean data, and then the zero-frequency component and non-zero frequency component of the noise data are separated.
\begin{figure}[htbp] \centering    	
	\subfigure[ ] {
		\label{Iidiam_clean}
		\includegraphics[width=0.21\columnwidth]{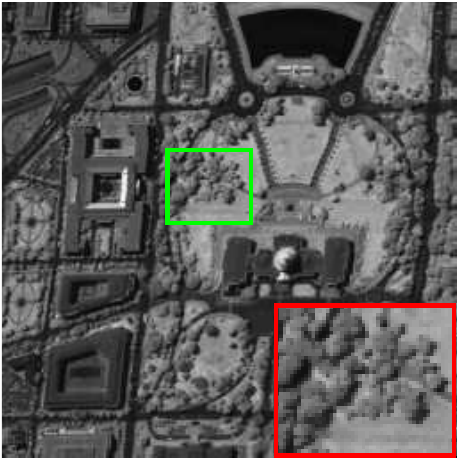}
	}
	\subfigure[  ] {
		\label{Iidiam_noise11}
		\includegraphics[width=0.21\columnwidth]{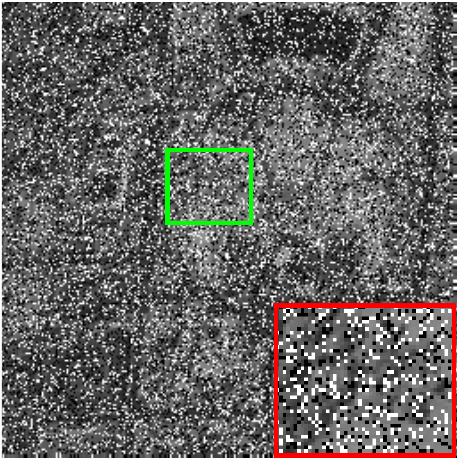}
	}
	\subfigure[  ] {
		\label{Iidiam_denoise1}
		\includegraphics[width=0.21\columnwidth]{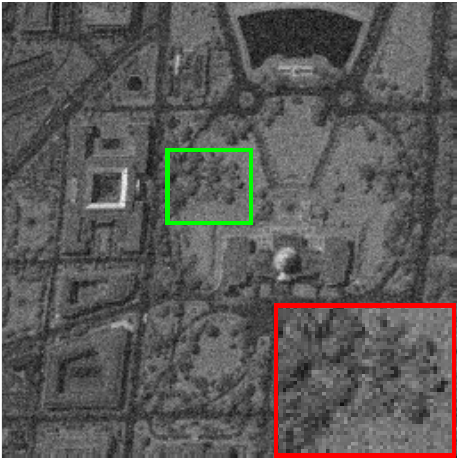}
	}
	\subfigure[ ] {
		\label{Iidiam_denoise_others}
		\includegraphics[width=0.21\columnwidth]{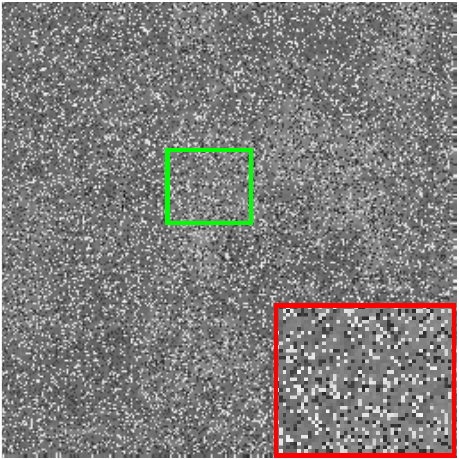}
	}\\
	\subfigure[ ] {
		\label{NNDC_ClearObservedNoisy1}
		\includegraphics[width=0.55\columnwidth]{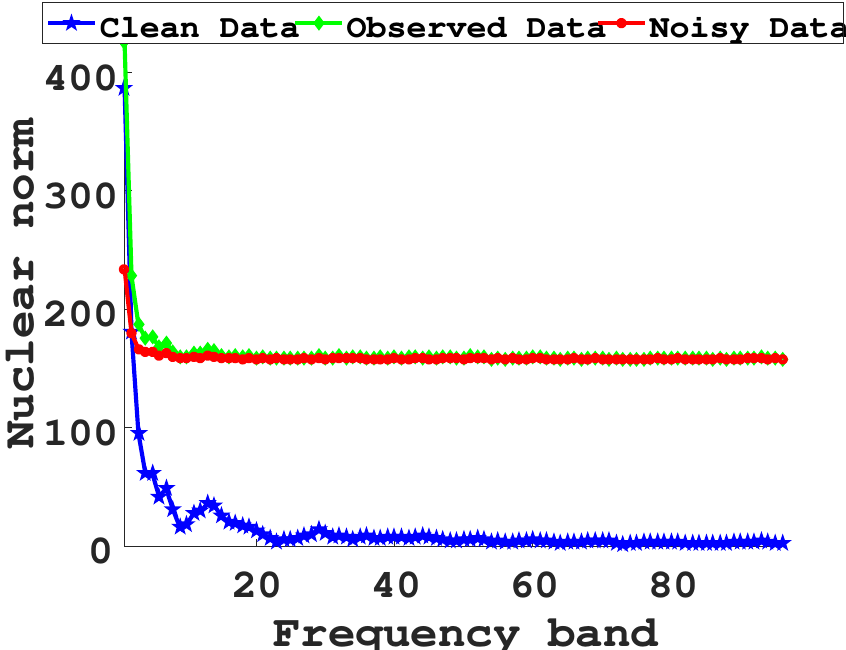}
	}\\
	\caption{Recovery of frequency band 60 of Washington DC Mall dataset. (a) clear data $\mathcal{X}$. (b) observed data $\mathcal{Y}$.  (c) zero frequency component of $\mathcal{Y}$. (d) non-zero frequency component of $\mathcal{Y}$. (e) the distribution of clean data, observed data, and noisy data in the frequency domain. }
	\label{Pavia_Recovery_frequency}
\end{figure}
Fig. \ref{Pavia_Recovery_frequency} shows the grayscale image of the 60th band.
We can find an interesting phenomenon that just using the information of the zero frequency component can get satisfactory results without any processing and the noise is more concentrated in the non-zero frequency component.
As the transform operator, $\texttt{fft}$ plays the role of information separation, which makes the image information and noise information as separate as possible.
The distribution of the nuclear norm shown in Fig. \ref{NNDC_ClearObservedNoisy1} is consistent with the observed phenomenon.
When the HSI is polluted by noise, the variation in the high-frequency component is bigger than that in the low-frequency component.
There is a natural idea that we can reduce the shrinkage of the low-frequency component and increase the penalty of the high-frequency component.
Therefore, base on FTNN, we redefine the frequency-weighted tensor nuclear norm (FWTNN):
\begin{equation}
\label{FWTNN}
\|\mathcal{X}\|_{FW*} := \sum_{k=1}^{n_3}w_{k}(\bar{\textbf{X}}^{(k)})\|\bar{\textbf{X}}^{(k)}\|_*.
\end{equation}
Our weights are different from the weights of FTNN, which are not pre-weight but data dependent.  They are discussed in Section \ref{FW_parameter}.

In the definition of FWTNN, by calculating the nuclear norm along the spectral direction,  $\texttt{fft}$ can capture low rank of the HSI spectrum.
However, the low-rankness of HSI exists not only in the spectral dimension, but also in the spatial dimensions \cite{LRTDTV,HaLRTC}.
Therefore, $\texttt{fft}$ is insufficient in exploring intra-mode and inter-mode correlations of HSI.
Besides, in the HSI denoising task, the HSI band may be polluted by stripe noise in the same direction.
In a low-rank model that only considers the spectral dimension, these stripe noises will be regarded as part of the low-rank image and are hard to remove.
In the frequency domain, for the same HSI, the information expressed by the permuted versions $\mathcal{X}_p$ in different directions is even more different.
Therefore, motivated by SNN \cite{HaLRTC}, it is necessary to consider the low-rankness of different modes for HSI denoising task.
	Different from SNN, to better retain the high-dimensional structures of HSI, its mode-$p$ permutation $\mathcal{X}_p$
	to replace its mode-$p$ unfolding matrix $\boldsymbol{\textbf{X}}_{p}$.
	
Based on the above analysis, we propose the multi-modal and frequency-weighted tensor nuclear norm (MFWTNN) as follows:
\begin{equation}
\label{MFWTNN}
\lVert \mathcal{X} \rVert _{MFW*}:=\sum_{p=1}^3{\alpha _p\lVert \mathcal{X}_p \rVert _{FW*}}=\sum_{p=1}^3{\sum_{k=1}^{n_3}{\alpha _pw_{k}^{p}\lVert \boldsymbol{\bar{\textbf{X}}}_{p}^{\left( k \right)} \rVert _*}} ,
\end{equation}
where $\bar{\mathcal{X}}_{p}=\texttt{fft}(\mathcal{X}_{p},[],3)$;
$\bar{\textbf{X}}^{(k)}_{p}$ is the $k$-th frontal slice of $\bar{\mathcal{X}}_{p}$ and its assigned weight is $w_{k}^{p}$;
$\alpha_p>0$ and $\sum_{p=1}^3\alpha_p=1$.
\subsection{Non-convex Multi-modal and Frequency-weighted Tensor Nuclear Norm }
In the same frequency slice, its singular values are treated equally.
However, the major information of HSI, such as smooth zones and profile, is contained in the larger singular value;
the noise information of HSI is contained in the smaller singular value \cite{zengSP,Non_LRMA,L1_L2,logsum}.
To combine both the frequency and the singular value prior information, we use the $\log$ norm to perform non-convex relaxation of the nuclear norm in the frequency band to more accurately describe the tensor rank function \cite{logsum}.
Thus, we propose the non-convex multi-modal and frequency-weighted tensor nuclear norm (NonMFWTNN) to more accurately describe the low-rankness of HSI. 
It is defined as follows:
\begin{equation}
\begin{aligned}
\lVert \mathcal{X} \rVert _{MFW*,Log}:&=\sum_{p=1}^3{\alpha _p\lVert \mathcal{X}_p \rVert _{FW*,Log}}\\
&=\sum_{p=1}^3{\sum_{k=1}^{n_3}{\alpha _pw_{k}^{p}\log \left( \lVert \boldsymbol{\bar{\textbf{X}}}_{p}^{\left( k \right)} \rVert _* \right)}},
\end{aligned}
\end{equation}
where $\log \left( \lVert \boldsymbol{\bar{X}}_{p}^{\left( k \right)} \rVert _* \right) =\sum_i^{n_{12}}{\left( \log \left( \sigma _i\left( \boldsymbol{\bar{X}}_{p}^{\left( k \right)} \right) +\varepsilon \right) \right)}$, $n_{12}=\texttt{min}(n_1,n_2)$.
Fig. \ref{fig:nonmfwtnnshiyitu06} shows the mechanism of NonMFWTNN.
As can be seen from Fig. \ref{fig:nonmfwtnnshiyitu06}, NonMFWTNN takes into account the relationship between the singular values inside and outside of frequency band.
The weights of classic matrix / tensor non-convex models \cite{Non_LRMA,logsum,3DTNN,PSTNN} are one-dimensional, which only consider the relationship between the singular values inside the matrix.
The way is equivalent to pulling the core singular value matrix $\overline{\mathbf{S}}_{p}$ in Fig. \ref{fig:nonmfwtnnshiyitu06}(c) into a column vector and using the weights as the coefficient of this column vector.
The weight of MFWTNN is to make a Cartesian product of the weight vector $W_F$ in the frequency direction and the weight vector $W_S$ in the singular value direction, which are two-dimensional weights.
If $W_F=1$, it degenerates to classical one-dimensional weights.
The table \ref{tab:similar} shows the distinction and relation between the proposed model and similar work in the weighted parameter.
It can be seen from Fig. \ref{fig:nonmfwtnnshiyitu06} and table \ref{tab:similar} that the key of the NonMFWTNN is that the core matrix is processed.
Liu et.al \cite{IRTPCA} claim that the core matrix is also low-rank.
\begin{table}[htbp]
	\centering
	\caption{The distinction and relation between the proposed model and similar work in the weighted parameter.}
	\scalebox{0.80}{
	\begin{tabular}{c|ccc}
		\hline
		& Modal direction  &Singular value direction & Frequency direction \\
		\hline
		NonMFWTNN &  \Checkmark     &  \Checkmark     & \Checkmark \\
		MFWTNN &    \Checkmark   &       & \Checkmark \\
		3DLogTNN\cite{3DTNN} &     \Checkmark  &     \Checkmark  &  \\
		3DTNN\cite{3DTNN} &    \Checkmark   &       &  \\
		FTNN \cite{FWTNN} &       &       & \Checkmark \\
		PSTNN \cite{PSTNN} &       &     \Checkmark  &  \\
		TNN  \cite{TRPCA} &       &       &  \\
		\hline
	\end{tabular}%
	\label{tab:similar}%
}
\end{table}%

\begin{figure*}
	\centering
	\includegraphics[width=1\linewidth]{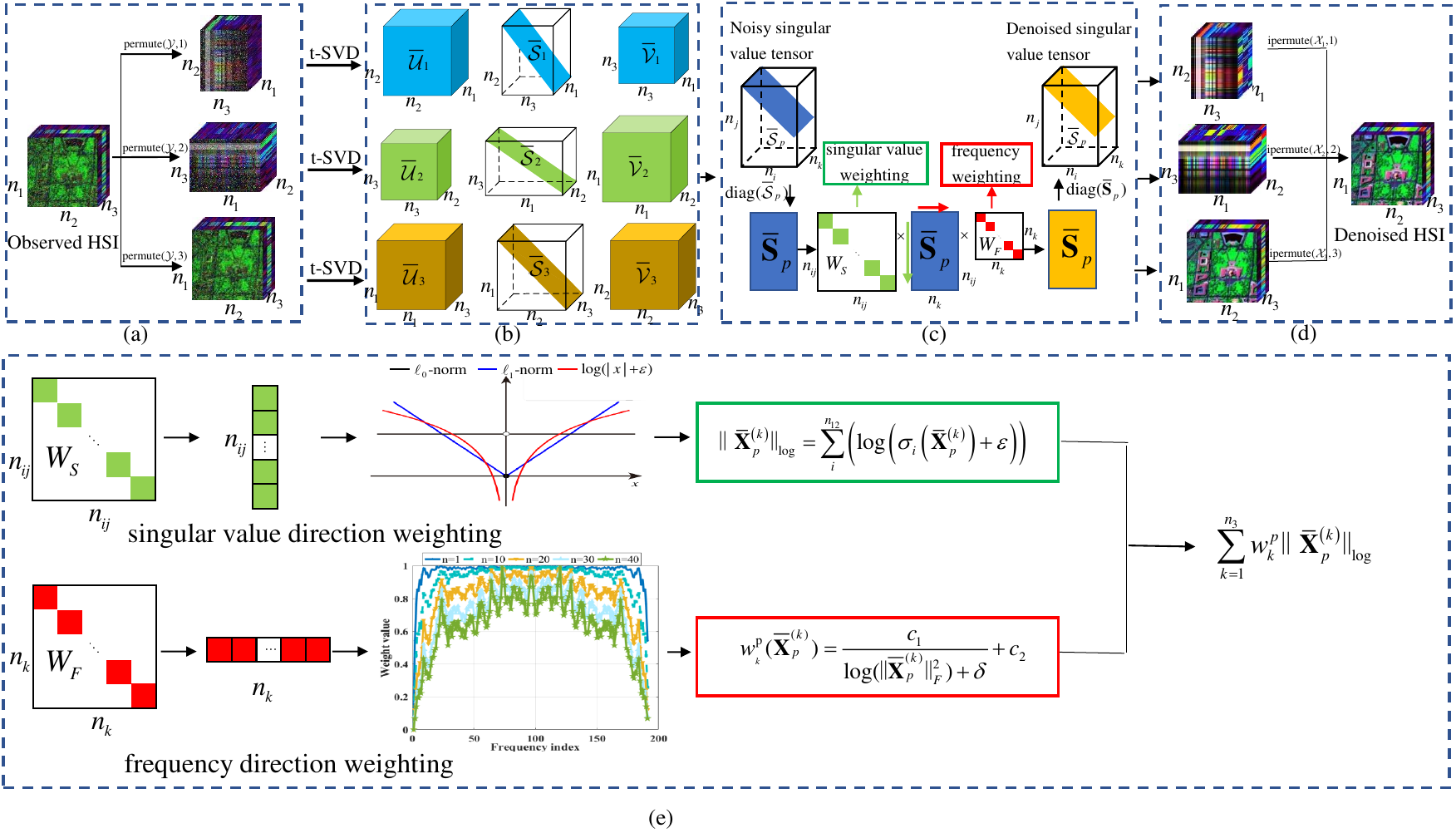}
	\caption{Illustration of NonMFWTNN. (a) Multi-modal permutations and Fourier transforms. (b) t-SVD of $\mathcal{X}_p$, p=1,2,3. (c) Shrinking the singular value tensor $ \bar{\mathcal{S}}_p$, where the size of $\mathcal{X}_p$ is $n_i\times n_j\times n_k$, $n_{ij}$=min$(n_i,n_j)$. The horizontal direction is the frequency weighting direction, and the vertical direction is the singular value weighting direction.  (d) Synthesis from multi-modal denoised results.  (e) The top is the singular value weighting, which is the weight of $\mathbf{\bar{S}}_p$ in the vertical direction. And the bottom is the frequency weighting, which is the weight of $\mathbf{\bar{S}}_p$ in the horizontal direction. }
	\label{fig:nonmfwtnnshiyitu06}
\end{figure*}
\subsection{Frequency-weighted parameter} \label{FW_parameter}

In the definition of MFWTNN and NonMFWTNN, the values of the frequency-weighted parameter $w_{k}^{p}$ are important.
The ideal  $w_{k}^{p}$ can take smaller values that reduce the penalty for the low-frequency slice nuclear norm to protect the main information of HSI in the low-frequency band, and take larger values that increase the penalty for the high-frequency slice nuclear norm to fully remove the noise in the high-frequency band.
It can be seen from Fig. \ref{nuclearnormPavia} that
the nuclear norm is opposite to the distribution characteristics of $w_{k}^{p}$ in the frequency domain.
It takes larger values in the low-frequency band and takes smaller values in the high-frequency band.
And the value of the nuclear norm decreases as the frequency increases.
However, the calculation of the nuclear norm brings high cost, which requires singular value decomposition.
In a finite dimensional space, any two norms are equivalent.
This shows that when transforming from the nuclear norm to the Frobenius norm, they are equal under the normalization condition.
Thus, the reciprocal of the Frobenius norm of each frequency slice matrix can be used as $w_{k}^{p}$.
\begin{equation}
\label{Weights}
w_{k}^{p}\left( \boldsymbol{\bar{\textbf{X}}}_{p}^{\left( k \right)} \right) =\frac{1}{\log \left( \lVert \boldsymbol{\bar{\textbf{X}}}_{p}^{\left( k \right)} \rVert _{F}^{2} \right) +\delta } ,
\end{equation}
where $\delta  = 10^{-6}$.
The purpose of using $\log$ function is to shrink the scale of the value without changing the trend of the Frobenius norm to prevent the extreme influence of the maximum value.
\begin{figure}[h] \centering    	
	\subfigure[ ] {
		\label{nuclearnormPavia}
		\includegraphics[width=0.400\columnwidth]{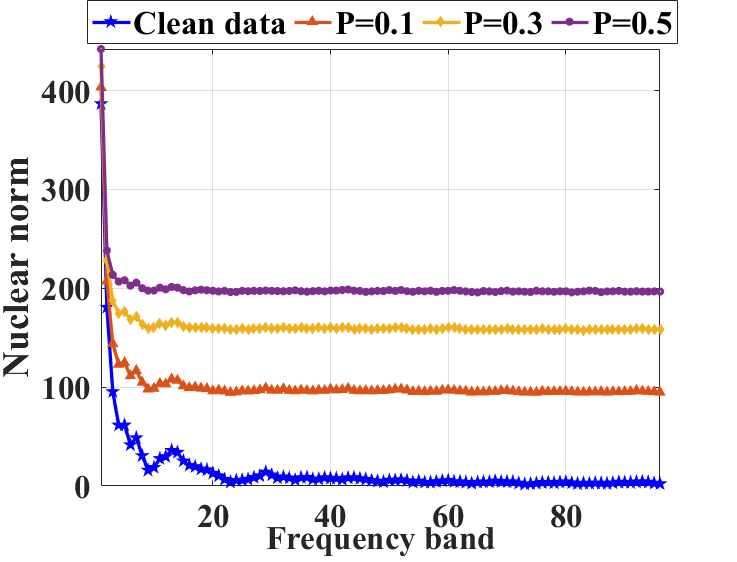}
	}
	\subfigure[  ] {
		\label{Fig511}
		\includegraphics[width=0.410\columnwidth]{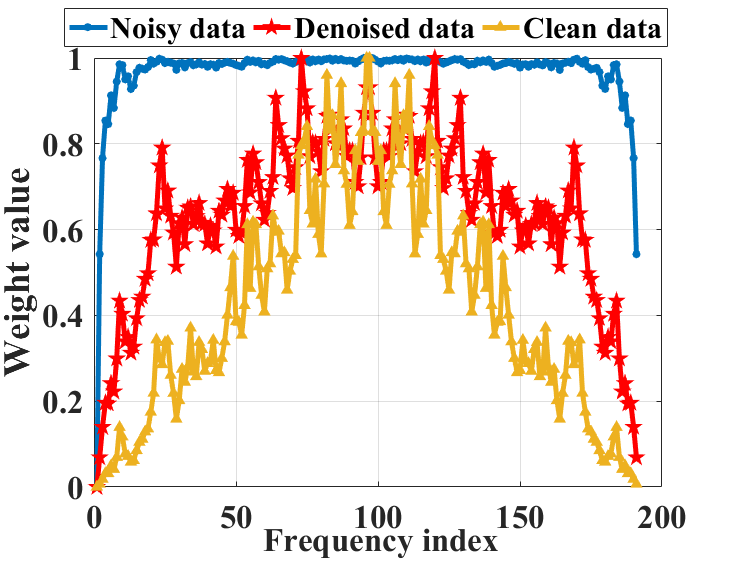}
	}
	\caption{ (a) The distribution of the nuclear norm of each frequency band of DC Mall data under different sparse noise intensities. (b) Frequency weights obtained from different data.
	}
	\label{Fig5411}
\end{figure}

The above method of calculating weights is feasible when the clean HSI is known.
However, for HSI denoising tasks, we often only have noisy HSI, and clear HSI is unknown.
It can be seen from Fig. \ref{Fig511} that the weights calculated by the denoised HSI is closer to the result obtained by the clean HSI than the noisy HSI.
Further, as shown in Fig. \ref{fig:nonmfwtnnshiyitu06}(e) , as the number of iterations $n$ continues to increase, the weights calculated by the denoised HSI will get closer to the result of clean HSI.
Based on this, we propose an iterative calculation method to update $w_{k}^{p}$.
Assuming that the $n$th iteration gets $\boldsymbol{\bar{\textbf{X}}}_{p}^{\left( k , n\right)}$, then the corresponding  $w_{k,(n+1)}^{p}$ of the $n+1$th iteration is
\begin{equation}
\label{Weights02}
w_{k,(n+1)}^{p} =C_1h_{k,n}^{p} +C_2 ,
\end{equation}
where $h_{k,n}^{p} =1/(\log \left( \lVert \boldsymbol{\bar{\textbf{X}}}_{p}^{\left( k,n \right)} \rVert _{F}^{2} \right) +\delta )$.
$C_1$ is the scaling factor after frequency normalization; $C_2$ is a constant.
Affected by noise, there is less useful information from the denoised HSI in the initial iteration.
If it is not corrected, the denoised result will develop in a bad direction.
Since TNN is a robust model, inspired by this, we introduce a constant $C_2$ to improve our weights.
When $C_1$=0, it degenerates into TNN.
Another reason we introduced $C_2$ is that noise exists in all frequency bands.
The previous conclusion is that there is less noise information in the low-frequency part instead of none.
$C_2$ can achieve a slight shrinkage of all frequency bands.
Presetting the value of $C_2$ is a feasible solution.

\section{HSI Denoising via MFWTNN and NonMFWTNN Minimization}
\subsection{Proposed model}
MFWTNN uses frequency components, modal information, and NonMFWTN uses the physical meaning of singular value distribution on the basis of MFWTNN.
They can provide a better approximation to the tensor rank.
Then we use MFWTNN and NonMFWTNN to replace the regularization term $\texttt{Rank}$ in \eqref{eq_2} and propose the HSI denoising model as follows:
\begin{equation}
\label{MFWTNN_main01}
\begin{aligned}
\arg\min_{\mathcal{X}, \mathcal{S}, \mathcal{N}} &\sum_{p=1}^{3}\alpha_p \|\mathcal{Z}_{p}\|_{FW*}+ \lambda\|\mathcal{S}\|_1 + \tau\|\mathcal{N}\|_F^2, \\
&s.t.  \mathcal{Y} = \mathcal{X} + \mathcal{S} + \mathcal{N}.
\end{aligned}
\end{equation}
\begin{equation}
\label{MFWTNN_main0002}
\begin{aligned}
\arg\min_{\mathcal{X}, \mathcal{S}, \mathcal{N}} &\sum_{p=1}^{3}\alpha_p \|\mathcal{Z}_{p}\|_{FW*,Log}+ \lambda\|\mathcal{S}\|_1 + \tau\|\mathcal{N}\|_F^2, \\
&s.t.  \mathcal{Y} = \mathcal{X} + \mathcal{S} + \mathcal{N}.
\end{aligned}
\end{equation}

Introducing auxiliary variables, model \eqref{MFWTNN_main01} and \eqref{MFWTNN_main0002} are equivalent to
\begin{equation}
\label{MFWTNN_main02}
\begin{array}{rl}
\displaystyle \arg \min_{\mathcal{X},\mathcal{S},\mathcal{N}} \sum_{p=1}^{3}\alpha_p \|\mathcal{Z}_{p}\|_{FW*}+\lambda \lVert \mathcal{S} \rVert _1+\tau \lVert \mathcal{N} \rVert _{F}^{2},\\
\displaystyle  s.t.\ \mathcal{Y}=\mathcal{X}+\mathcal{S}+\mathcal{N},\ \mathcal{Z}_p=\mathcal{X}_p,\ p=1,2,3 .
\end{array}
\end{equation}
\begin{equation}
\label{MFWTNN_main03}
\begin{array}{rl}
\displaystyle \arg \min_{\mathcal{X},\mathcal{S},\mathcal{N}} \sum_{p=1}^{3}\alpha_p \|\mathcal{Z}_{p}\|_{FW*,Log}+\lambda \lVert \mathcal{S} \rVert _1+\tau \lVert \mathcal{N} \rVert _{F}^{2},\\
\displaystyle  s.t.\ \mathcal{Y}=\mathcal{X}+\mathcal{S}+\mathcal{N},\ \mathcal{Z}_p=\mathcal{X}_p,\ p=1,2,3 .
\end{array}
\end{equation}

By augmented Lagrangian multiplier method, the Lagrangian function of model \eqref{MFWTNN_main02} can be written as
\begin{equation}
\begin{aligned}
&L_{\mu _p,\beta}\left( \mathcal{X},\mathcal{Z}_p,\mathcal{N},\mathcal{S},\Gamma _p,\Lambda \right) = \lambda \lVert \mathcal{S} \rVert _1+\tau \lVert \mathcal{N} \rVert _{F}^{2} \\
&+ <\mathcal{Y}-\left( \mathcal{X}+\mathcal{S}+\mathcal{N} \right) ,\Lambda > +\frac{\beta}{2}\lVert \mathcal{Y}-\left( \mathcal{X}+\mathcal{S}+\mathcal{N} \right) \rVert _{F}^{2} \\
&+\sum_{p=1}^3{\left\{ \alpha _p\lVert \mathcal{X}_p \rVert _{FW*}+<\mathcal{X}_p-\mathcal{Z}_p,\Gamma _p>+\frac{\mu _p}{2}\lVert \mathcal{X}_p-\mathcal{Z}_p \rVert _{F}^{2} \right\}},\\
\end{aligned}
\end{equation}
where $\Lambda$ and $\Gamma_p$ are the Lagrangian multipliers;
$\beta$ and $\mu _p$ are the Lagrange penalty parameters.
Compared with MFWTNN, the model based on NonMFWTNN only differs in the steps of solving $\mathcal{Z}_p$. See the solution of $\mathcal{Z}_p$ for details.
\subsection{Algorithmic optimization}
ADMM is an effective framework to solve these type of minimization problems \cite{ADMM}.
When other variables are fixed at the $n$-th iteration, each variable in the Lagrangian function can be updated by solving its corresponding subproblem respectively at the $(n+1)$-th iteration.

For $\mathcal{Z}_p$, $p=1,2,3$, their corresponding subproblems can be written as
\begin{equation}
\label{update_Zp01}
\arg\min_{\mathcal{Z}_{p}} \alpha _p\lVert \mathcal{Z}_p \rVert _{FW*}+\frac{\mu _p}{2}\left\| \mathcal{Z}_p-\left( \mathcal{X}_{p}^{n}+\frac{\Gamma _{p}^{n}}{\mu _p} \right) \right\| _{F}^{2}.
\end{equation}
The closed-form solution of \eqref{update_Zp01} obtained from theorem 1 of \cite{FWTNN}  are as follows:
\begin{equation}
\begin{array}{rl}
\mathcal{Z}_{p}^{n+1}=\mathrm{FTSVT}^{w\left( \mathcal{X}_p^n \right) ,\frac{\alpha _p}{\mu _p}}\left( \mathcal{X}_{p}^{n}+\frac{\Gamma _{p}^{n}}{\mu _p} \right) . \\
\end{array}\label{updateZp01}
\end{equation}

Similarly, when $\texttt{Rank}(\mathcal{X})$ is replaced by $\sum_{p=1}^3{ \alpha _p\lVert \mathcal{X}_p \rVert _{FW*,Log}}$ in \eqref{MFWTNN_main03}, the sub-problem of $\mathcal{Z}_p$ is rewritten as
\begin{equation}
\label{update_Zp02}
\arg\min_{\mathcal{Z}_{p}} \alpha _p\lVert \mathcal{Z}_p \rVert _{FW*,Log}+\frac{\mu _p}{2}\left\| \mathcal{Z}_p-\left( \mathcal{X}_{p}^{n}+\frac{\Gamma _{p}^{n}}{\mu _p} \right) \right\| _{F}^{2}.
\end{equation}
The closed-form solution of \eqref{update_Zp02} obtained from theorem \ref{thm2} are as follows:
\begin{equation}
\begin{array}{rl}
\mathcal{Z}_{p}^{n+1}=\mathcal{D}\mathcal{W}^{w\left( \mathcal{X}_p^n \right),\varepsilon  ,\frac{\alpha _p}{\mu _p}}\left( \mathcal{X}_{p}^{n}+\frac{\Gamma _{p}^{n}}{\mu _p} \right) . \\
\end{array}\label{updateZp02}
\end{equation}

For  $ \mathcal{X} $, its corresponding subproblem can be reformulated as
\begin{equation}
\begin{aligned}
\mathcal{X}^{n+1}=
&\arg\min_{\mathcal{X}} \sum_{p=1}^3{\frac{\mu _p}{2}}\lVert \mathcal{X}-\mathcal{Z}_{p}^{n+1}+\frac{\Gamma _{p}^{n}}{\mu _p} \rVert _{F}^{2}\\
&+\frac{\beta}{2}\lVert \mathcal{Y}-\left( \mathcal{X}+\mathcal{S}^n+\mathcal{N}^n \right) +\frac{\Lambda ^n}{\beta} \rVert _{F}^{2}.
\end{aligned}
\end{equation}
$ \mathcal{X} $ can be updated as follows:
\begin{equation}
\begin{array}{rl}
\mathcal{X}^{n+1}=\frac{\sum_{p=1}^3{\mu _p}\left( \mathcal{Z}_{p}^{n+1}-\frac{\Gamma _{p}^{n}}{\mu _p} \right) +\beta \left( \mathcal{Y}-\mathcal{S}^n-\mathcal{N}^n+\frac{\Lambda ^n}{\beta} \right)}{1+\beta}. \\
\end{array}\label{update_X}
\end{equation}

For  $ \mathcal{S} $, its corresponding subproblem can be reformulated as
\begin{equation}
\label{solve_S}
\arg\min_{\mathcal{S}}\lambda \lVert \mathcal{S} \rVert _1+\frac{\beta}{2}\lVert \mathcal{Y}-\left( \mathcal{X}^{n+1}+\mathcal{S}+\mathcal{N}^{n} \right) +\frac{\Lambda ^n}{\beta} \rVert _{F}^{2}.
\end{equation}
It can be solved by the soft-thresholding operator \cite{soft_threshold} as:
\begin{equation}
\begin{array}{rl}
\mathcal{S}^{n+1}=\texttt{shrink}\left( \mathcal{Y}-\mathcal{X}^{n+1}-\mathcal{N}^{n+1}+\frac{\Lambda ^n}{\beta},\frac{\lambda}{\beta} \right).
\end{array}\label{update_S}
\end{equation}

For  $ \mathcal{N} $, its corresponding subproblem can be reformulated as
\begin{equation}
\arg\min_{\mathcal{N}} \tau \|\mathcal{N}\|_{F}^{2}+\frac{\beta}{2}\|\mathcal{Y}-\left( \mathcal{X}^{n+1}+\mathcal{S}^{n+1}+\mathcal{N} \right) +\frac{\Lambda ^n}{\beta}\|_{F}^{2}.
\end{equation}
$ \mathcal{N} $ can be updated as follows :
\begin{equation}
\begin{array}{rl}
\mathcal{N}^{n+1}=\frac{\beta \left( \mathcal{Y}-\mathcal{X}^{n+1}-\mathcal{S}^n+\frac{\Lambda ^n}{\beta} \right)}{2\tau +\beta}.
\end{array}\label{update_N}
\end{equation}

For multipliers $\Gamma_{p}$ and $\Lambda$, they can be updated as follows:
\begin{equation}
\left\{ \begin{array}{l}
\Gamma _{p}^{n+1}=\Gamma _{p}^{n}+\mu _p\left( \mathcal{Z}_p^{n+1}-\mathcal{X}^{n+1} \right), p=1,2,3\\
\Lambda ^{n+1}=\Lambda ^n+\beta \left( \mathcal{Y}-\mathcal{X}^{n+1}-\mathcal{S}^{n+1}-\mathcal{N}^{n+1} \right). \\
\end{array} \right. \label{update_multiplier}
\end{equation}

\subsection{Time Complexity Analysis}
The HSI denoising models based on MFWTNN and NonMFWTNN can be solved by Algorithm 1.
Further, we discuss the time complexity of this algorithm.
Computing $\mathcal{Z}_p$, $p=1,2,3$, in both MFWTNN-based solver and NonMFWTNN based solver, has a complexity  of	
$O(n_1n_2n_3\log ( n_1n_2n_3 )$   $ +\sum_{i=1}^3{\max ( n_i,n_{i+1} )}\min ^2( n_i,n_{i+1} ) n_{i+2})$, where $n_4=n_1, n_5=n_2$.
Calculating $w_k$ has a complexity of $O(\sum_{i=1}^3{\min \left( n_i,n_{i+1} \right) n_{i+2}})$.
The cost of computing $ \mathcal{X} $, $ \mathcal{S} $, and $ \mathcal{N} $ are all $O(n_1n_2n_3)$.
Then, by calculating the complexity of the above variables, the total complexity of the proposed algorithm can be obtained as
$O(n_1n_2n_3\log \left( n_1n_2n_3 \right) +\sum_{i=1}^3{\max \left( n_i,n_{i+1} \right) \min ^2( n_i,n_{i+1} ) n_{i+2}}$ $+\sum_{i=1}^3{\min \left( n_i,n_{i+1} \right)}n_{i+2}+3n_1n_2n_3)$.
\begin{algorithm}[htbp]
	\caption{HSI denoising via the MFWTNN and NonMFWTNN minimization}
	{\bf Input:} The observed tensor $\mathcal{Y}$; weight parameters $c_1$, $c_2$, $\varepsilon $;  regularization parameters $\lambda$, $\tau$; and  stopping criterion $\epsilon$. \\
	{\bf Output:} Denoised image $\mathcal{X}$. \\
	\hspace*{0.02in} 1: Initialize: $\mathcal{Y}$=$\mathcal{X}$=$\mathcal{S}$=$\mathcal{N}$=$\mathcal{Z}_p$; $\Gamma _p=\Lambda =0$;  $\mu_p$=$\beta $=$10^{-3}$; \\
	\hspace*{0.07in}$p=1,2,3$; $\mu_{max}=10^{10}$;  $\rho=1.2$ and $n=1$. \\
	\hspace*{0.02in} 2: Repeat until convergence:\\
	\hspace*{0.02in} 3. Update $\mathcal{X}, \mathcal{S}, \mathcal{N}, \mathcal{Z}_p, \Lambda$, $\beta$, $\mu_p$, $w_k$, $\Gamma_p$ via\\
	\hspace*{0.2in} step 1: Update $ \mathcal{Z}_p $ by \eqref{updateZp01}  or \eqref{updateZp02}  \\
	\hspace*{0.2in} step 2: Update $ \mathcal{X} $ by \eqref{update_X}\\
	\hspace*{0.2in} step 3: Update $ \mathcal{S} $ by \eqref{update_S}\\
	\hspace*{0.2in} step 4: Update $ \mathcal{N} $ by \eqref{update_N}\\
	\hspace*{0.2in} step 5: Update $ \Gamma _{p}, \Lambda $ by \eqref{update_multiplier}\\	
	\hspace*{0.2in} step 6: Update $\mu_p=\rho\mu_p$, $\beta=\rho\beta$, $w_k$ by \eqref{Weights02}\\
	\hspace*{0.02in} 4: Check the convergence condition.
\end{algorithm}

\section{Experiment Results and Discussion}

In this section, we conduct experiments on simulated and real-world HSIs to substantiate the effectiveness of the proposed MFWTNN and NonMFWTNN model for HSI denoising.
For evaluating the effectiveness of the models, our models are compared with six state-of-the-art HSI denoised methods, i.e., LRTA \cite{LRTA}, BM4D \cite{BM4D}, LRMR \cite{LRMR}, LRTDTV \cite{LRTDTV}, 3DTNN \cite{3DTNN} and 3DLogTNN \cite{3DTNN}.
Since LRTA and BM4D are only suitable for removing Gaussian noise and our experiments are to remove hybrid noise, for fair comparison, we first use the RPCA \cite{TRPCA} to remove sparse noise, and then use these models to obtain denoised results.
Before the experiments, all pixels of models in the HSI band are normalized to [0,1].
In all experiments, the parameters in these comparative methods are manually adjusted according to the suggestions in relevant papers.
In Section \ref{Discussion_Parameter},  the parameters of the proposed models are discussed in detail.
We use visual comparison and quantitative comparison to comprehensively evaluate the performance of different denoising methods.
In quantitative comparisons, we use five quantitative picture quality indices (PQIs), i.e., PSNR \cite{PSNR}, SSIM \cite{SSIM}, FSIM \cite{FSIM}, ERGAS \cite{EGRAS} and SAM \cite{SAM}, to assess the denoised results.
Better denoised results correspond to higher values in MPSNR, MSSIM and MFSIM, and lower values in MERGAS and MSAM.

\subsection{Simulated HSI Data Experiments}
We select two HSI datasets for the simulated experiments to evaluate the performance of our methods (see Fig. \ref{PaviaDCW_3Dtube}).
One is collected by the reflection optical system imaging spectrometer (ROSIS-03), which is named the Pavia City Center dataset\footnote{\url{http://www.ehu.eus/ccwintco/index.php/}}.
Its size is $1096 \times 1096$, with a total of $102$ bands.
The reason we chose a sub-block with a size of $200 \times 200 \times 80$ as a simulation dataset is that some bands are seriously polluted by noise and have lost the meaning of reference.
The other one is collected by the HSI acquisition sensor (HYDICE), which is named the Washington DC Mall dataset\footnote{\url{http://lesun.weebly.com/hyperspectral-data-set.html}}.
Its size is $1208 \times 307$, with a total of 191 bands.
We choose a sub-block with a size of $256 \times 256 \times 191$ as a simulation dataset.

The observed HSIs are usually degraded by hybrid noises, i.e. Gaussian noise, sparse noise and stripe noise.
Thus, to better simulate the degradation mechanism of HSI, we add various intensities of hybrid noises to clean HSI.
In the simulation experiments, hybrid noises with eight different intensity levels are added to the simulation dataset band by band.
Let $G$ and $P$  denote the variance of Gaussian white noise and percentage of impulse noise, respectively.
In noise cases 1-5, the same intensity noise is added to all the bands.
Specifically, in noise case 1, $G$=0.1 and $P$=0.2;
in noise case 2, $G$=0.1 and $P$=0.3;
in noise case 3, $G$=0.1 and $P$=0.4;
in noise case 4, $G$=0.15 and $P$=0.2;
in noise case 5, $G$=0.2 and $P$=0.2;
in noise cases 6-8, the noise intensities are different for different bands;
in noise case 6, $G$=0.1 and $P$ is randomly selected from 0.2 to 0.4;
in noise case 7, $G$ is randomly selected from 0.1 to 0.3 and $P$=0.2;
in noise case 8, $G$ is randomly selected from 0.1 to 0.3 and $P$ is randomly selected from 0.1 to 0.3. In addition,  stripe noises are added to the 54th-64th bands of the Pavia City Center and the 70th–100th bands of the Washington DC Mall.
\begin{figure}[h] \centering    	
	\subfigure[ ] {
		\label{Pavia_3Dtube}
		\includegraphics[width=0.33\columnwidth]{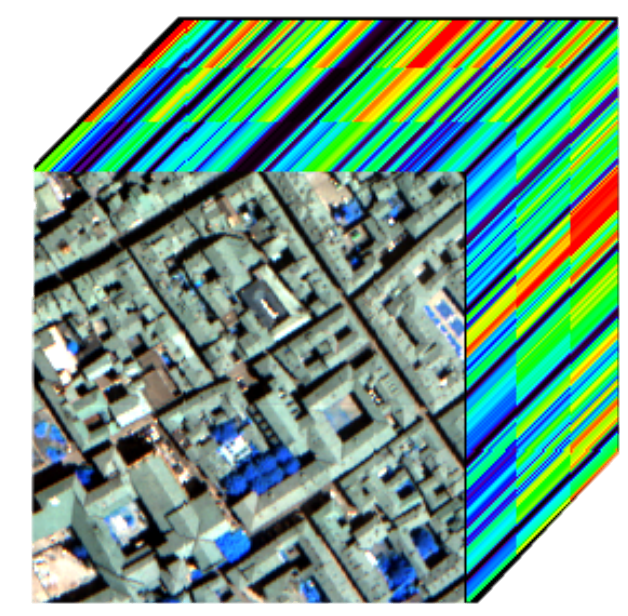}
	}
	\subfigure[  ] {
		\label{DCW_3Dtube}
		\includegraphics[width=0.33\columnwidth]{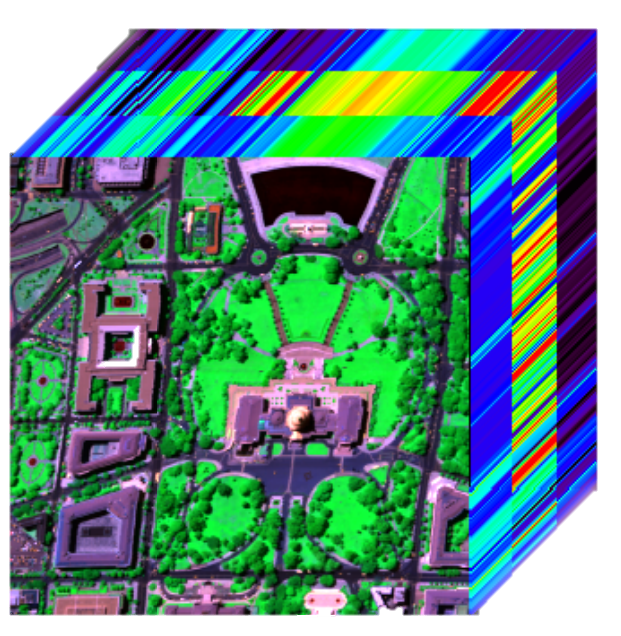}
	}
	\caption{ Datasets used in the simulated experiment. (a) Pavia City Center dataset (R:20, G:50, B:80). (b) Washington DC Mall dataset (R:50, G:80, B:150).
	}
	\label{PaviaDCW_3Dtube}
\end{figure}

1) Pavia City Center: In this subsection, 
we show the evaluation results of all denoising models including visual and quantitative quality on the Pavia City Center dataset.

\begin{table*}[htbp]
	\centering
	\caption{
		Quantitative comparison and time of eight denoising models on Pavia City Center dataset under eight noise cases.
	}
	\scalebox{0.80}{
		\begin{tabular}{cccccccccccc}
			\toprule
			Case  & Level & Index & Noise & LRTA  & BM4D  & LRMR  & LRTDTV & 3DTNN & 3DLogTNN & MFWTNN & NonMFWTNN \\
			\midrule
			\multirow{6}[2]{*}{Case 1} & \multirow{3}[1]{*}{G=0.1} & MPSNR & 11.122 & 29.440 & 29.701 & 31.259 & 32.297 & 31.696 & 33.487 & 32.470 & \textbf{34.398} \\
			&       & MSSIM & 0.105 & 0.905 & 0.920  & 0.905 & 0.914 & 0.924 & 0.942 & 0.928 & \textbf{0.945} \\
			&       & MFSIM & 0.510  & 0.947 & 0.949 & 0.946 & 0.942 & 0.954 & 0.964 & 0.955 & \textbf{0.968} \\
			& \multirow{3}[1]{*}{P=0.2} & MERGAS & 1013.518 & 123.549 & 119.608 & 99.715 & 87.316 & 94.645 & 77.865 & 89.194 & \textbf{69.458} \\
			&       & MSAM  & 45.712 & 6.805 & 5.840  & 6.824 & 4.930  & 4.844 & 4.394 & 5.656 & \textbf{4.339} \\
			&       & time/s  & -    & 14.013 & 115.944 & 116.111 & 128.850 & 56.638 & 86.658 & 66.917 & 96.791 \\
			\midrule
			\multirow{6}[2]{*}{Case 2} & \multirow{3}[1]{*}{G=0.1} & MPSNR & 9.541 & 28.484 & 28.864 & 30.173 & 31.188 & 30.657 & 32.823 & 31.600  & \textbf{33.420} \\
			&       & MSSIM & 0.065 & 0.887 & 0.910  & 0.881 & 0.897 & 0.903 & 0.933 & 0.915 & \textbf{0.935} \\
			&       & MFSIM & 0.452 & 0.938 & 0.942 & 0.935 & 0.931 & 0.941 & 0.958 & 0.947 & \textbf{0.961} \\
			& \multirow{3}[1]{*}{P=0.3} & MERGAS & 1216.767 & 138.09 & 132.124 & 112.427 & 99.239 & 106.370 & 83.831 & 97.753 & \textbf{77.512} \\
			&       & MSAM  & 47.676 & 7.204 & 6.236 & 7.250  & 5.349 & 5.619 & 4.570  & 6.065 & \textbf{4.528} \\
			&       & time/s  & -    & 14.006 & 116.229 & 109.304 & 127.285 & 59.130 & 89.352 & 67.251 & 98.926 \\
			\midrule
			\multirow{6}[2]{*}{Case 3} & \multirow{3}[1]{*}{G=0.1} & MPSNR & 8.384 & 27.265 & 27.745 & 28.937 & 29.741 & 27.696 & 31.339 & 30.570 & \textbf{31.879} \\
			&       & MSSIM & 0.043 & 0.859 & 0.892 & 0.848 & 0.873 & 0.790  & \textbf{0.915} & 0.898 & 0.913 \\
			&       & MFSIM & 0.411 & 0.926 & 0.933 & 0.920  & 0.917 & 0.892 & 0.945 & 0.936 & \textbf{0.948} \\
			& \multirow{3}[1]{*}{P=0.4} & MERGAS & 1390.482 & 158.869 & 150.540 & 128.938 & 117.629 & 148.361 & 99.511 & 109.129 & \textbf{92.185} \\
			&       & MSAM  & 48.58 & 7.629 & 6.745 & 7.740  & 6.007 & 9.875 & 5.042 & 6.506 & \textbf{4.925} \\
			&       & time/s  & -     & 14.197 & 111.924 & 113.685 & 132.354 & 66.067 & 94.827 & 72.480 & 107.781 \\
			\midrule
			\multirow{6}[2]{*}{Case 4} & \multirow{3}[1]{*}{G=0.15} & MPSNR & 10.716 & 27.026 & 27.420 & 29.013 & 30.111 & 29.089 & 31.133 & 30.451 & \textbf{32.021} \\
			&       & MSSIM & 0.088 & 0.848 & 0.879 & 0.849 & 0.867 & 0.868 & 0.905 & 0.891 & \textbf{0.913} \\
			&       & MFSIM & 0.483 & 0.923 & 0.925 & 0.921 & 0.913 & 0.920 & 0.942 & 0.933 & \textbf{0.948} \\
			& \multirow{3}[1]{*}{P=0.2} & MERGAS & 1060.262 & 161.397 & 154.021 & 127.846 & 112.153 & 126.82 & 101.226 & 110.568 & \textbf{91.095} \\
			&       & MSAM  & 46.411 & 7.825 & 6.671 & 7.690  & 5.848 & 5.895 & 5.227 & 6.562 & \textbf{5.050} \\
			&       & time/s  & -     & 14.683 & 112.103 & 112.627 & 131.750 & 62.625 & 94.155 & 72.962 & 104.037 \\
			\midrule
			\multirow{6}[2]{*}{Case 5} & \multirow{3}[1]{*}{G=0.2} & MPSNR & 10.265 & 25.163 & 25.619 & 27.321 & 28.605 & 27.009 & 29.134 & 29.088 & \textbf{30.520} \\
			&       & MSSIM & 0.074 & 0.789 & 0.835 & 0.791 & 0.821 & 0.799 & 0.860  & 0.856 & \textbf{0.885} \\
			&       & MFSIM & 0.461 & 0.901 & 0.900 & 0.898 & 0.886 & 0.88  & 0.917 & 0.913 & \textbf{0.931} \\
			& \multirow{3}[1]{*}{P=0.2} & MERGAS & 1115.059 & 198.966 & 188.524 & 154.623 & 133.338 & 160.649 & 126.652 & 128.254 & \textbf{108.312} \\
			&       & MSAM  & 46.978 & 8.379 & 7.17  & 8.385 & 6.685 & 6.923 & 5.981 & 7.143 & \textbf{5.580} \\
			&       & time/s  & -     & 14.355 & 115.494 & 117.568 & 126.082 & 58.988 & 88.231 & 67.559 & 97.910 \\
			\midrule
			\multirow{6}[2]{*}{Case 6} & \multirow{3}[1]{*}{G=0.1} & MPSNR & 9.515 & 28.374 & 28.761 & 30.090 & 31.062 & 30.400  & 32.569 & 31.512 & \textbf{33.405} \\
			&       & MSSIM & 0.066 & 0.885 & 0.908 & 0.879 & 0.895 & 0.894 & 0.927 & 0.913 & \textbf{0.937} \\
			&       & MFSIM & 0.452 & 0.937 & 0.942 & 0.933 & 0.931 & 0.938 & 0.956 & 0.946 & \textbf{0.960} \\
			& \multirow{3}[1]{*}{P=(0.2,0.4)} & MERGAS & 1235.513 & 139.857 & 133.652 & 113.657 & 100.96 & 109.752 & 86.176 & 98.851 & \textbf{78.252} \\
			&       & MSAM  & 47.824 & 7.278 & 6.281 & 7.294 & 5.501 & 6.078 & 4.991 & 6.148 & \textbf{4.589} \\
			&       & time/s  & -   & 14.490 & 115.407 & 116.126 & 128.118 & 59.013 & 87.868 & 67.084 & 97.160 \\
			\midrule
			\multirow{6}[2]{*}{Case 7} & \multirow{3}[1]{*}{G=(0.1,0.3)} & MPSNR & 10.193 & 25.281 & 25.589 & 27.437 & 28.773 & 26.772 & 29.211 & 29.218 & \textbf{30.800} \\
			&       & MSSIM & 0.074 & 0.803 & 0.833 & 0.798 & 0.826 & 0.784 & 0.868 & 0.861 & \textbf{0.893} \\
			&       & MFSIM & 0.461 & 0.906 & 0.898 & 0.900   & 0.890  & 0.875 & 0.920  & 0.916 & \textbf{0.934} \\
			& \multirow{3}[1]{*}{P=0.2} & MERGAS & 1131.564 & 198.204 & 189.307 & 153.217 & 131.728 & 167.334 & 126.702 & 126.994 & \textbf{105.541} \\
			&       & MSAM  & 47.045 & 9.057 & 7.138 & 8.518 & 6.650  & 7.733 & 6.194 & 6.885 & \textbf{5.672} \\
			&       & time/s  & -     & 14.881 & 116.565 & 116.969 & 125.336 & 58.311 & 87.113 & 67.129 & 96.770 \\
			\midrule
			\multirow{6}[2]{*}{Case 8} & \multirow{3}[1]{*}{G=(0.1,0.3)} & MPSNR & 9.052 & 24.289 & 24.653 & 26.402 & 27.922 & 25.048 & 28.711 & 28.489 & \textbf{29.332} \\
			&       & MSSIM & 0.053 & 0.770  & 0.812 & 0.759 & 0.803 & 0.673 & 0.852 & 0.839 & \textbf{0.865} \\
			&       & MFSIM & 0.425 & 0.893 & 0.887 & 0.884 & 0.876 & 0.836 & 0.913 & 0.903 & \textbf{0.916} \\
			& \multirow{2}[0]{*}{P=(0.2,0.4)} & MERGAS & 1299.108 & 222.242 & 210.596 & 172.408 & 145.080 & 207.854 & 134.585 & 136.996 & \textbf{125.036} \\
			&       & MSAM  & 48.246 & 9.538 & 7.394 & 9.122 & 7.027 & 11.026 & 7.233 & 7.689 & \textbf{7.526} \\
			& stipes & time/s  & -     & 13.861 & 116.280 & 113.436 & 125.442 & 60.487 & 88.043 & 67.376 & 97.515 \\
			\bottomrule
		\end{tabular}%
		\label{tab:PaviaTablePQI}%
	}
\end{table*}%

Table \ref{tab:PaviaTablePQI} reports quantitative comparison and CPU running time of all the compared methods under the condition of the eight noise cases on the Pavia City Center dataset.
When compared to the other methods, our proposed methods obtain the optimal PQIs among all the denoising models in most noise cases, indicating the advantage of the proposed methods in HSI denoising.
Although LRTA and LRTDTV also consider the low-rankness among different modes, they can not fully exploit the connection among the modes.
BM4D  use non-local similarity information and LRTDTV use spatial smoothing information.
These make the denoised result excessively smooth and lack texture details.
Although 3DTNN considers the correlation among different modes, it is still not accurate enough to represent the low-rankness of HSI.
3DlogTNN and MFWTNN use the physical meaning of singular values and frequency components to improve 3DTNN respectively.
But these are not the best descriptions of the low-rankness of HSI.
NonMFWTNN inherits the advantages of 3DLogTNN and MFWTNN, and improves the ability to express low-rankness of HSI in terms of frequency components and singular values.
Although the time complexity of our models is higher, our models have obtained better-denoised results.
\begin{figure*}
	\centering
	\includegraphics[width=0.90\columnwidth]{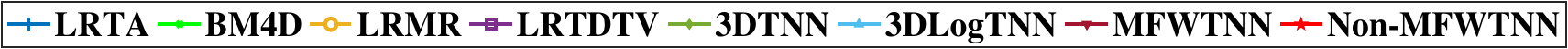}\\	
	\subfigure[Case 1]{
		\begin{minipage}[t]{0.21\textwidth}
			\includegraphics[width=1.00\columnwidth]{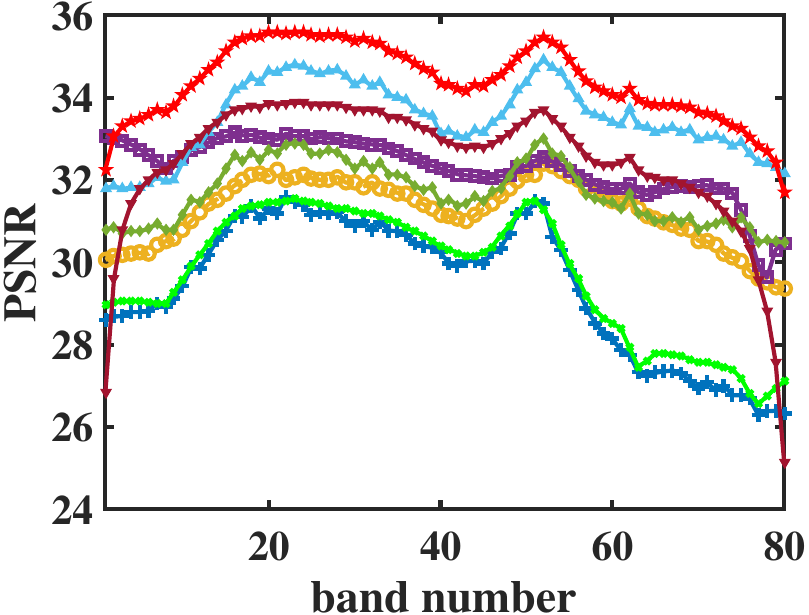}\\
			\includegraphics[width=1.00\columnwidth]{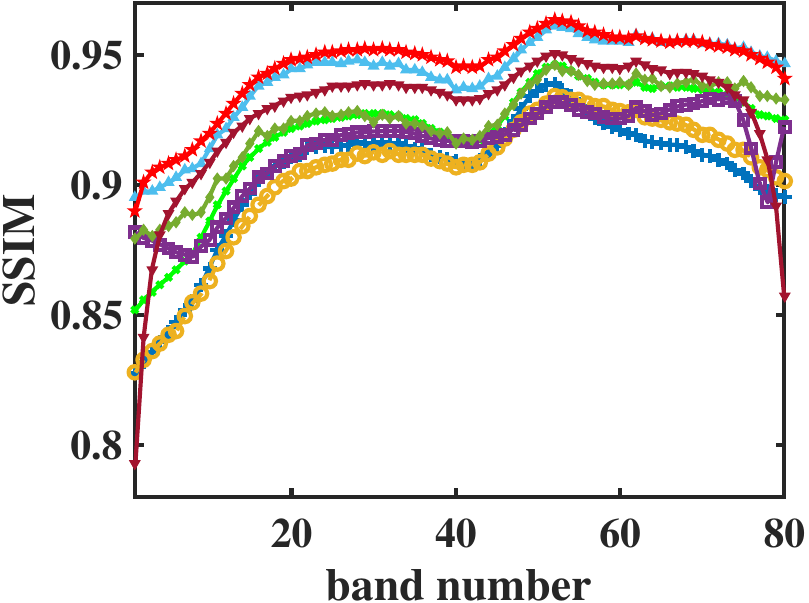}
		\end{minipage}
	}
	\subfigure[Case 2]{
		\begin{minipage}[t]{0.21\textwidth}
			\includegraphics[width=1.00\columnwidth]{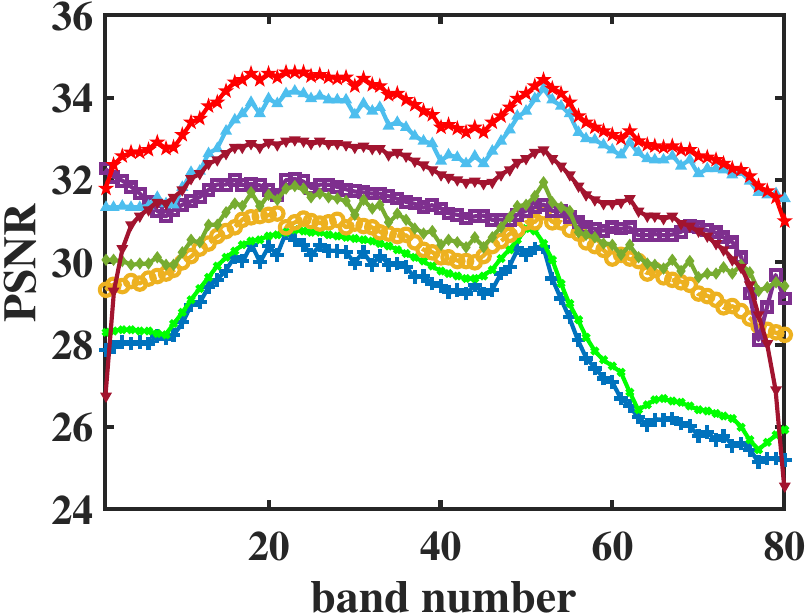}\\
			\includegraphics[width=1.00\columnwidth]{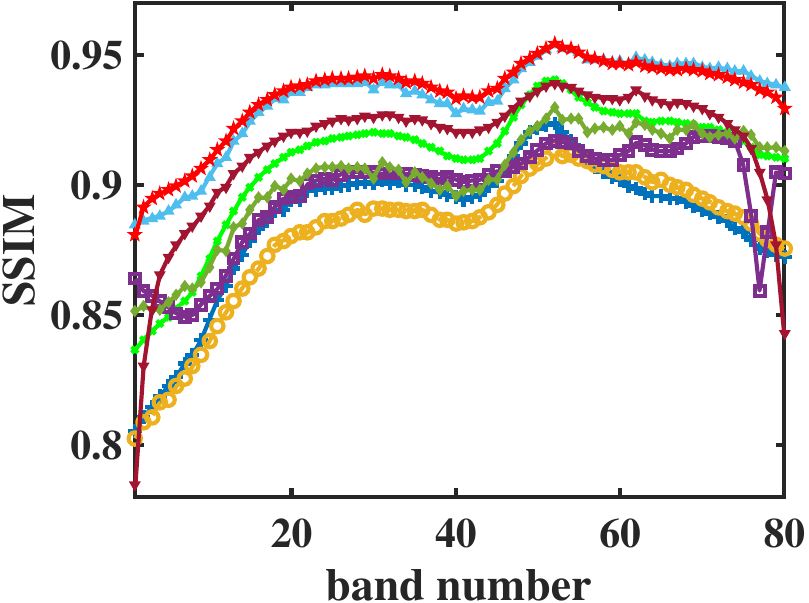}
		\end{minipage}
	}
	\subfigure[Case 3]{
		\begin{minipage}[t]{0.21\textwidth}
			\includegraphics[width=1.00\columnwidth]{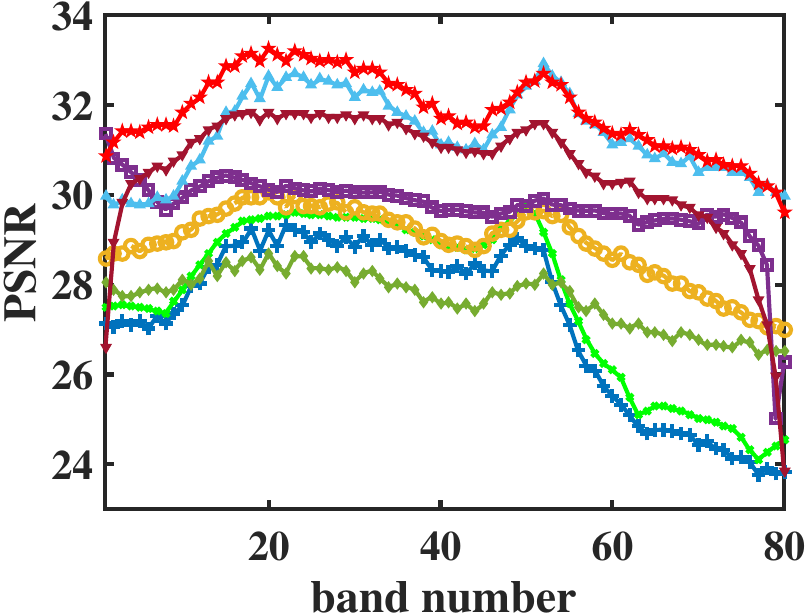}\\
			\includegraphics[width=1.00\columnwidth]{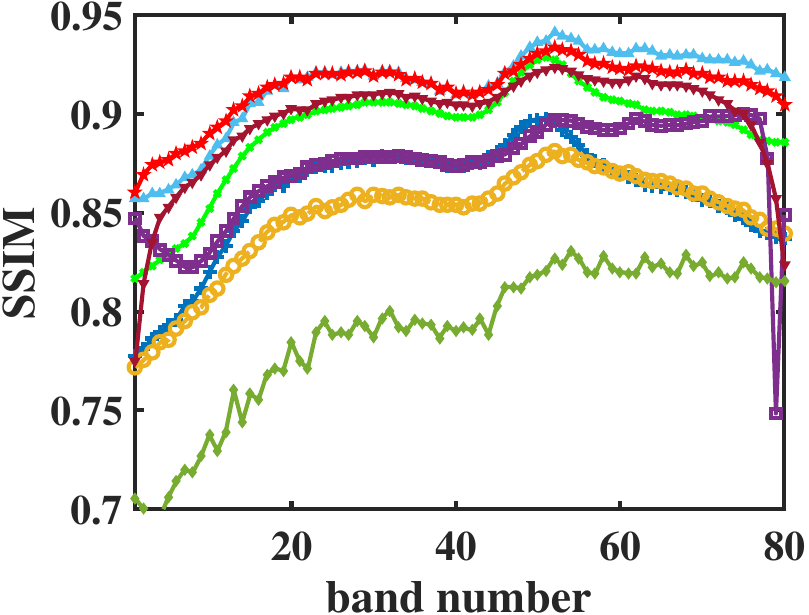}
		\end{minipage}
	}
	\subfigure[Case 4]{
		\begin{minipage}[t]{0.21\textwidth}
			\includegraphics[width=1.00\columnwidth]{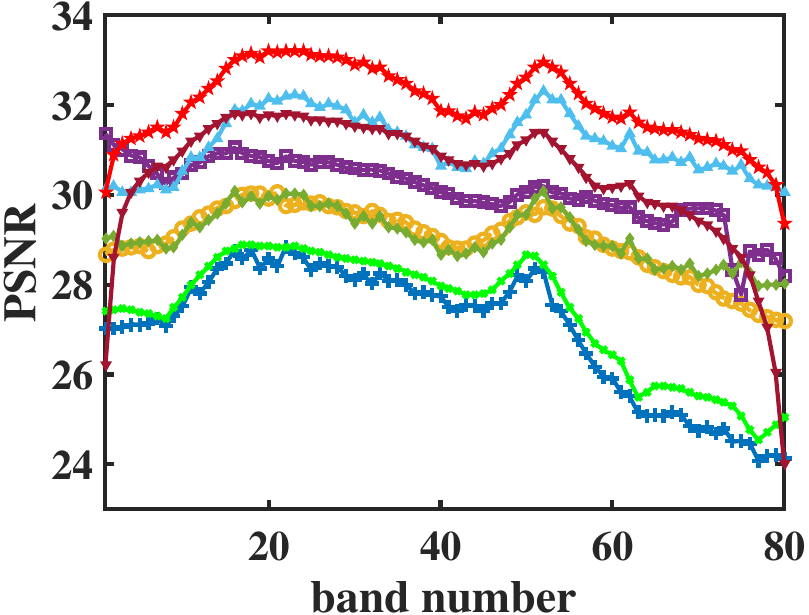}\\
			\includegraphics[width=1.00\columnwidth]{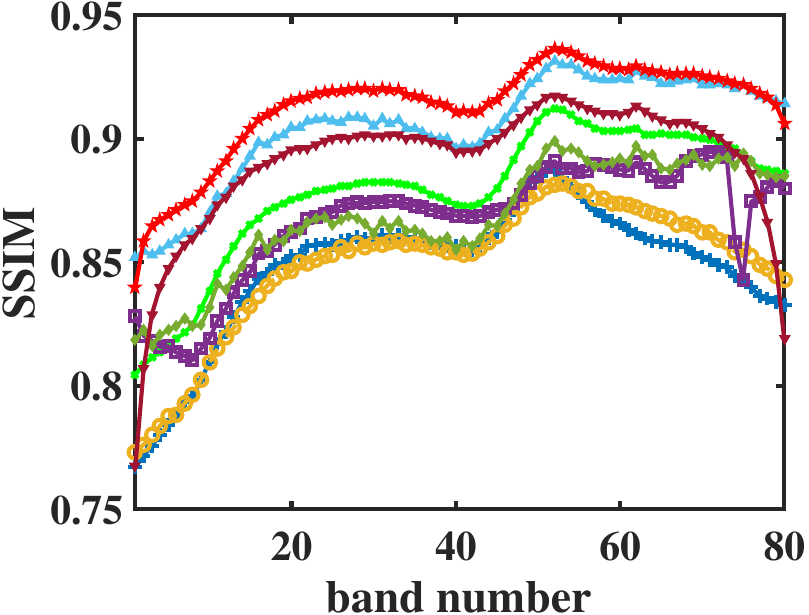}
		\end{minipage}
	}
	\subfigure[Case 5]{
		\begin{minipage}[t]{0.21\textwidth}
			\includegraphics[width=1.00\columnwidth]{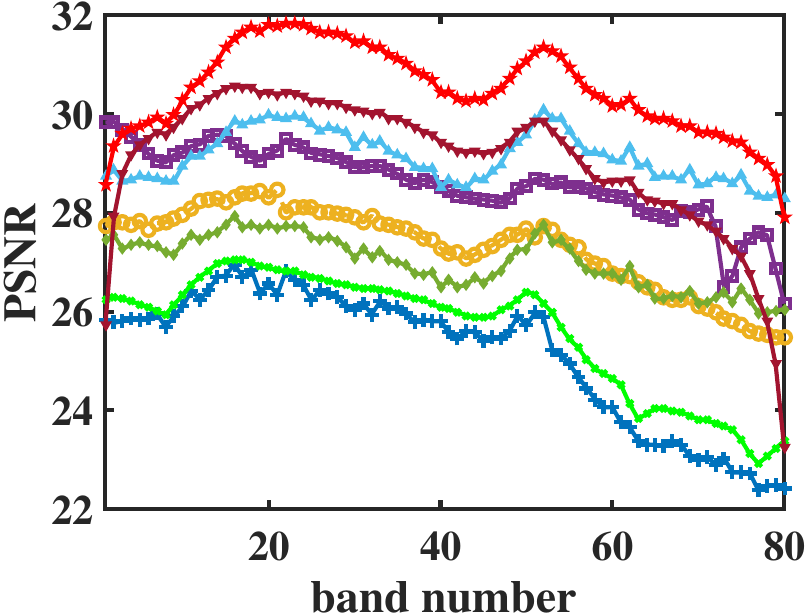}\\
			\includegraphics[width=1.00\columnwidth]{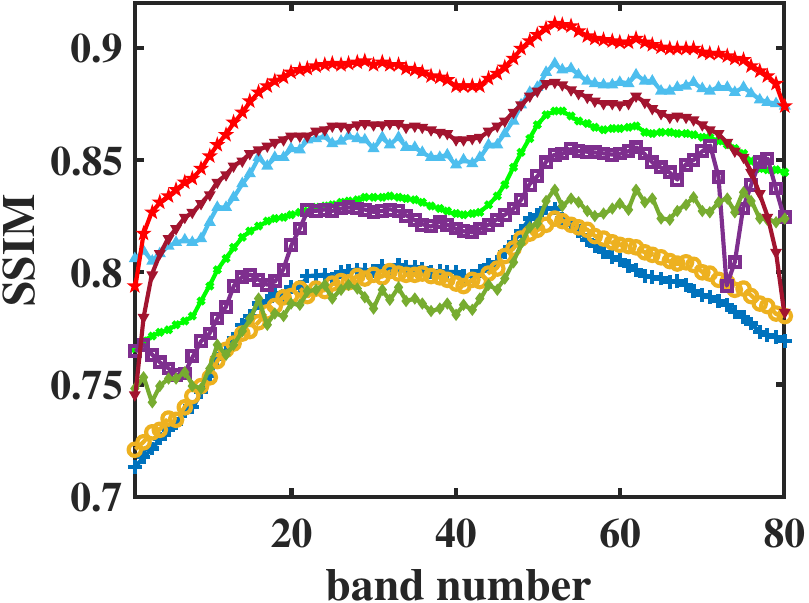}
		\end{minipage}
	}
	\subfigure[Case 6]{
		\begin{minipage}[t]{0.21\textwidth}
			\includegraphics[width=1.00\columnwidth]{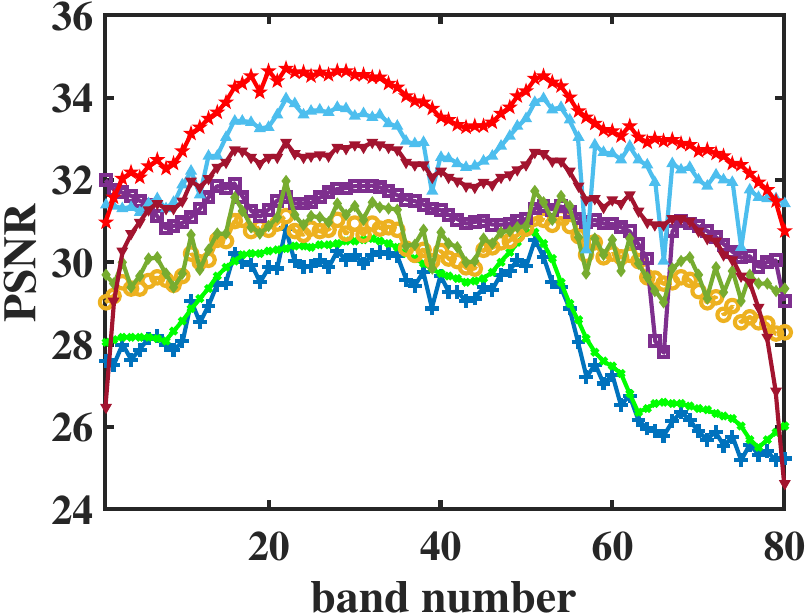}\\
			\includegraphics[width=1.00\columnwidth]{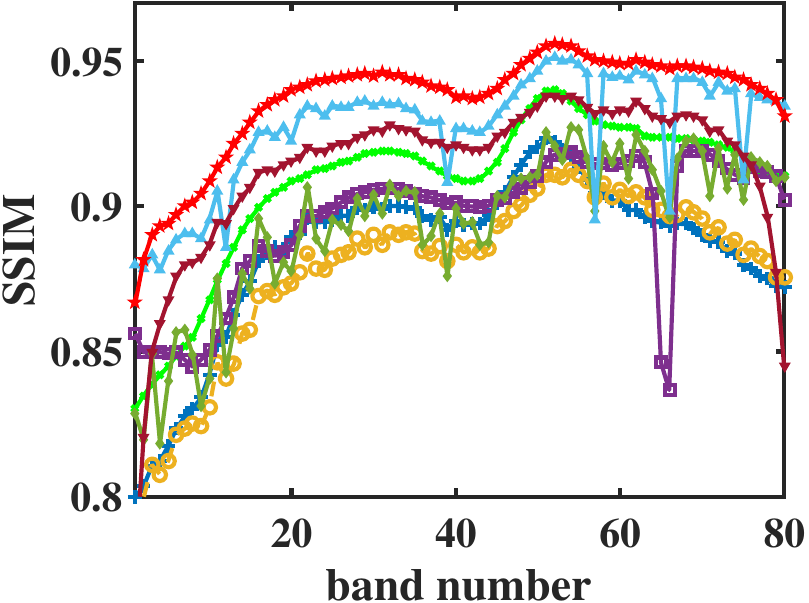}
		\end{minipage}
	}
	\subfigure[Case 7]{
		\begin{minipage}[t]{0.21\textwidth}
			\includegraphics[width=1.00\columnwidth]{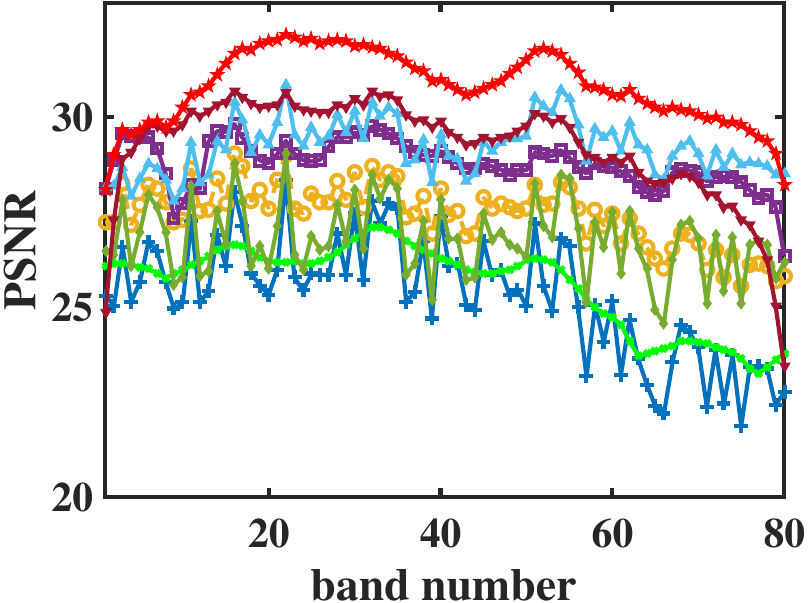}\\
			\includegraphics[width=1.00\columnwidth]{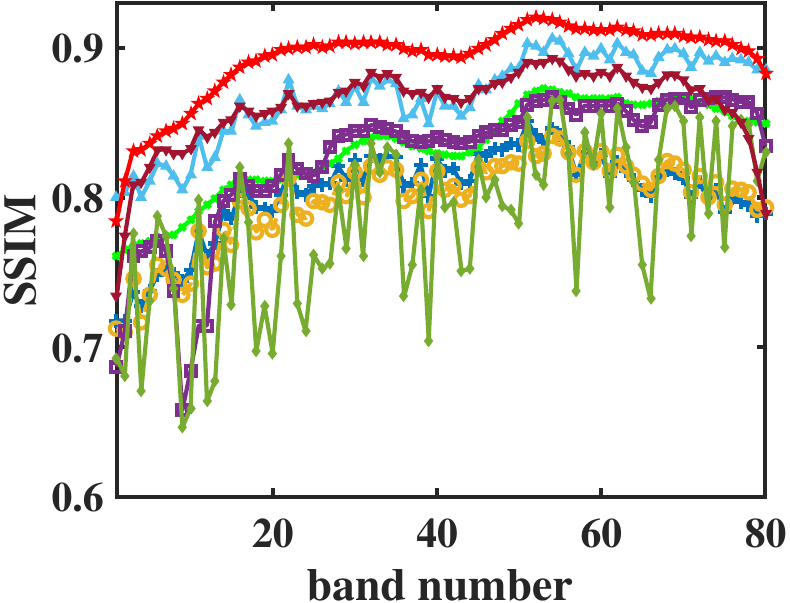}
		\end{minipage}
	}
	\subfigure[Case 8]{
		\begin{minipage}[t]{0.21\textwidth}
			\includegraphics[width=1.00\columnwidth]{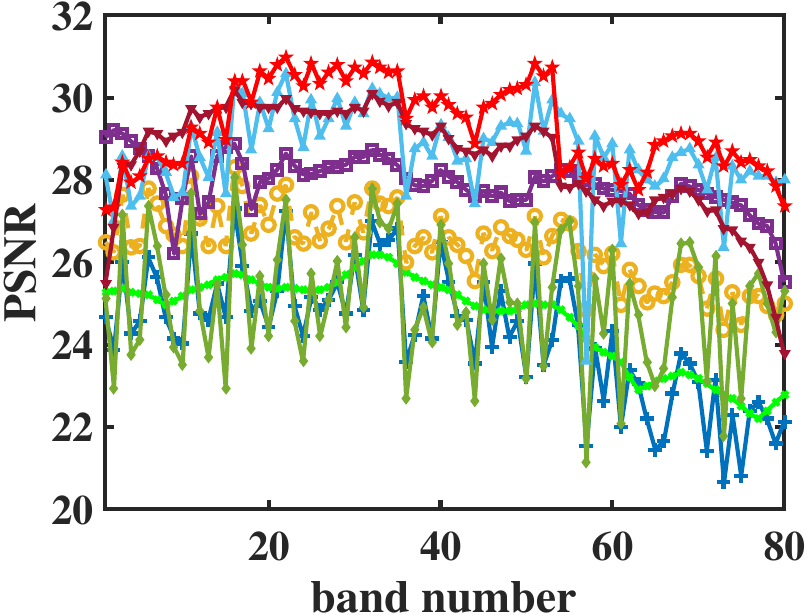}\\
			\includegraphics[width=1.00\columnwidth]{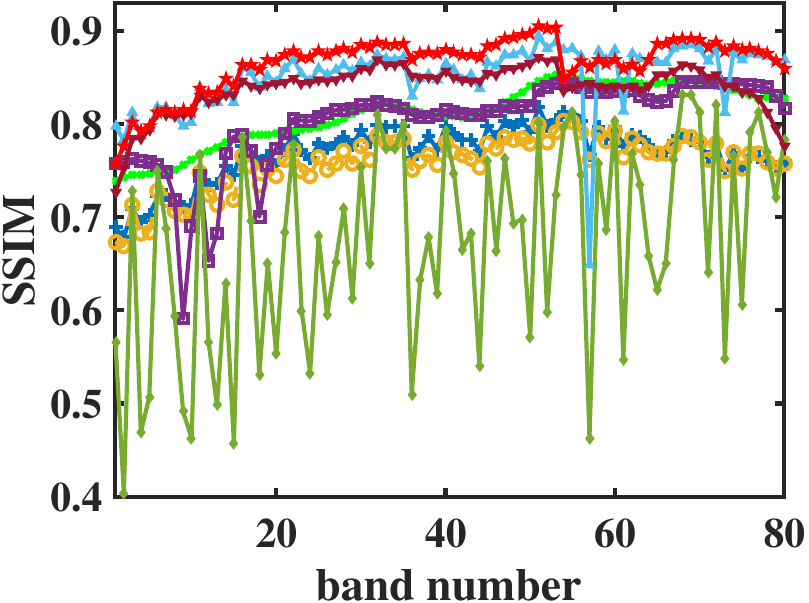}
		\end{minipage}
	}	
	\caption{ 
		The PSNR and SSIM of all denoising models for each band under the eight noise cases in Pavia City Center dataset.
	}
	\label{Pavia_imshow0302_PSNR_SSIM}
\end{figure*}

In Fig. \ref{Pavia_imshow0302_PSNR_SSIM},  we show denoised results of the Pavia City Center dataset in terms of PSNR and SSIM for each band under all noisy case experiments.
As this figure shows, our methods obtain the optimal PSNR and SSIM values in most bands.
\begin{figure*}
	\centering
	\subfigure[Original image] {
		
		\includegraphics[width=0.33\columnwidth]{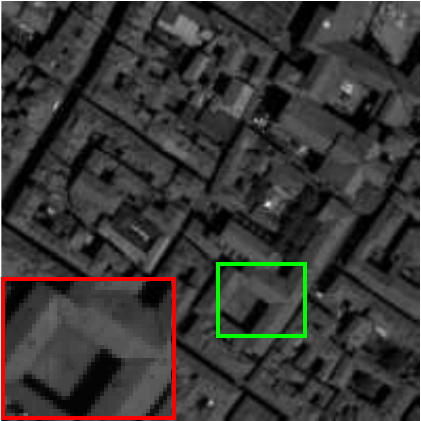}
	}
	\subfigure[Noisy image(9.60dB)] {
		
		\includegraphics[width=0.33\columnwidth]{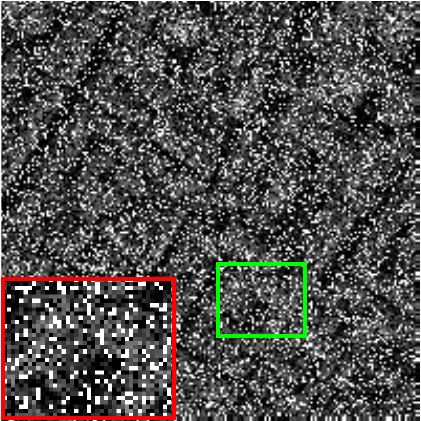}
	}
	\subfigure[LRTA(30.33dB)] {
		
		\includegraphics[width=0.33\columnwidth]{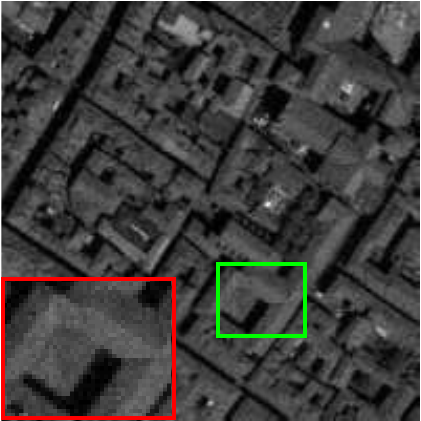}
	}
	\subfigure[BM4D(30.81dB)] {
		
		\includegraphics[width=0.33\columnwidth]{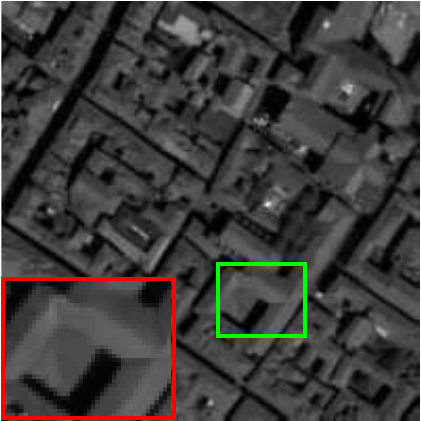}
	}
	\subfigure[LRMR(30.98dB)] {
		
		\includegraphics[width=0.33\columnwidth]{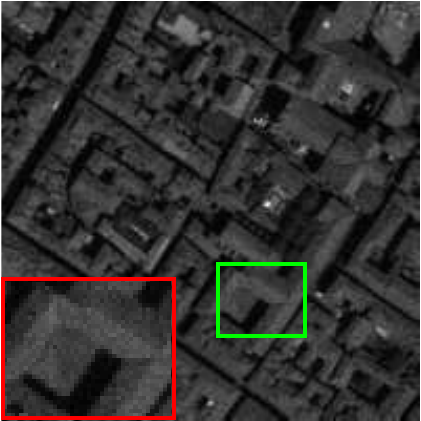}
	}
	\subfigure[LRTDTV(31.25dB)] {
		
		\includegraphics[width=0.33\columnwidth]{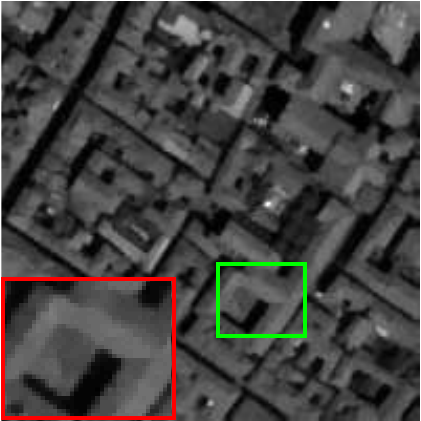}
	}
	\subfigure[3DTNN(31.40dB)] {
		
		\includegraphics[width=0.33\columnwidth]{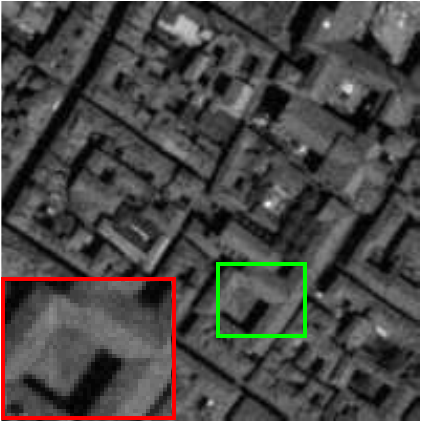}
	}
	\subfigure[3DLogTNN(33.67dB)] {
		
		\includegraphics[width=0.33\columnwidth]{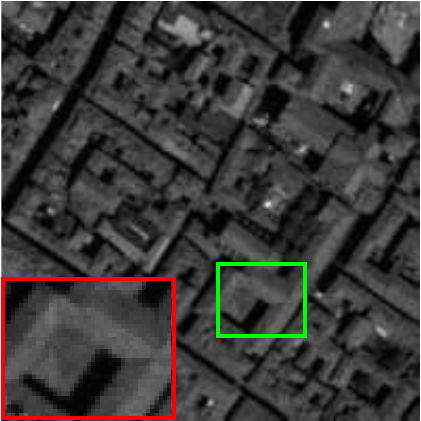}
	}
	\subfigure[MFWTNN(32.51dB)] {
		
		\includegraphics[width=0.33\columnwidth]{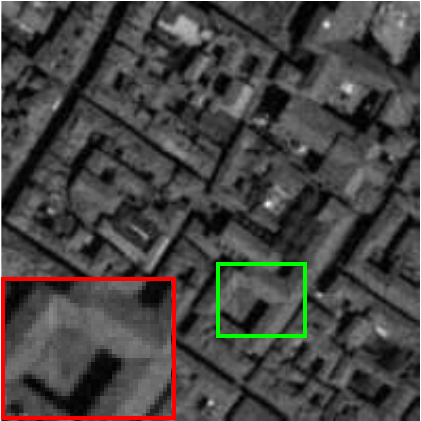}
	}
	\subfigure[\hspace*{-0.05in}NonMFWTNN(34.02dB)] {
		
		\includegraphics[width=0.33\columnwidth]{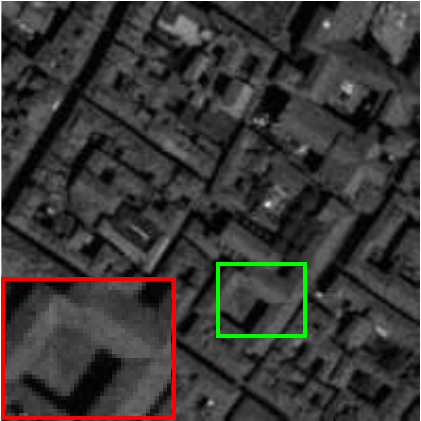}
	}
	\caption{ The 50th band of the denoised results of the Pavia City Center dataset under noise Case 2.}
	\label{Paviaimshow0103_50}
\end{figure*}
\begin{figure*}
	\centering
	\subfigure[Original image] {
		
		\includegraphics[width=0.33\columnwidth]{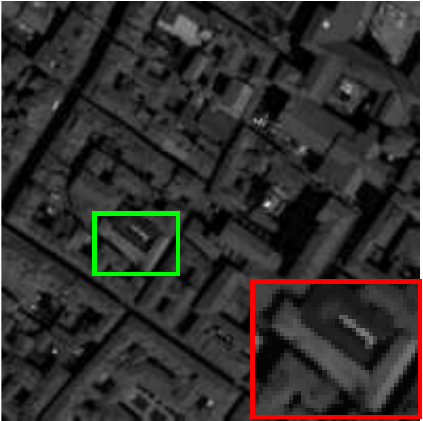}
	}
	\subfigure[Noisy image(8.32dB)] {
		
		\includegraphics[width=0.33\columnwidth]{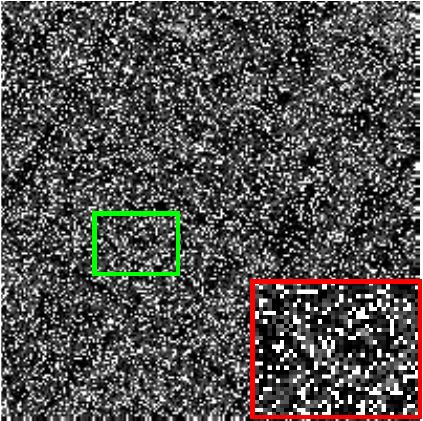}
	}
	\subfigure[LRTA(28.89dB)] {
		
		\includegraphics[width=0.33\columnwidth]{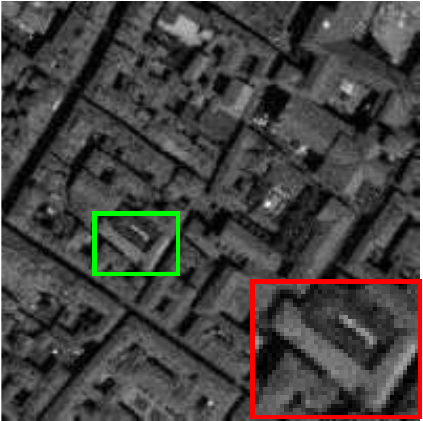}
	}
	\subfigure[BM4D(29.48dB)] {
		
		\includegraphics[width=0.33\columnwidth]{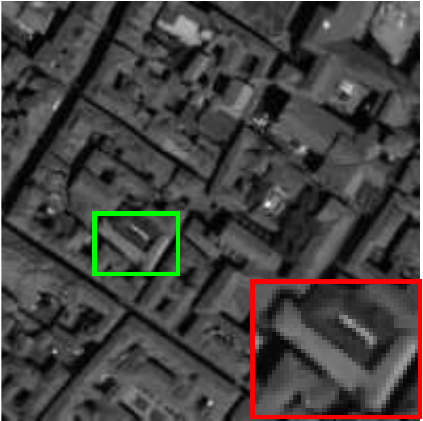}
	}
	\subfigure[LRMR(29.57dB)] {
		
		\includegraphics[width=0.33\columnwidth]{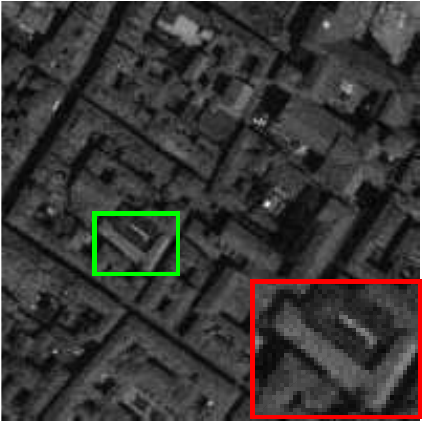}
	}
	\subfigure[LRTDTV(30.05dB)] {
		
		\includegraphics[width=0.33\columnwidth]{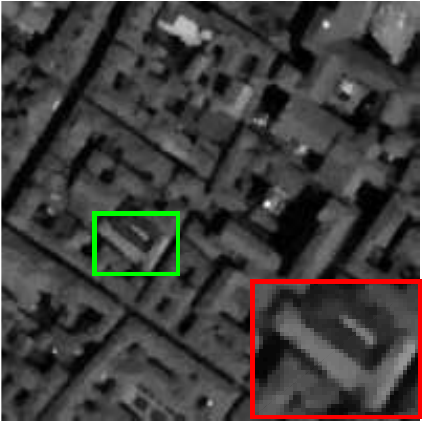}
	}
	\subfigure[3DTNN(28.31dB)] {
		
		\includegraphics[width=0.33\columnwidth]{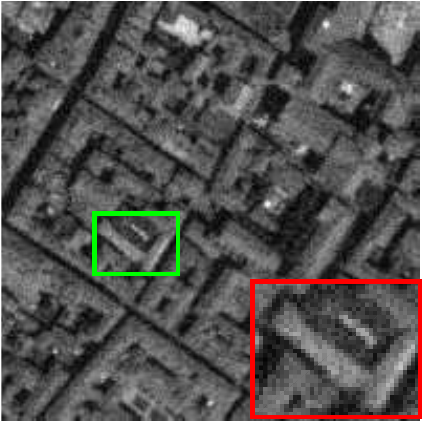}
	}
	\subfigure[3DLogTNN(30.15dB)] {
		
		\includegraphics[width=0.33\columnwidth]{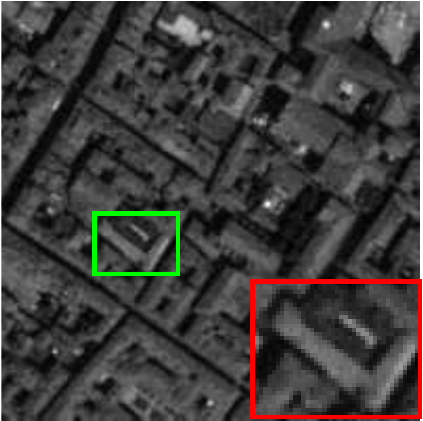}
	}
	\subfigure[MFWTNN(31.58dB)] {
		
		\includegraphics[width=0.33\columnwidth]{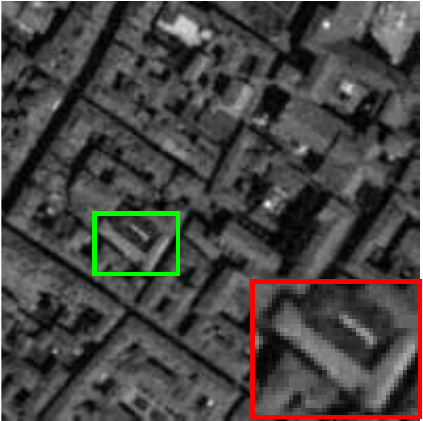}
	}
	\subfigure[\hspace*{-0.05in}NonMFWTNN(33.16dB)] {
		
		\includegraphics[width=0.33\columnwidth]{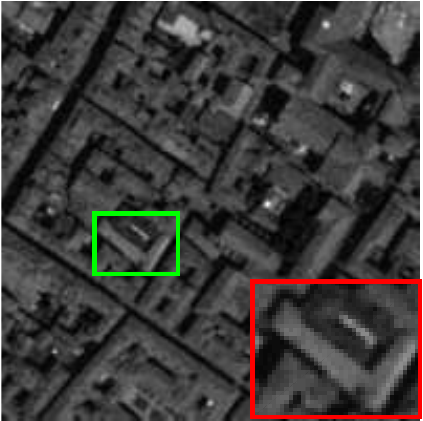}
	}
	\caption{ The 32th band of the denoised results of the Pavia City Center dataset under noise Case 3.}
	\label{Pavia_imshow0104_32}
\end{figure*}
In Fig. \ref{Paviaimshow0103_50}, we show the 50th band of the denoised results of all comparison methods in case 2.
In Fig. \ref{Pavia_imshow0104_32}, we show the 32th band of the denoised results in case 3.
As these grayscale images show, LRTA, LRMR and 3DTNN cannot completely remove some high-intensity noises.
For BM4D and LRTDTV, although they can remove more noise, they also lose more details.
This makes the denoised result too smooth.
Compared with them, our proposed models can remove more noise while retaining more details.
\begin{figure*}
	\centering
	\subfigure[Noisy image] {
		
		\includegraphics[width=0.33\columnwidth]{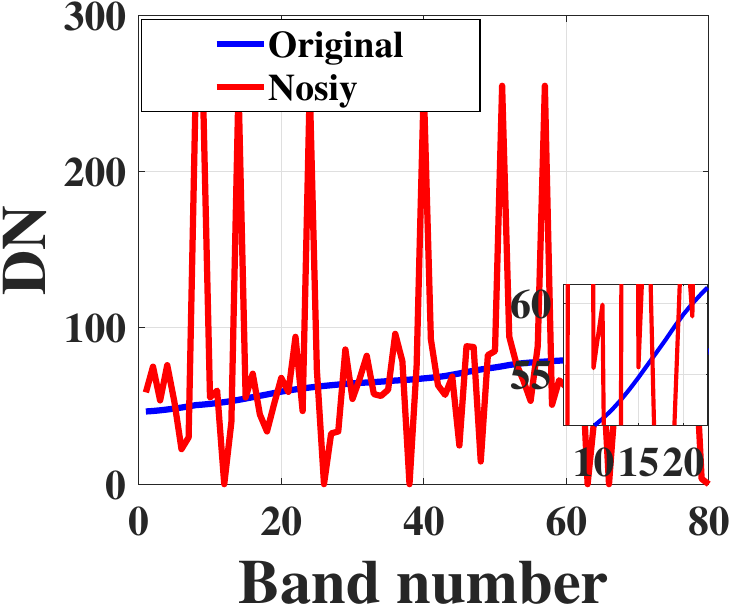}
	}
	\subfigure[LRTA] {
		
		\includegraphics[width=0.33\columnwidth]{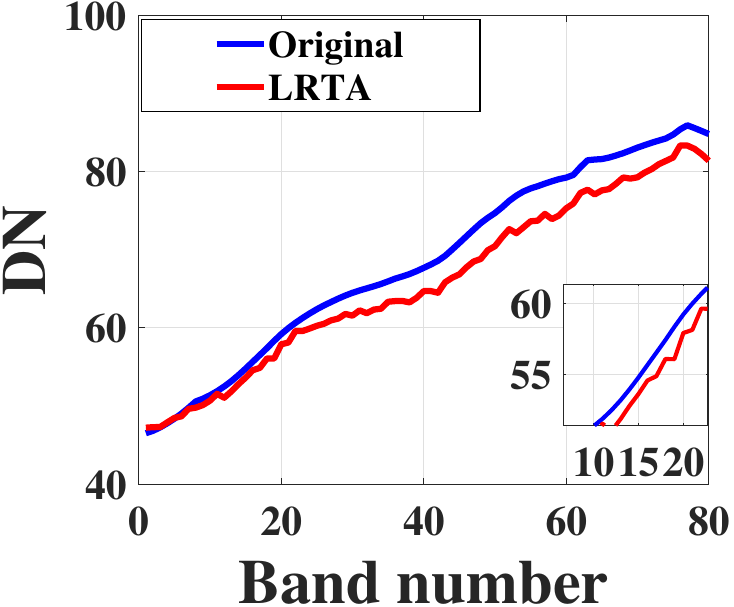}
	}
	\subfigure[BM4D] {
		
		\includegraphics[width=0.33\columnwidth]{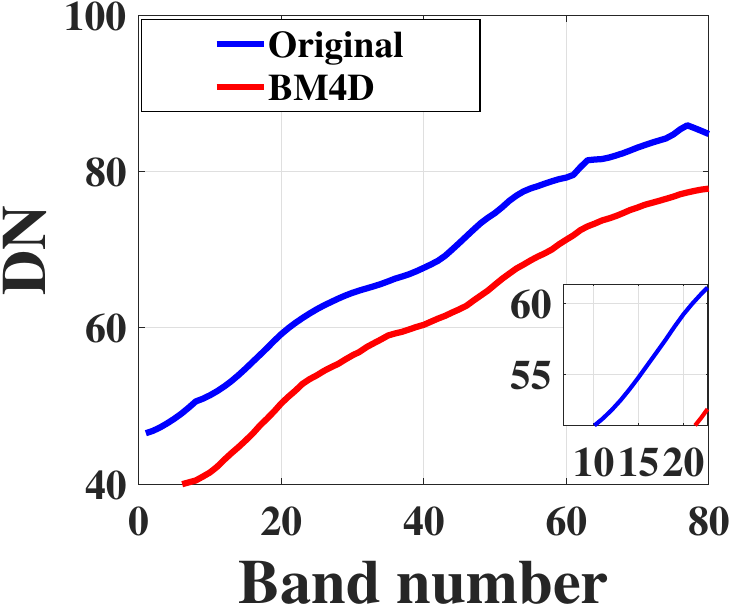}
	}
	\subfigure[LRMR] {
		
		\includegraphics[width=0.33\columnwidth]{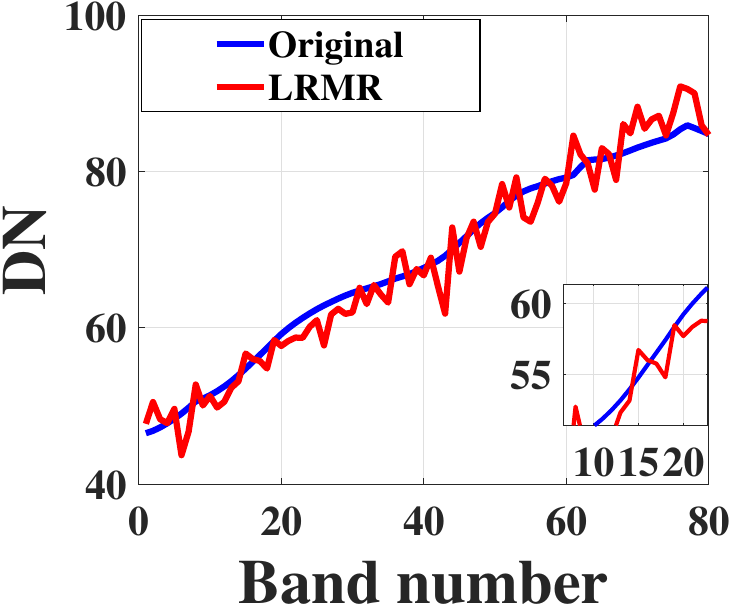}
	}
	\subfigure[LRTDTV] {
		
		\includegraphics[width=0.33\columnwidth]{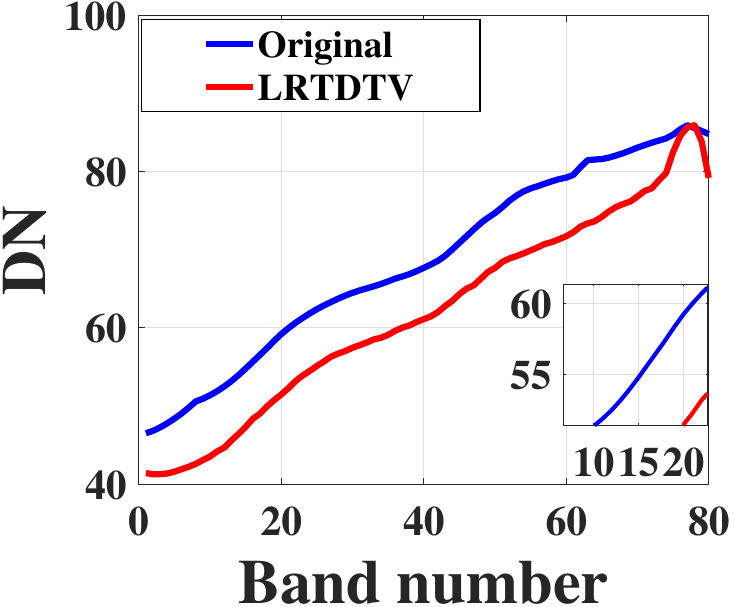}
	}
	\subfigure[3DTNN] {
		
		\includegraphics[width=0.33\columnwidth]{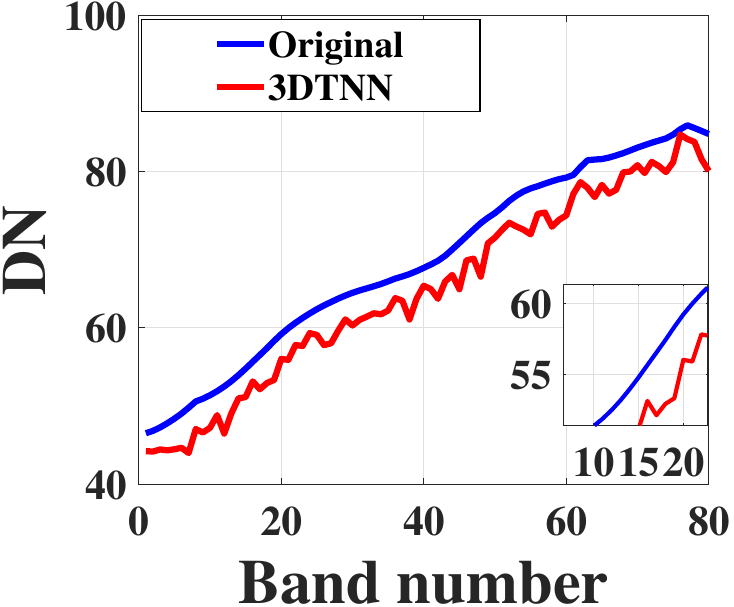}
	}
	\subfigure[3DLogTNN] {
		
		\includegraphics[width=0.33\columnwidth]{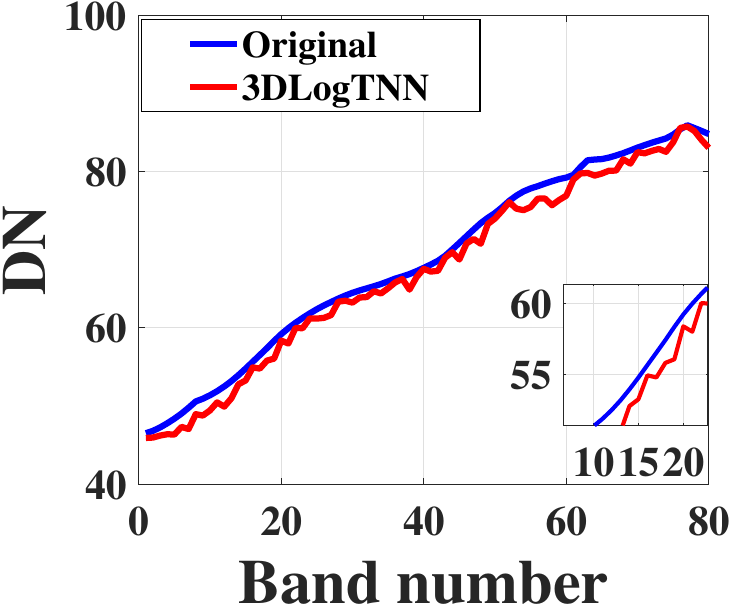}
	}
	\subfigure[MFWTNN] {
		
		\includegraphics[width=0.33\columnwidth]{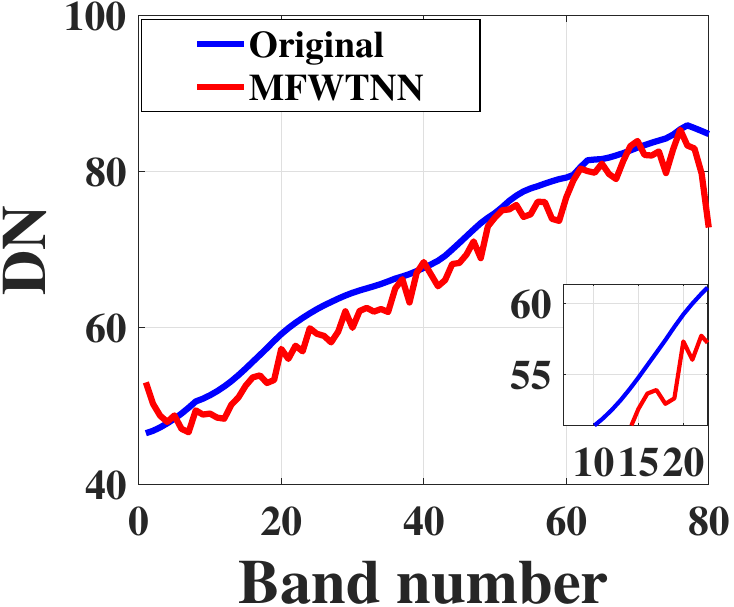}
	}
	\subfigure[NonMFWTNN] {
		
		\includegraphics[width=0.33\columnwidth]{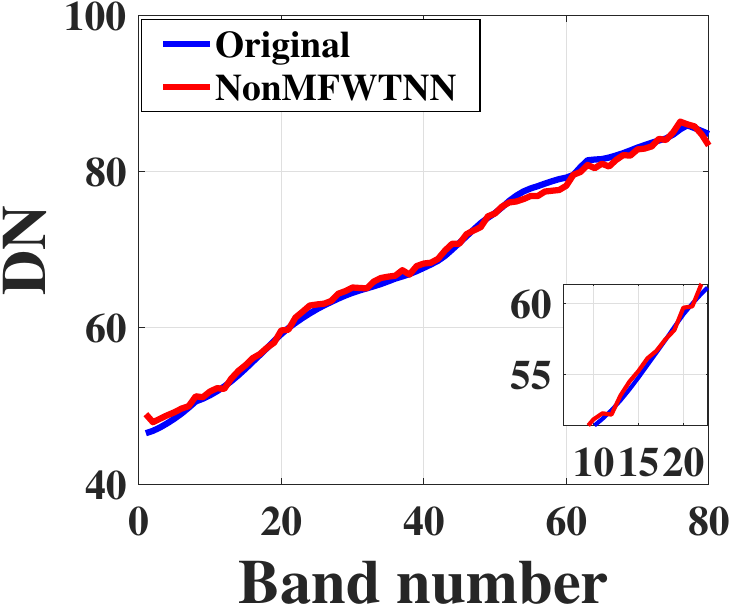}
	}
	\caption{The Spectral signatures curve of  Pavia City Center dataset in (170,99) under noise case 1.}
	\label{Pavia_imshow_pixel_0102}
\end{figure*}
The spectral signatures are important indicators for HSI.
To further compare the denoised image quality, 
we show the spectral curve of the pixel (170, 99) under case 1 in Fig. \ref{Pavia_imshow_pixel_0102}.
It can be seen that spectral signatures curves are relative to ERGAS values in table \ref{tab:PaviaTablePQI}, and our models can obtain the optimal value.

\begin{table*}
	\centering
	\caption{
		Quantitative comparison and time of eight denoising models on Washington DC Mall dataset under eight noise cases.
	}
	\scalebox{0.80}{
		\begin{tabular}{cccccccccccc}
			\toprule
			Case  & Level & Index & Noise & LRTA  & BM4D  & LRMR  & LRTDTV & 3DTNN & 3DLogTNN & MFWTNN & Non-MFWTNN \\
			\midrule
			\multirow{6}[2]{*}{Case 1} & \multirow{3}[1]{*}{G=0.1} & MPSNR & 11.024  & 32.089  & 31.209  & 31.856  & 33.243  & 33.227  & 35.518  & 34.296  & \textbf{36.344 } \\
			&       & MSSIM & 0.077  & 0.883  & 0.889  & 0.848  & 0.900  & 0.889  & 0.932  & 0.923  & \textbf{0.942 } \\
			&       & MFSIM & 0.438  & 0.938  & 0.919  & 0.929  & 0.931  & 0.943  & 0.961  & 0.950  & \textbf{0.966 } \\
			& \multirow{3}[1]{*}{P=0.2} & MERGAS & 1333.908  & 104.005  & 115.882  & 113.823  & 96.698  & 97.016  & 75.616  & 86.176  & \textbf{67.825 } \\
			&       & MSAM  & 43.237  & 4.116  & 4.519  & 4.937  & 4.175  & 4.136  & 3.148  & 3.579  & \textbf{2.864 } \\
			&       & time/s  & -  & 56.158  & 544.679  & 424.176  & 545.315  & 291.548  & 404.145  & 330.736  & 543.209  \\
			\midrule
			\multirow{6}[2]{*}{Case 2} & \multirow{3}[1]{*}{G=0.1} & MPSNR & 9.420  & 31.156  & 30.585  & 30.595  & 32.208  & 31.611  & 35.030  & 33.329  & \textbf{35.746 } \\
			&       & MSSIM & 0.049  & 0.863  & 0.878  & 0.814  & 0.881  & 0.882  & 0.927  & 0.908  & \textbf{0.935 } \\
			&       & MFSIM & 0.384  & 0.928  & 0.912  & 0.913  & 0.919  & 0.925  & 0.957  & 0.941  & \textbf{0.962 } \\
			& \multirow{3}[1]{*}{P=0.3} & MERGAS & 1611.157  & 115.466  & 124.189  & 130.309  & 110.034  & 123.077  & 80.302  & 96.378  & \textbf{73.121 } \\
			&       & MSAM  & 47.557  & 4.526  & 4.805  & 5.667  & 4.718  & 4.885  & 3.305  & 3.987  & \textbf{3.052 } \\
			&       & time/s  & -  & 55.783  & 544.675  & 382.850  & 539.389  & 272.741  & 407.267  & 330.584  & 554.722  \\
			\midrule
			\multirow{6}[2]{*}{Case 3} & \multirow{3}[1]{*}{G=0.1} & MPSNR & 8.253  & 29.852  & 29.630  & 29.149  & 30.686  & 29.982  & 33.579  & 32.342  & \textbf{34.402 } \\
			&       & MSSIM & 0.034  & 0.833  & 0.859  & 0.768  & 0.850  & 0.851  & 0.910  & 0.887  & \textbf{0.922 } \\
			&       & MFSIM & 0.346  & 0.911  & 0.899  & 0.890  & 0.901  & 0.910  & 0.940  & 0.928  & \textbf{0.947 } \\
			& \multirow{3}[1]{*}{P=0.4} & MERGAS & 1846.435  & 133.476  & 138.002  & 151.548  & 133.569  & 146.816  & 94.251  & 108.312  & \textbf{85.311 } \\
			&       & MSAM  & 50.441  & 5.126  & 5.262  & 6.635  & 6.007  & 5.955  & 3.761  & 4.478  & \textbf{3.490 } \\
			&       & time/s  & - & 57.749  & 528.381  & 447.568  & 584.706  & 323.020  & 424.376  & 370.543  & 601.149  \\
			\midrule
			\multirow{6}[2]{*}{Case 4} & \multirow{3}[1]{*}{G=0.15} & MPSNR & 10.635  & 29.490  & 28.928  & 29.286  & 31.359  & 31.477  & 33.121  & 32.220  & \textbf{34.568 } \\
			&       & MSSIM & 0.066  & 0.822  & 0.834  & 0.774  & 0.856  & 0.874  & 0.896  & 0.881  & \textbf{0.916 } \\
			&       & MFSIM & 0.411  & 0.906  & 0.879  & 0.893  & 0.902  & 0.920  & 0.941  & 0.923  & \textbf{0.952 } \\
			& \multirow{3}[1]{*}{P=0.2} & MERGAS & 1379.815  & 138.205  & 148.879  & 150.799  & 120.696  & 123.227  & 102.607  & 109.969  & \textbf{83.629 } \\
			&       & MSAM  & 44.251  & 5.410  & 5.774  & 6.561  & 5.241  & 4.979  & 4.167  & 4.636  & \textbf{3.518 } \\
			&       & time/s  & -  & 59.081  & 529.783  & 446.593  & 585.548  & 315.768  & 445.453  & 366.438  & 600.614  \\
			\midrule
			\multirow{6}[2]{*}{Case 5} & \multirow{3}[1]{*}{G=0.2} & MPSNR & 10.188  & 27.573  & 27.245  & 27.427  & 29.970  & 28.918  & 31.223  & 30.919  & \textbf{33.034 } \\
			&       & MSSIM & 0.056  & 0.766  & 0.783  & 0.705  & 0.814  & 0.809  & 0.859  & 0.844  & \textbf{0.891 } \\
			&       & MFSIM & 0.390  & 0.877  & 0.841  & 0.861  & 0.875  & 0.887  & 0.922  & 0.900  & \textbf{0.933 } \\
			& \multirow{3}[1]{*}{P=0.2} & MERGAS & 1437.539  & 170.752  & 179.235  & 183.763  & 141.567  & 166.078  & 132.159  & 127.459  & \textbf{100.032 } \\
			&       & MSAM  & 45.387  & 6.557  & 6.845  & 8.049  & 6.214  & 6.782  & 5.291  & 5.466  & \textbf{4.167 } \\
			&       & time/s  & -  & 54.268  & 547.776  & 427.316  & 540.747  & 290.324  & 399.630  & 340.013  & 546.367  \\
			\midrule
			\multirow{6}[2]{*}{Case 6} & \multirow{3}[1]{*}{G=0.1} & MPSNR & 9.512  & 31.135  & 30.555  & 30.572  & 32.063  & 31.553  & 34.525  & 33.337  & \textbf{35.728 } \\
			&       & MSSIM & 0.052  & 0.863  & 0.878  & 0.812  & 0.881  & 0.882  & 0.909  & 0.908  & \textbf{0.937 } \\
			&       & MFSIM & 0.388  & 0.927  & 0.911  & 0.912  & 0.920  & 0.925  & 0.950  & 0.941  & \textbf{0.960 } \\
			& \multirow{3}[1]{*}{P=(0.2,0.4)} & MERGAS & 1605.940  & 116.259  & 124.905  & 131.211  & 112.085  & 123.948  & 84.541  & 96.388  & \textbf{73.122 } \\
			&       & MSAM  & 47.511  & 4.539  & 4.814  & 5.708  & 4.944  & 4.934  & 3.907  & 3.991  & \textbf{3.045 } \\
			&       & time/s  & -  & 54.281  & 543.748  & 427.928  & 541.382  & 272.044  & 411.018  & 331.695  & 543.483  \\
			\midrule
			\multirow{6}[2]{*}{Case 7} & \multirow{3}[1]{*}{G=(0.1,0.3)} & MPSNR & 10.175  & 27.979  & 27.445  & 27.667  & 30.182  & 29.006  & 31.356  & 31.197  & \textbf{32.843 } \\
			&       & MSSIM & 0.057  & 0.787  & 0.795  & 0.714  & 0.827  & 0.815  & 0.869  & 0.853  & \textbf{0.895 } \\
			&       & MFSIM & 0.393  & 0.886  & 0.848  & 0.865  & 0.883  & 0.892  & 0.924  & 0.906  & \textbf{0.937 } \\
			& \multirow{3}[1]{*}{P=0.2} & MERGAS & 1444.165  & 165.560  & 176.071  & 182.967  & 142.743  & 167.263  & 128.476  & 124.285  & \textbf{104.171 } \\
			&       & MSAM  & 45.531  & 6.443  & 6.700  & 8.063  & 6.347  & 6.934  & 5.226  & 5.369  & \textbf{4.389 } \\
			&       & time/s  & -  & 53.851  & 543.508  & 429.160  & 534.911  & 289.693  & 394.217  & 336.045  & 540.726  \\
			\midrule
			\multirow{6}[2]{*}{Case 8} & \multirow{3}[1]{*}{G=(0.1,0.3)} & MPSNR & 8.937  & 26.924  & 26.614  & 26.396  & 29.159  & 27.306  & 29.804  & 30.031  & \textbf{30.303 } \\
			&       & MSSIM & 0.040  & 0.756  & 0.771  & 0.662  & 0.802  & 0.778  & 0.836  & 0.818  & \textbf{0.844 } \\
			&       & MFSIM & 0.357  & 0.868  & 0.830  & 0.841  & 0.870  & 0.863  & \textbf{0.914} & 0.887  & {0.909 } \\
			& \multirow{2}[0]{*}{P=(0.2,0.4)} & MERGAS & 1684.836  & 186.771  & 193.299  & 209.482  & 165.661  & 200.091  & 162.363  & 144.022  & \textbf{145.351 } \\
			&       & MSAM  & 48.823  & 7.237  & 7.322  & 9.412  & 7.297  & 8.633  & 6.768  & 6.360  & \textbf{6.325 } \\
			& stipes & time/s  & -  & 60.181  & 543.885  & 424.387  & 532.568  & 285.417  & 407.518  & 341.152  & 548.572  \\
			\bottomrule
		\end{tabular}%
		\label{tab:DCTablePQI}%
	}
\end{table*}%

2) Washington DC Mall: In this subsection, 
we show the evaluation results of all denoising models including visual and quantitative quality on the Washington DC Mall dataset under the condition of the eight noise cases.

Table \ref{tab:DCTablePQI} lists the PQIs of denoised results under eight noise cases in the Washington DC Mall dataset.
As this table shows, our proposed models obtain preeminently higher MPSNR and MSSIM values than those of the comparison denoising models.
\begin{figure*}
	\centering
	\includegraphics[width=0.9\columnwidth]{image/DCMall_PSNRSSIM/DCMall_tuli}\\	
	\subfigure[Case 1]{
		\begin{minipage}[t]{0.21\textwidth}
			\includegraphics[width=1.00\columnwidth]{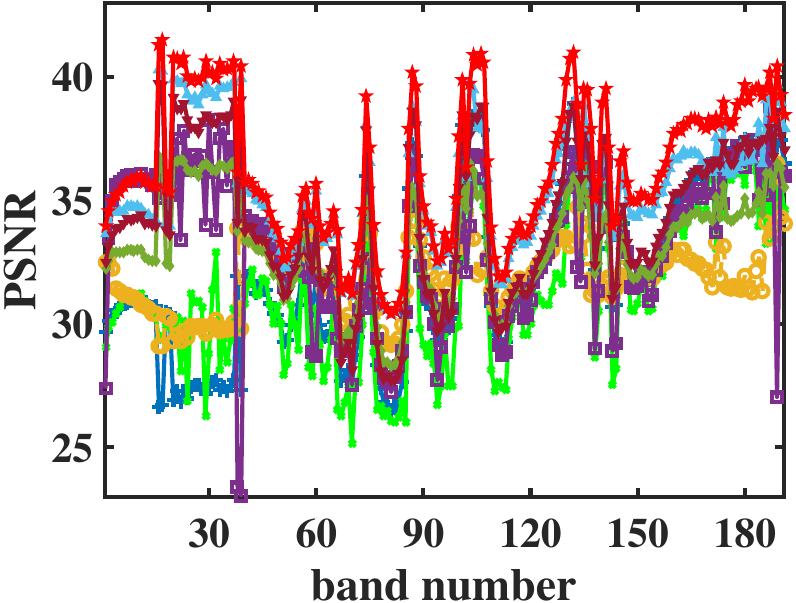}\\
			\includegraphics[width=1.00\columnwidth]{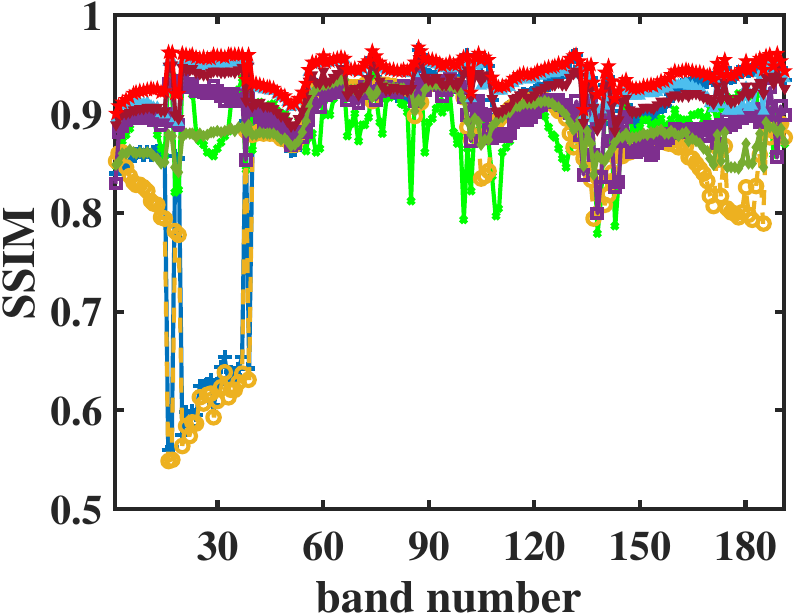}
		\end{minipage}
	}
	\subfigure[Case 2]{
		\begin{minipage}[t]{0.21\textwidth}
			\includegraphics[width=1.00\columnwidth]{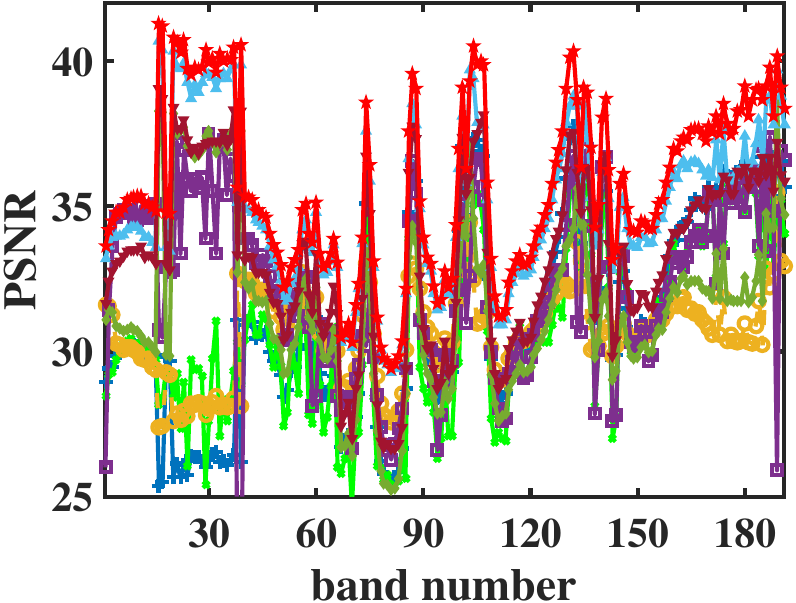}\\
			\includegraphics[width=1.00\columnwidth]{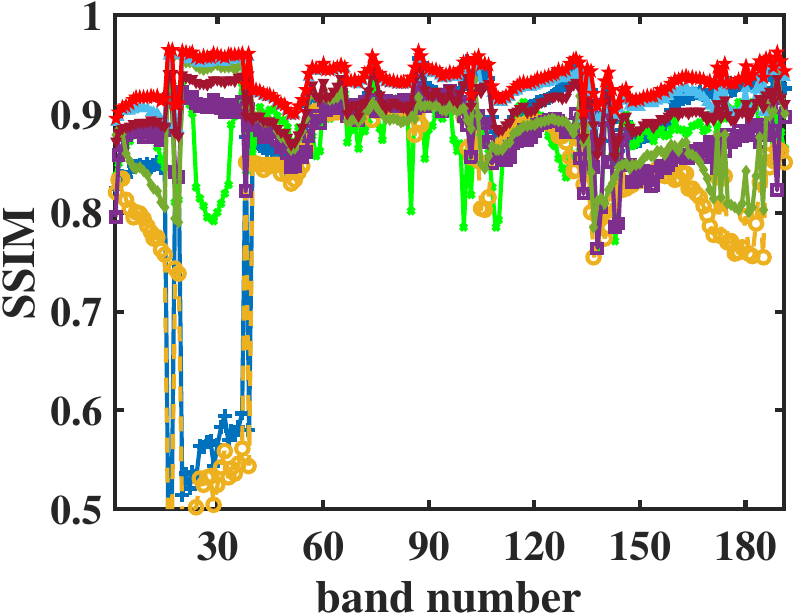}
		\end{minipage}
	}
	\subfigure[Case 3]{
		\begin{minipage}[t]{0.21\textwidth}
			\includegraphics[width=1.00\columnwidth]{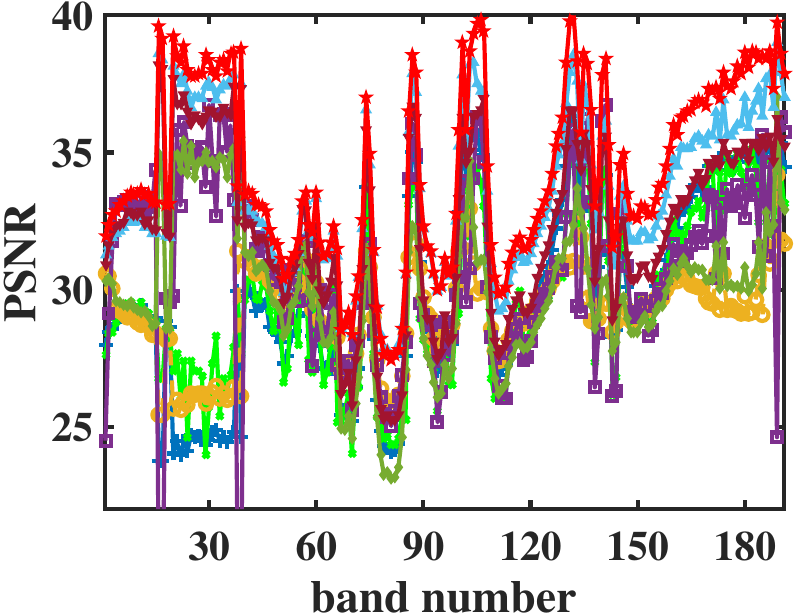}\\
			\includegraphics[width=1.00\columnwidth]{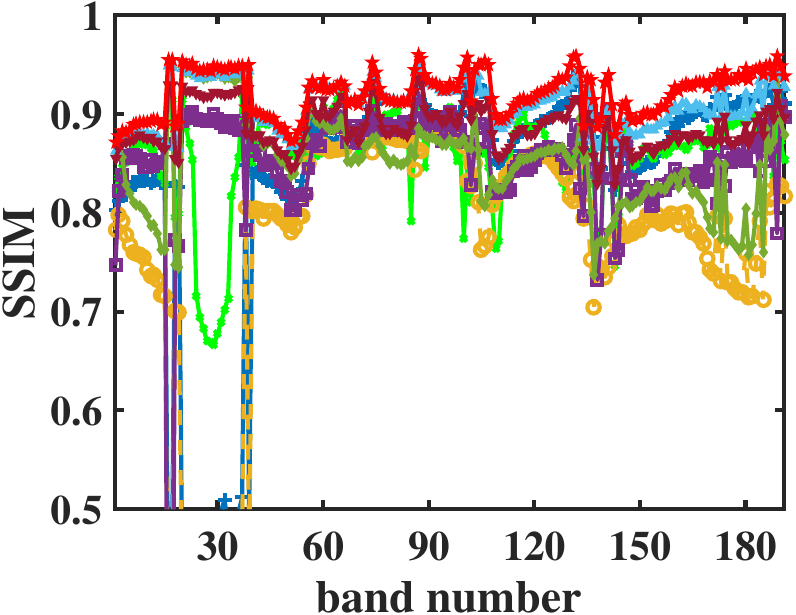}
		\end{minipage}
	}
	\subfigure[Case 4]{
		\begin{minipage}[t]{0.21\textwidth}
			\includegraphics[width=1.00\columnwidth]{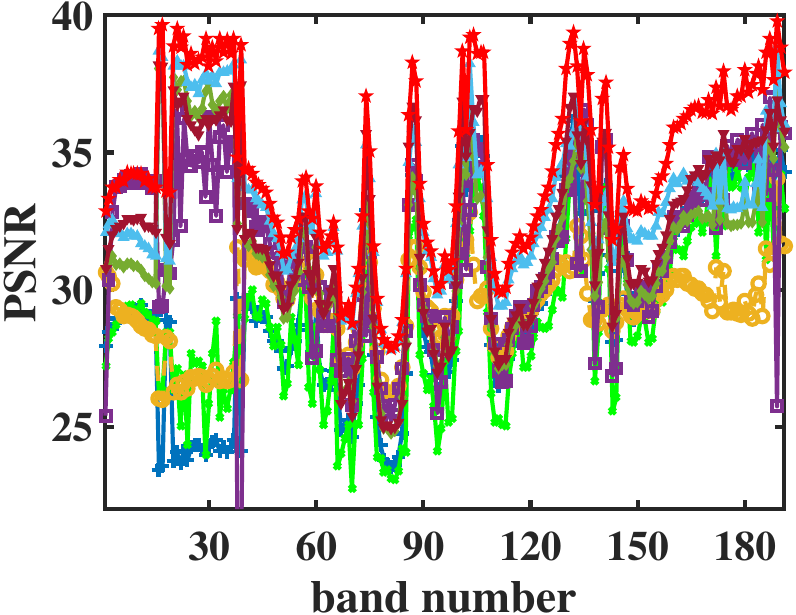}\\
			\includegraphics[width=1.00\columnwidth]{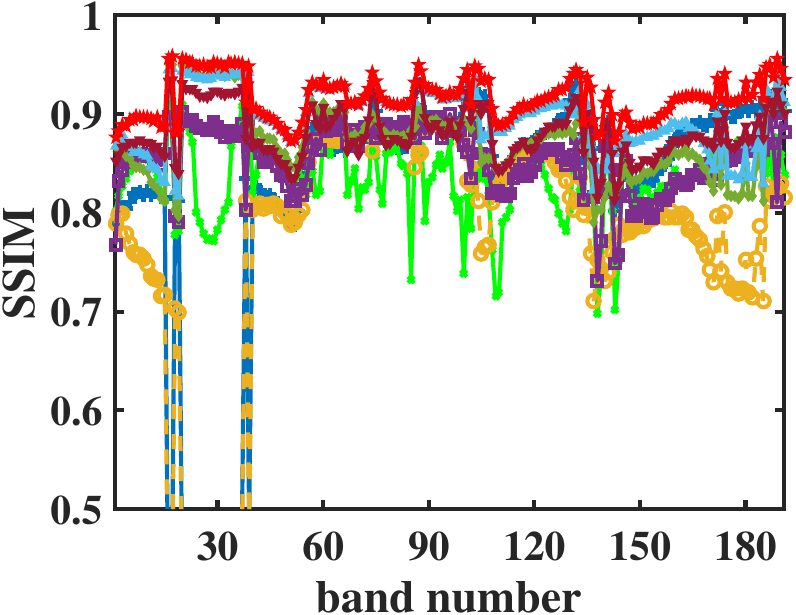}
		\end{minipage}
	}
	\subfigure[Case 5]{
		\begin{minipage}[t]{0.21\textwidth}
			\includegraphics[width=1.00\columnwidth]{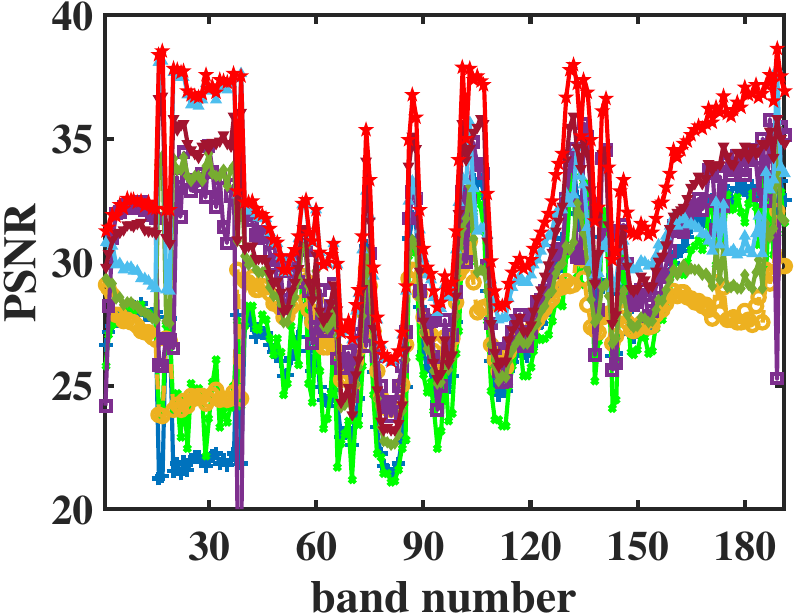}\\
			\includegraphics[width=1.00\columnwidth]{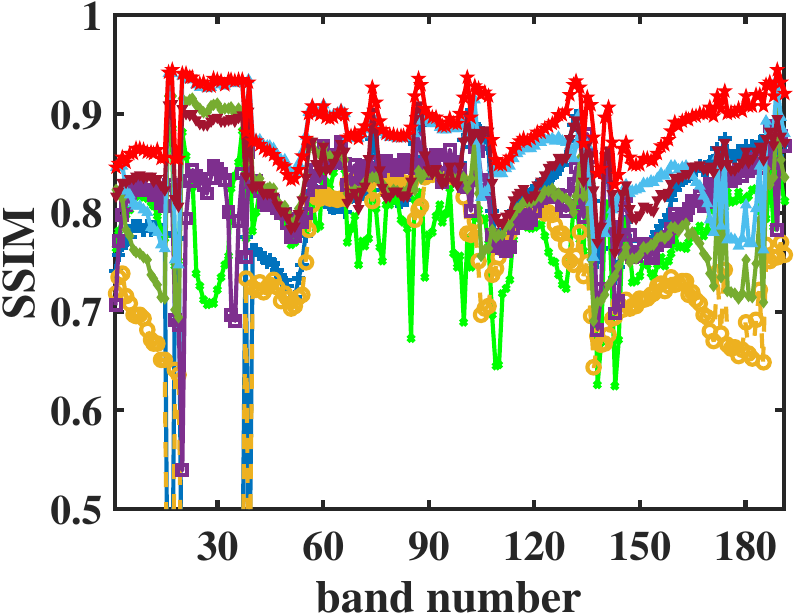}
		\end{minipage}
	}
	\subfigure[Case 6]{
		\begin{minipage}[t]{0.21\textwidth}
			\includegraphics[width=1.00\columnwidth]{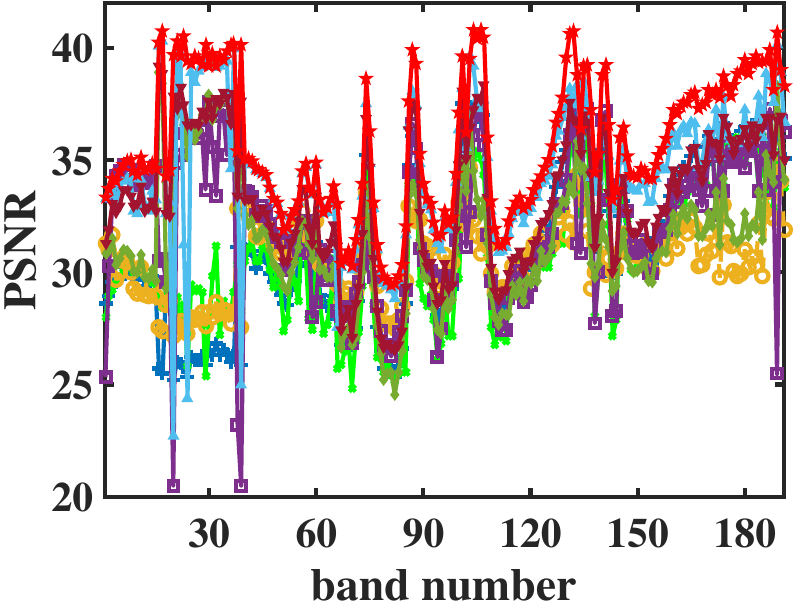}\\
			\includegraphics[width=1.00\columnwidth]{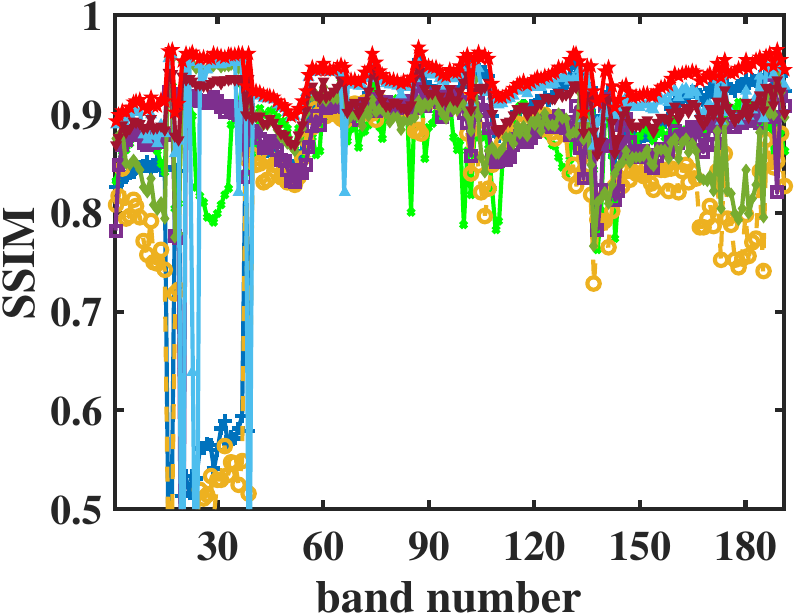}
		\end{minipage}
	}
	\subfigure[Case 7]{
		\begin{minipage}[t]{0.21\textwidth}
			\includegraphics[width=1.00\columnwidth]{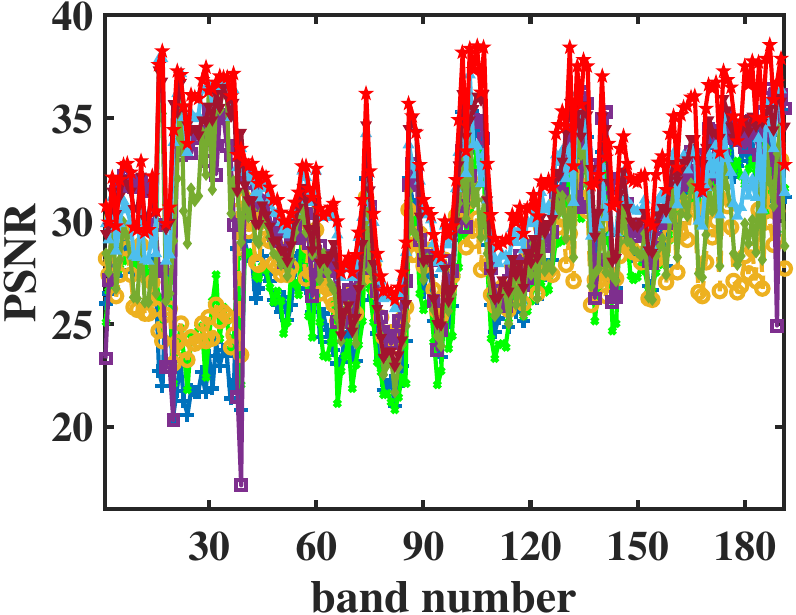}\\
			\includegraphics[width=1.00\columnwidth]{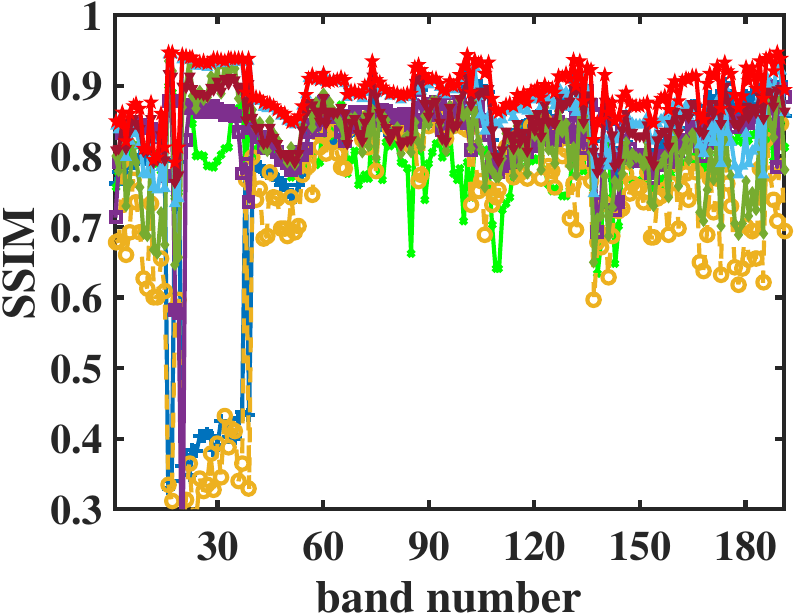}
		\end{minipage}
	}
	\subfigure[Case 8]{
		\begin{minipage}[t]{0.21\textwidth}
			\includegraphics[width=1.00\columnwidth]{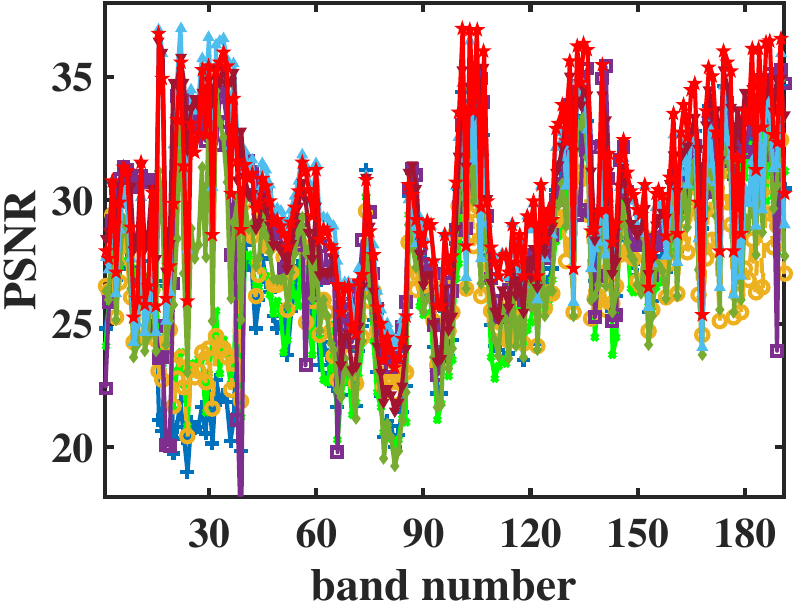}\\
			\includegraphics[width=1.00\columnwidth]{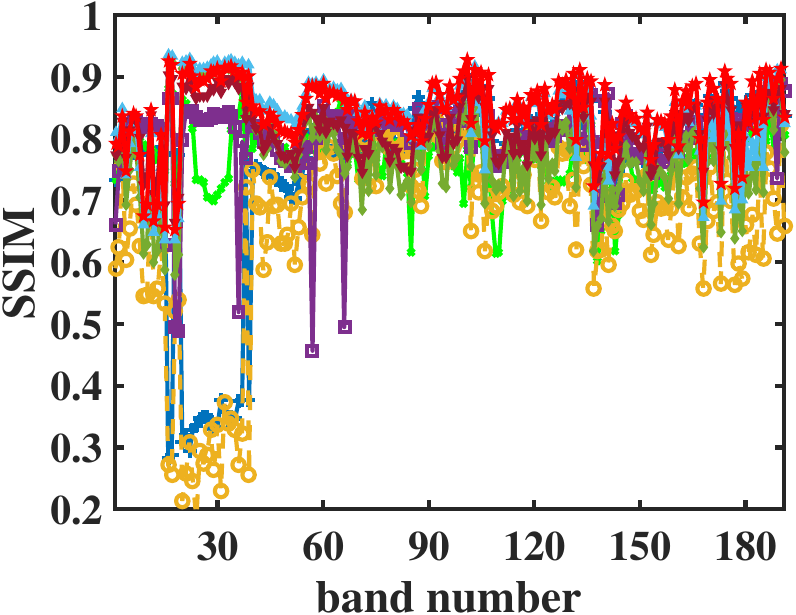}
		\end{minipage}
	}	
	\caption{ 
		The PSNR and SSIM of all denoising models for each band under the eight noise cases in Washington DC Mall dataset.
	}
	\label{DCMall_imshow_PSNR_SSIM}
\end{figure*}
As shown in the Fig. \ref{DCMall_imshow_PSNR_SSIM}, it shows the quantitative comparison with PSNR and SSIM of each band under all noisy cases in the Washington DC Mall dataset, and our methods obtain the optimal PSNR and SSIM values in most bands.

\begin{figure*}
	\centering
	\subfigure[Original image] {
		
		\includegraphics[width=0.33\columnwidth]{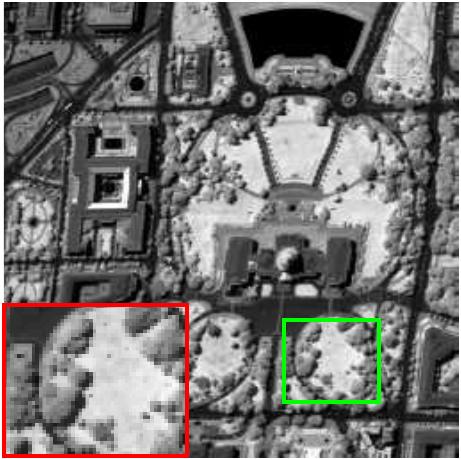}
	}
	\subfigure[Noisy image(10.02dB)] {
		
		\includegraphics[width=0.33\columnwidth]{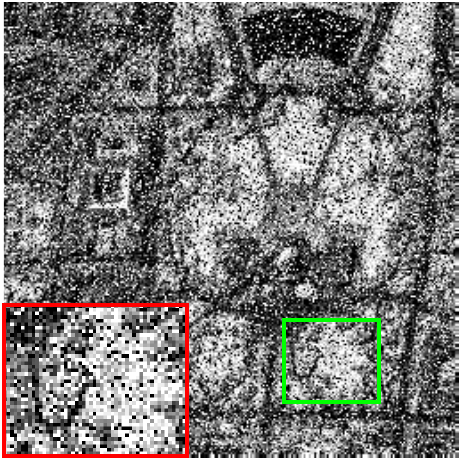}
	}
	\subfigure[LRTA(24.91dB)] {
		
		\includegraphics[width=0.33\columnwidth]{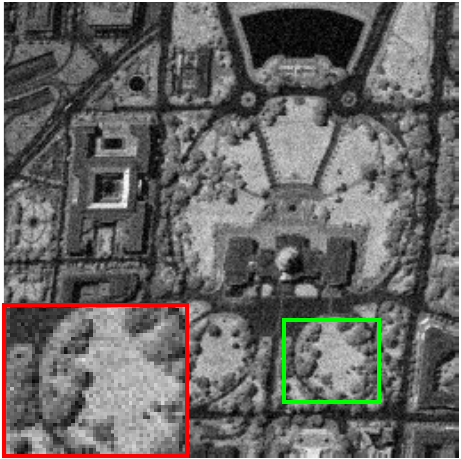}
	}
	\subfigure[BM4D(22.78dB)] {
		
		\includegraphics[width=0.33\columnwidth]{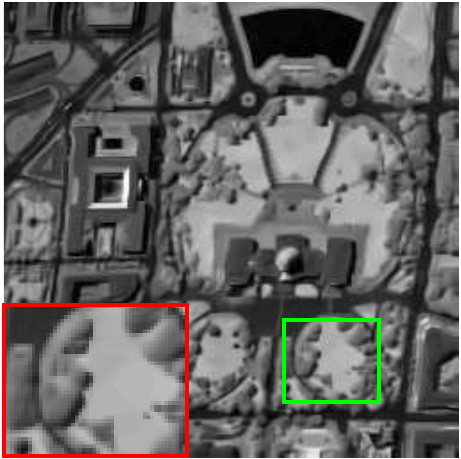}
	}
	\subfigure[LRMR(26.66dB)] {
		
		\includegraphics[width=0.33\columnwidth]{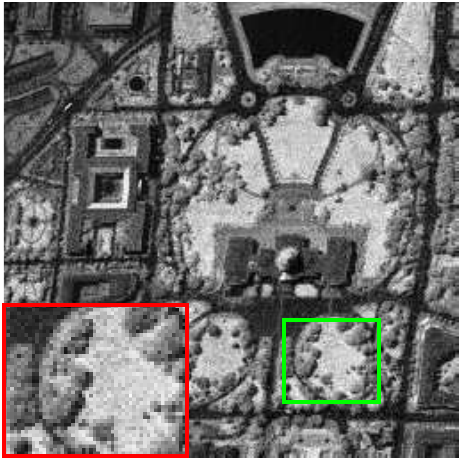}
	}
	\subfigure[LRTDLV(27.16dB)] {
		
		\includegraphics[width=0.33\columnwidth]{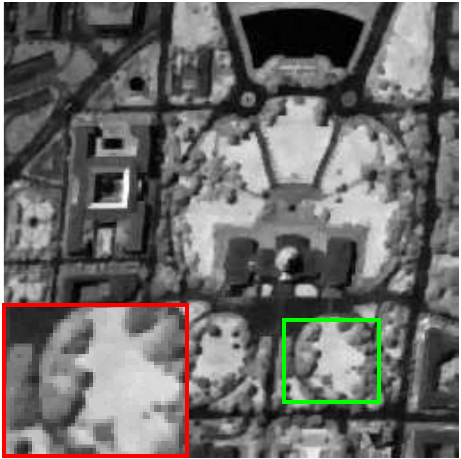}
	}
	\subfigure[3DTNN(26.04dB)] {
		
		\includegraphics[width=0.33\columnwidth]{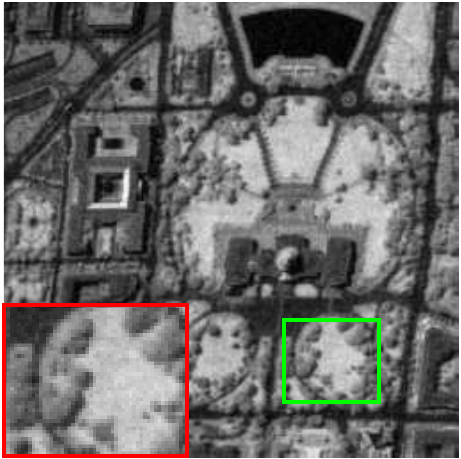}
	}
	\subfigure[3DLogTNN(28.84dB)] {
		
		\includegraphics[width=0.33\columnwidth]{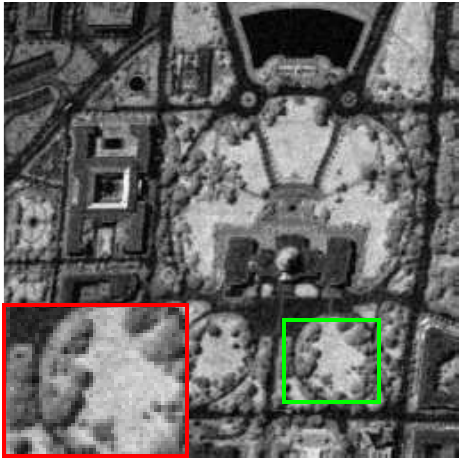}
	}
	\subfigure[MFWTNN(26.44dB)] {
		
		\includegraphics[width=0.33\columnwidth]{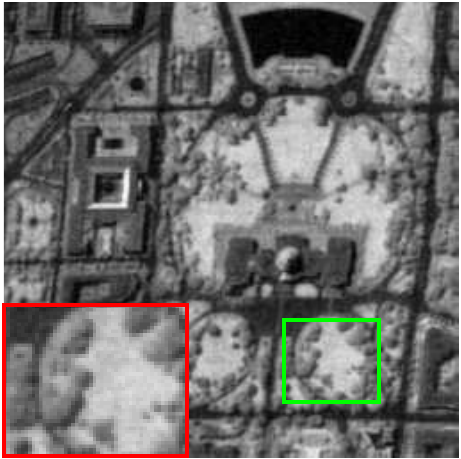}
	}
	\subfigure[\hspace*{-0.05in}NonMFWTNN(29.05dB)] {
		
		\includegraphics[width=0.33\columnwidth]{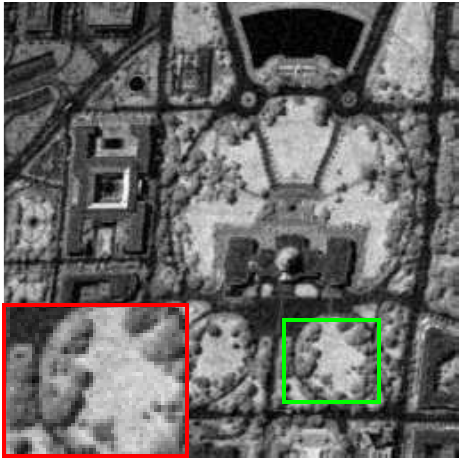}
	}
	\caption{  The 82th band of the denoised results of the Washington DC Mall dataset under noise Case 5.}
	\label{DC_imshow0202_82}
\end{figure*}

\begin{figure*}
	\centering
	\subfigure[Original image] {
		
		\includegraphics[width=0.33\columnwidth]{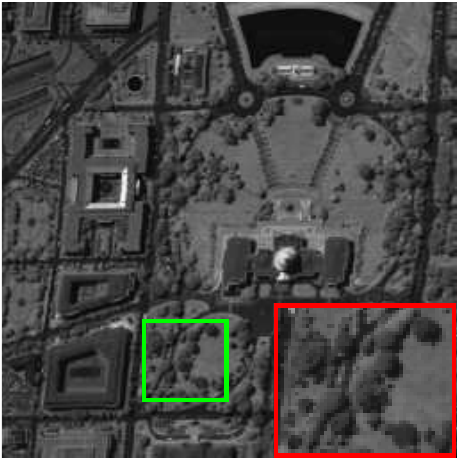}
	}
	\subfigure[Noisy image(10.33dB)] {
		
		\includegraphics[width=0.33\columnwidth]{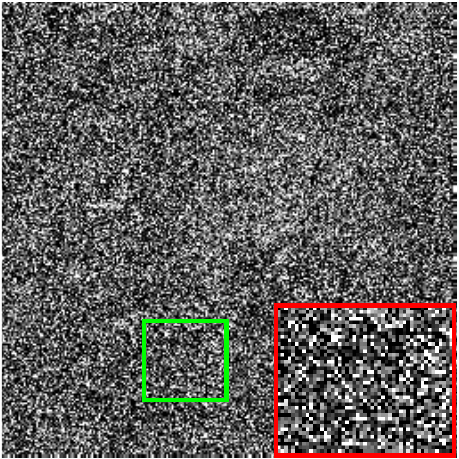}
	}
	\subfigure[LRTA(26.73dB)] {
		
		\includegraphics[width=0.33\columnwidth]{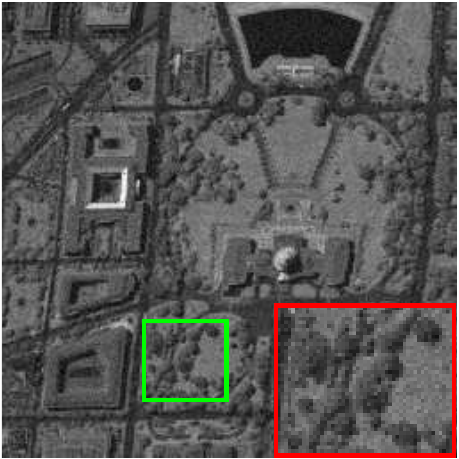}
	}
	\subfigure[BM4D(28.11dB)] {
		
		\includegraphics[width=0.33\columnwidth]{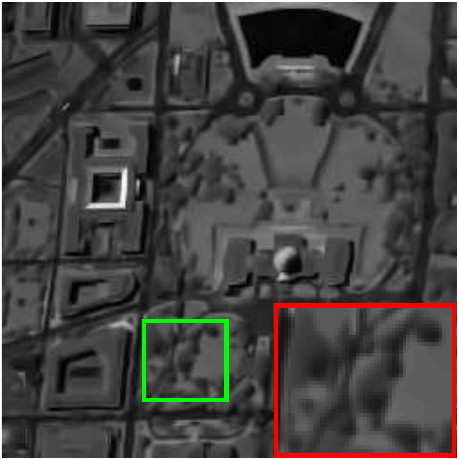}
	}
	\subfigure[LRMR(27.21dB)] {
		
		\includegraphics[width=0.33\columnwidth]{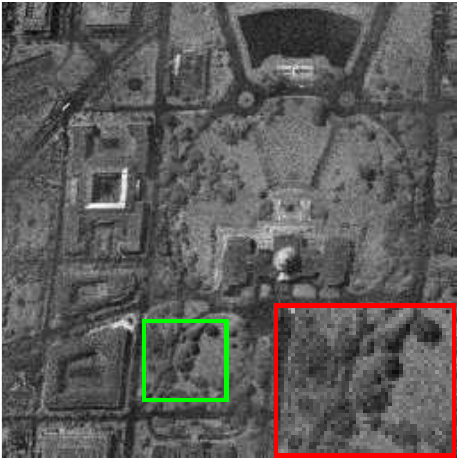}
	}
	\subfigure[LRTDLV(30.37dB)] {
		
		\includegraphics[width=0.33\columnwidth]{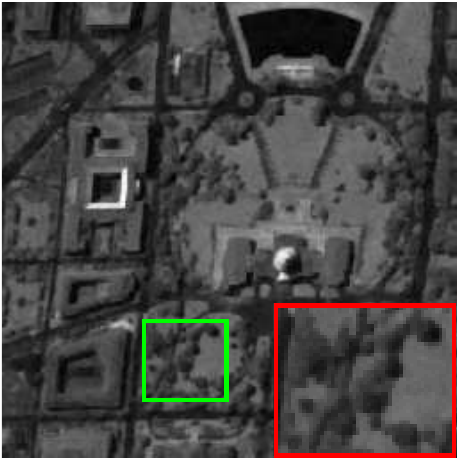}
	}
	\subfigure[3DTNN(29.44dB)] {
		
		\includegraphics[width=0.33\columnwidth]{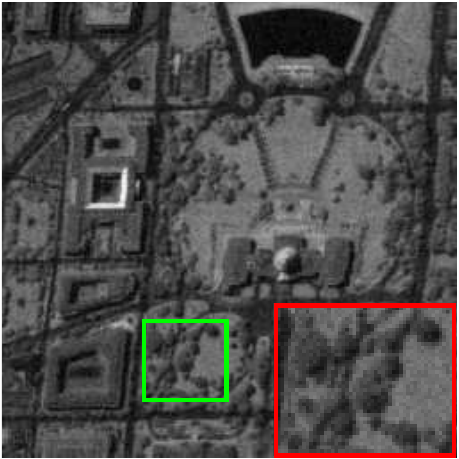}
	}
	\subfigure[3DLogTNN(32.13dB)] {
		
		\includegraphics[width=0.33\columnwidth]{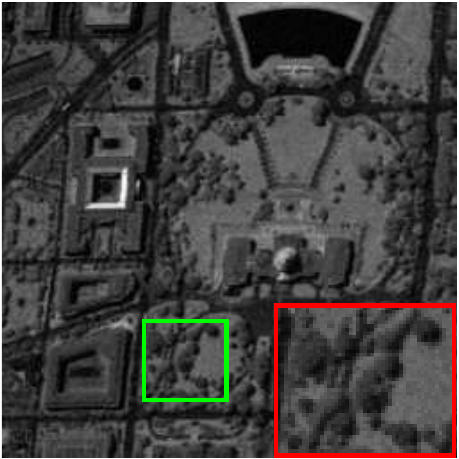}
	}
	\subfigure[MFWTNN(31.03dB)] {
		
		\includegraphics[width=0.33\columnwidth]{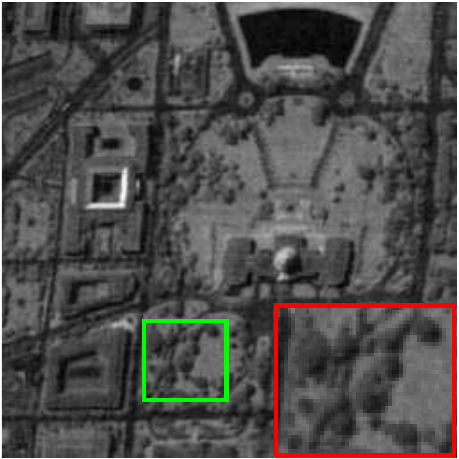}
	}
	\subfigure[\hspace*{-0.05in}NonMFWTNN(32.67dB)] {
		
		\includegraphics[width=0.33\columnwidth]{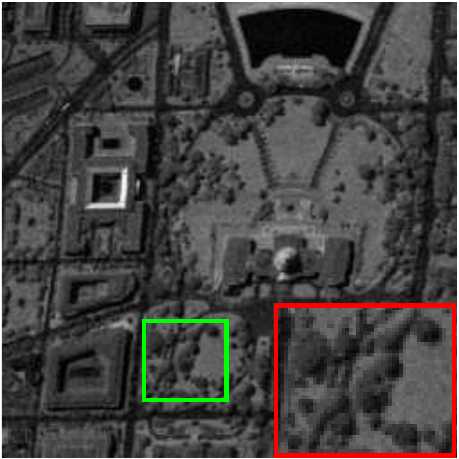}
	}
	\caption{ The 57th band of the denoised results of the Washington DC Mall dataset under noise Case 7.}
	\label{DC_imshow010302_57}
\end{figure*}
In Fig. \ref{DC_imshow0202_82}, we show the 82th band of the denoised results in case 5.
In Fig. \ref{DC_imshow010302_57}, we show the 57th band of the denoised results in case 7.
As these grayscale images show, LRTA, LRMR and 3DTNN cannot completely remove some high-intensity noises and retain more noise.
For BM4D and LRTDTV, although they can remove more noise, they also lose more details. This makes the denoised result too smooth.
Compared with them, our proposed models can remove more noise while retaining more details.
\begin{figure*}
	\centering
	\subfigure[Noisy image] {
		
		\includegraphics[width=0.33\columnwidth]{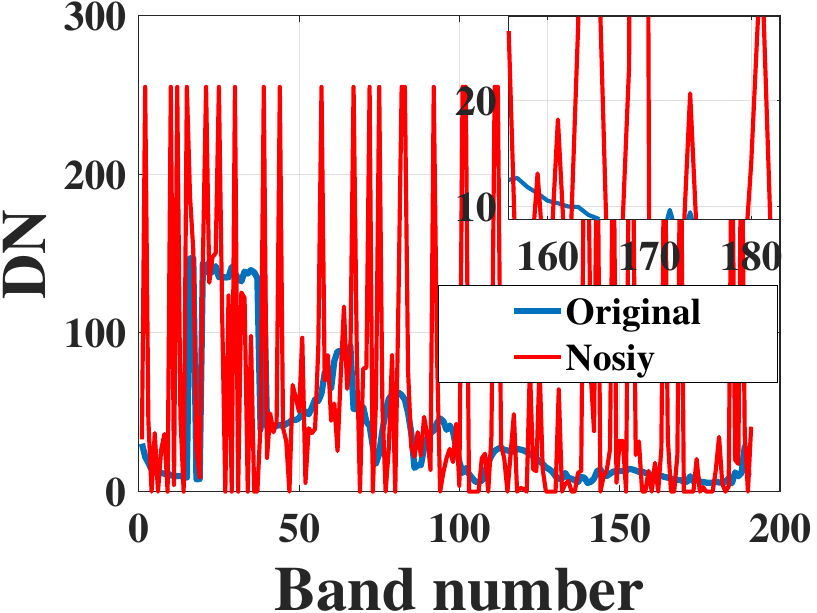}
	}
	\subfigure[LRTA] {
		
		\includegraphics[width=0.33\columnwidth]{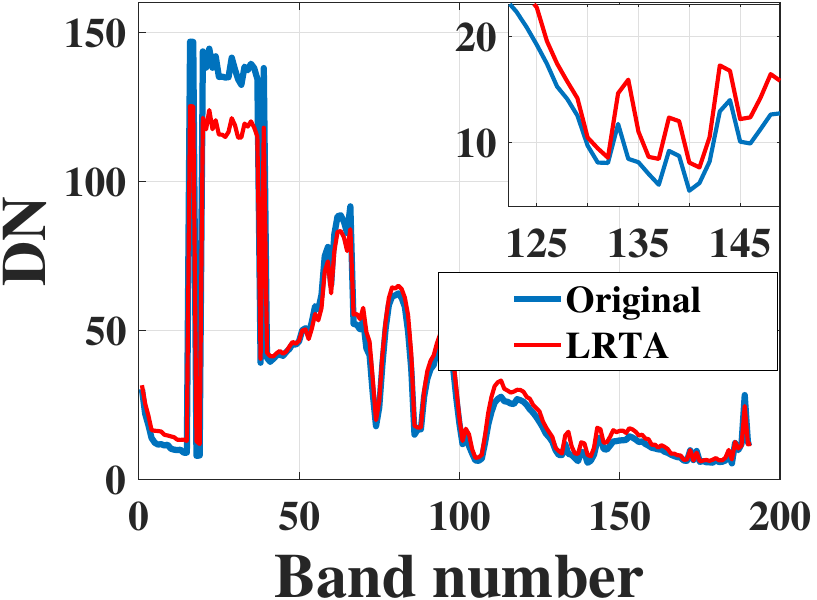}
	}
	\subfigure[BM4D] {
		
		\includegraphics[width=0.33\columnwidth]{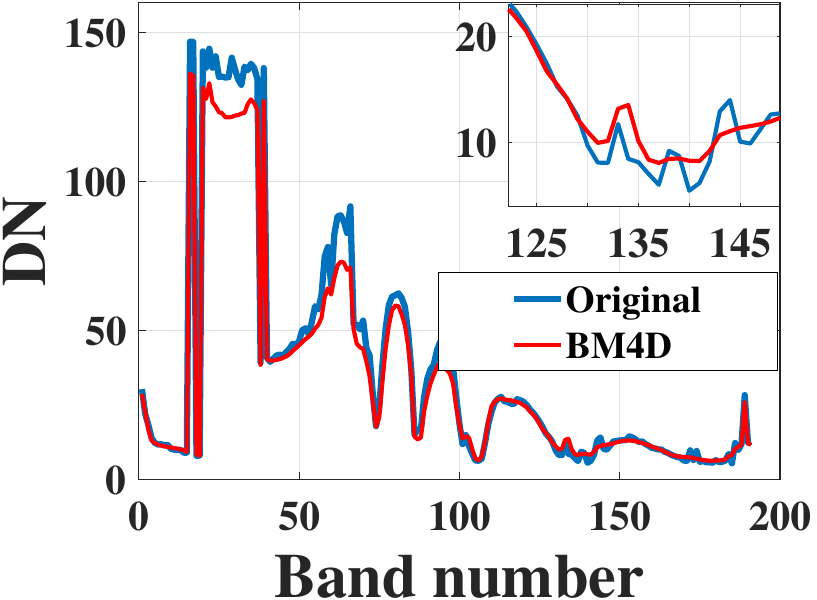}
	}
	\subfigure[LRMR] {
		
		\includegraphics[width=0.33\columnwidth]{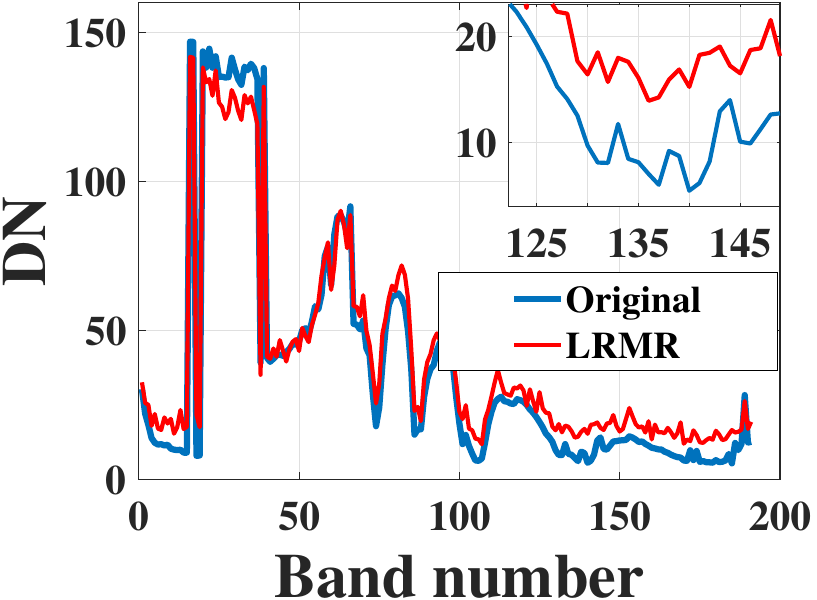}
	}
	\subfigure[LRTDTV] {
		
		\includegraphics[width=0.33\columnwidth]{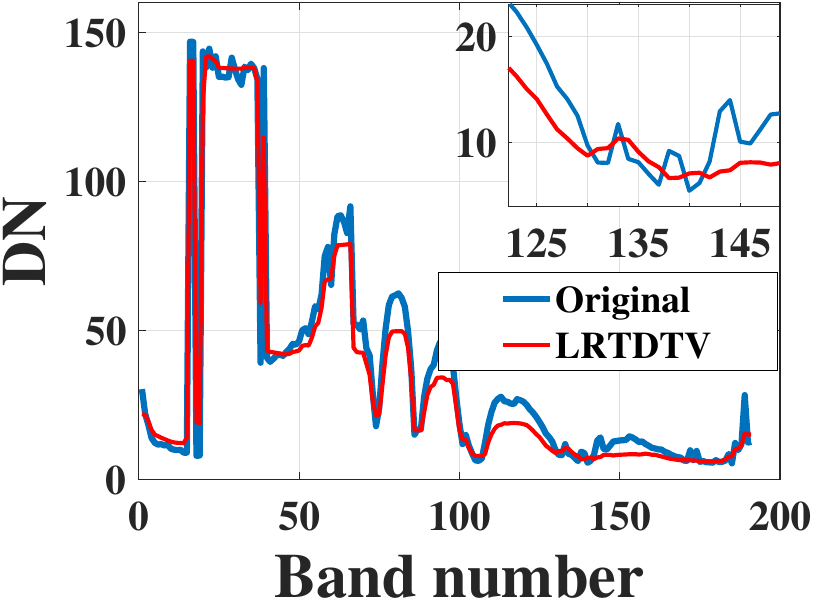}
	}
	\subfigure[3DTNN] {
		
		\includegraphics[width=0.33\columnwidth]{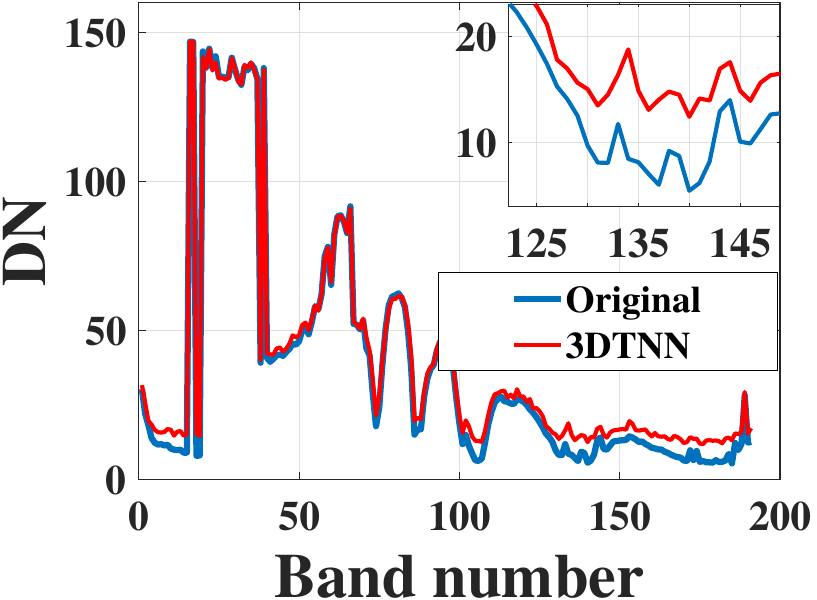}
	}
	\subfigure[3DLogTNN] {
		
		\includegraphics[width=0.33\columnwidth]{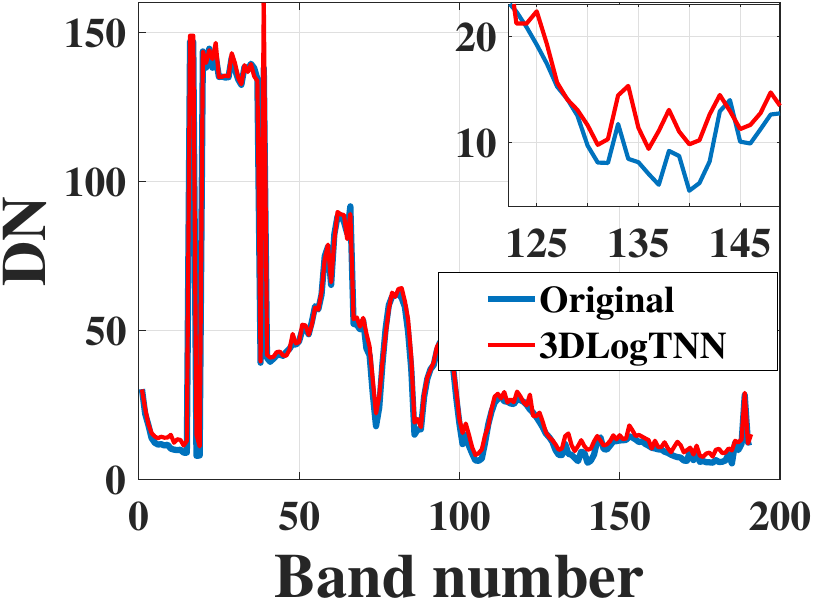}
	}
	\subfigure[MFWTNN] {
		
		\includegraphics[width=0.33\columnwidth]{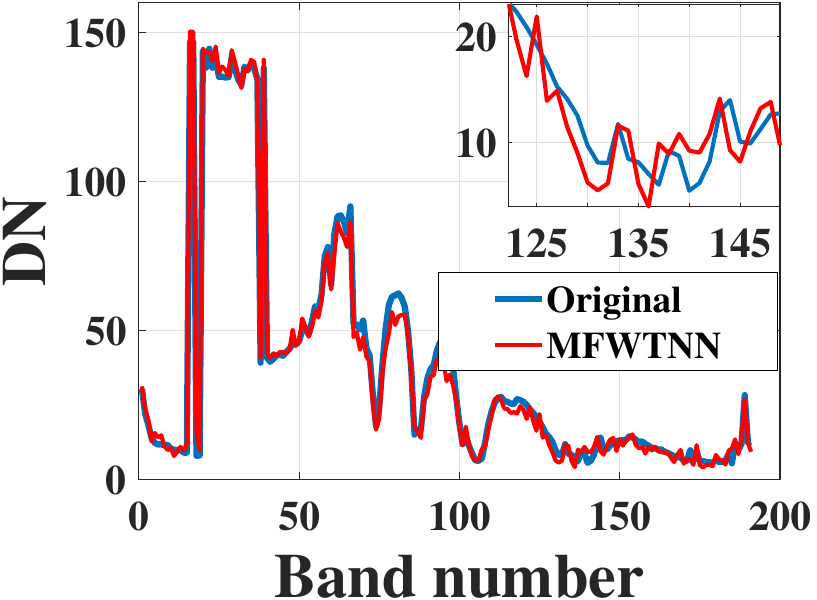}
	}
	\subfigure[NonMFWTNN] {
		
		\includegraphics[width=0.33\columnwidth]{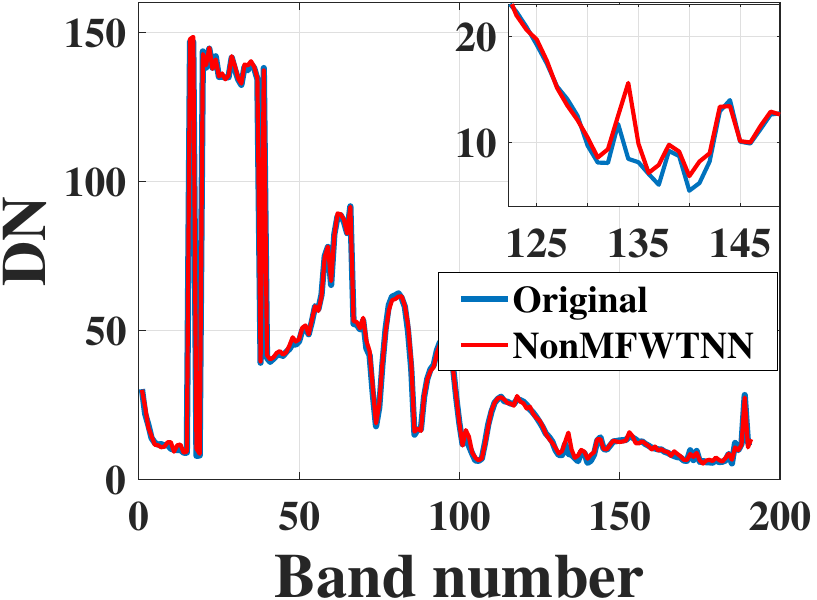}
	}
	\caption{  The Spectral signatures curve of Washington DC Mall in (234,171) under noise case 6}
	\label{DC_imshow_pixel_0103}
\end{figure*}
Similar to the Pavia City Center dataset, we show the spectral curve of the pixel (234, 71) under case 6 in Fig. \ref{DC_imshow_pixel_0103}.
It can be seen that spectral signatures curves are relative to ERGAS values in Table \ref{tab:DCTablePQI}, and our models can obtain the optimal result.

\subsection{Real HSI Data Experiments}
In this subsection, we chose the real HSI dataset in this experiment.
It is AVIRIS Indian Piness dataset\footnote{\url{https://engineering.purdue.edu/}}.
It is filmed by the Airborne Visible Infrared Imaging Spectrometer (AVIRIS).
The reason we chose a sub-block with a size of $145 \times 145 \times 220$ as a real dataset is because some bands are seriously polluted by Gaussian white noise and impulse noise and have lost the meaning of reference.
We evaluate our denoising model from visual evaluation and vertical average profile.
In Fig. \ref{Indian_imshow0103_105}, we show the 105th band of the observed and denoised results in Indian dataset.
It can be seen that it is seriously polluted by Gaussian and sparse noise, and the original ground features information cannot be distinguished.
Although the competitive denoising methods restore the feature information of HSI, the denoised HSIs still lose detailed information or residual noise, resulting in missing or blurred texture details of HSI.
Our models not only restore more feature information of HSI but also remove more local noise than the competitive methods.
Therefore, these denoised HSIs have clearer texture information.
In Fig. \ref{real_imshowIndianDN}, we show the vertical mean profiles of the 50th band of Indian dataset.
Due to the influence of hybrid noises, the vertical mean profiles of the observed HSI shows rapid fluctuations.
This means that the smaller fluctuations in the vertical mean profiles, the higher quality of the denoised HSI.
Obviously, the mean profile curves obtained by our model are the most stable and its fluctuation is the smallest.
This result is relative to the facts reflected in Fig. \ref{Indian_imshow0103_105}.
\begin{figure*}[htbp]
	\centering
	%
	\subfigure[Noisy image] {
		
		\includegraphics[width=0.25\columnwidth]{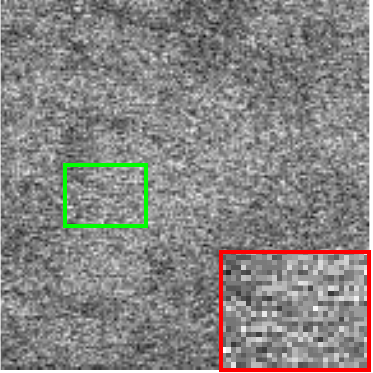}
	}
%
	\subfigure[BM4D] {
		
		\includegraphics[width=0.25\columnwidth]{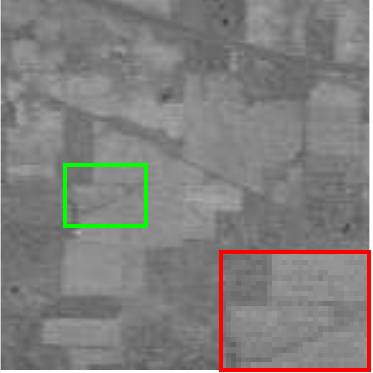}
	}
	\subfigure[LRMR] {
		
		\includegraphics[width=0.25\columnwidth]{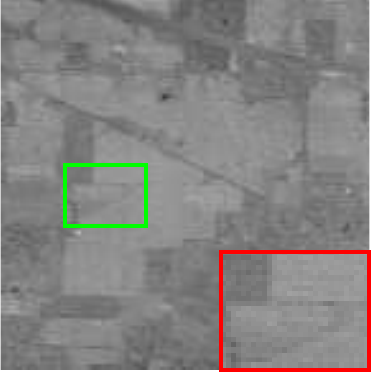}
	}
	\subfigure[LRTDTV] {
		
		\includegraphics[width=0.25\columnwidth]{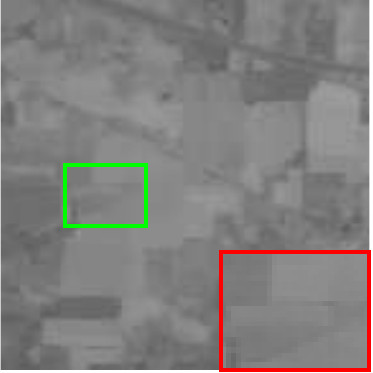}
	}
%
	\subfigure[3DLogTNN] {
		
		\includegraphics[width=0.25\columnwidth]{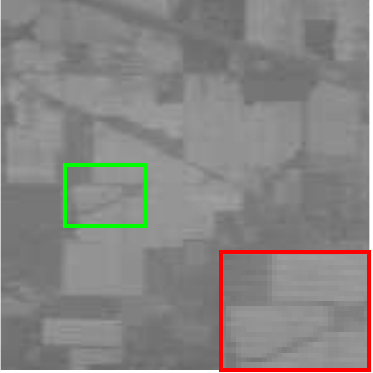}
	}
	\subfigure[MFWTNN] {
		
		\includegraphics[width=0.25\columnwidth]{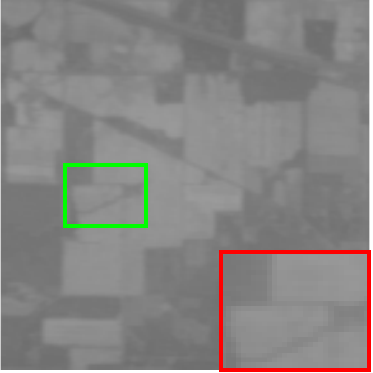}
	}
	\subfigure[NonMFWTNN] {
		
		\includegraphics[width=0.25\columnwidth]{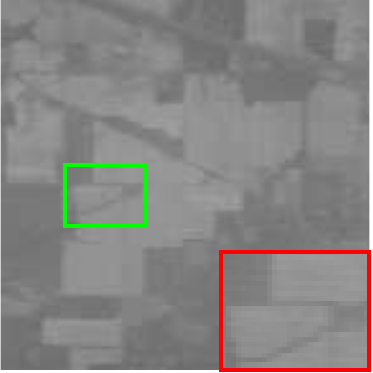}
	}
	\caption{ The 105th band of all denoising results for the Indian dataset.}
	\label{Indian_imshow0103_105}
\end{figure*}

\begin{figure*}[htbp]
	\centering
	\subfigure[Noisy image] {
		
		\includegraphics[width=0.25\columnwidth]{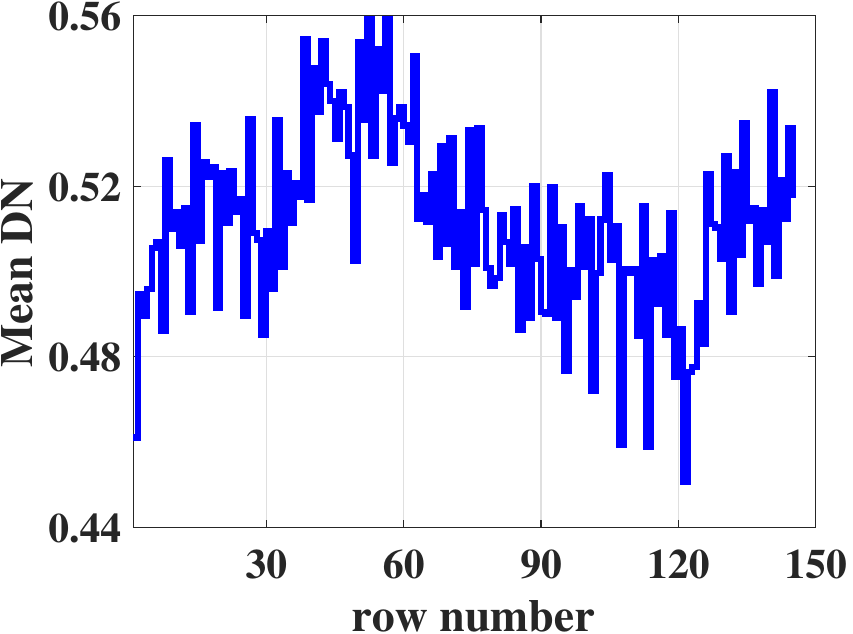}
	}
	\subfigure[BM4D] {
		
		\includegraphics[width=0.25\columnwidth]{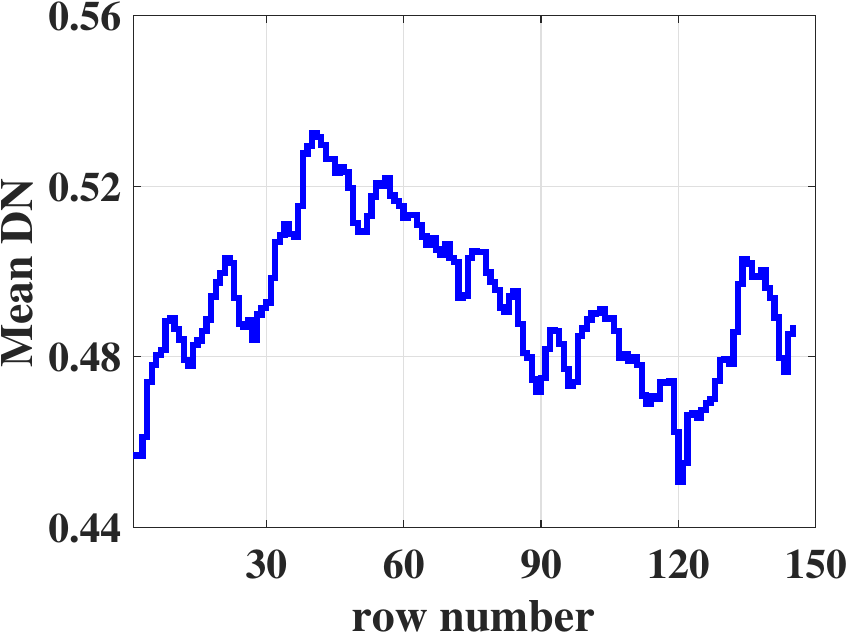}
	}
	\subfigure[LRMR] {
		
		\includegraphics[width=0.25\columnwidth]{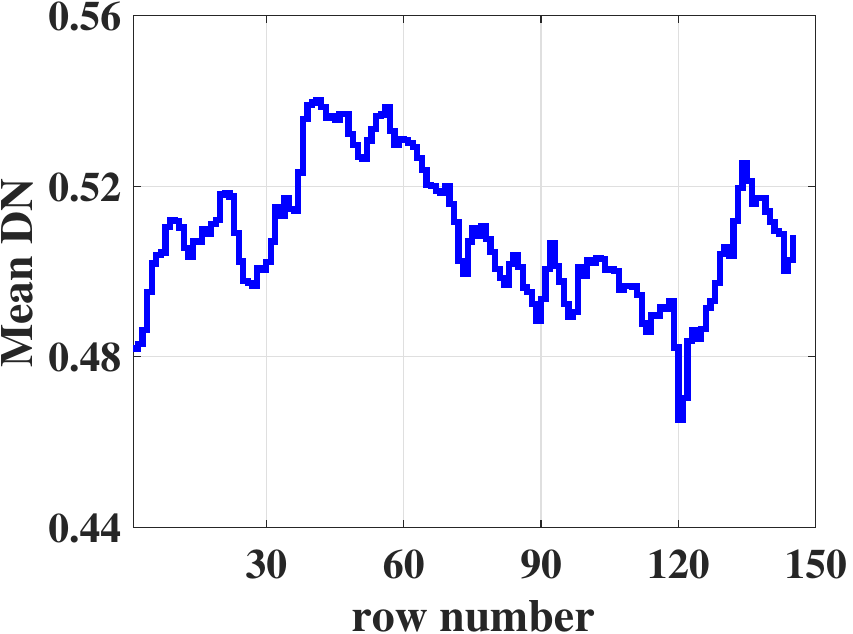}
	}
	\subfigure[LRTDTV] {
		
		\includegraphics[width=0.25\columnwidth]{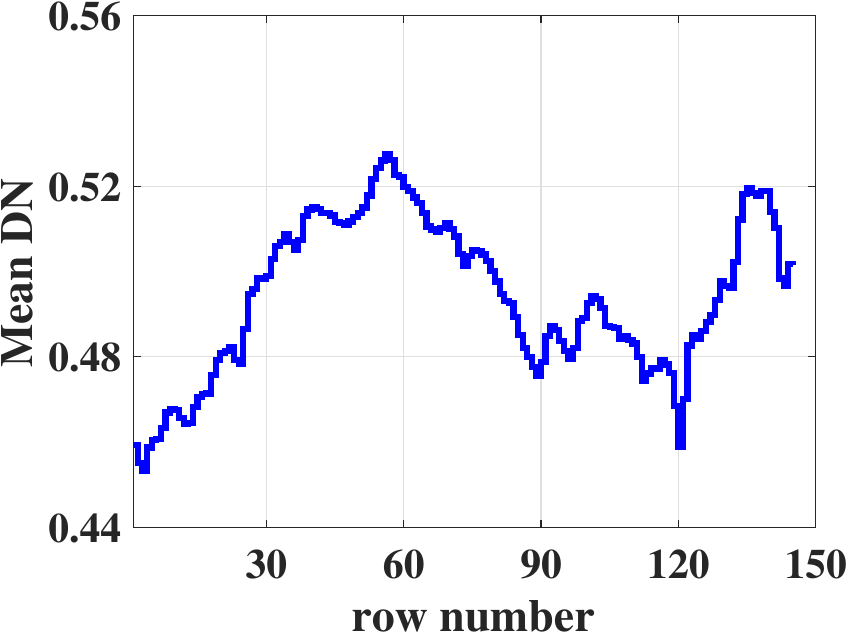}
	}
	\subfigure[3DLogTNN] {
		
		\includegraphics[width=0.25\columnwidth]{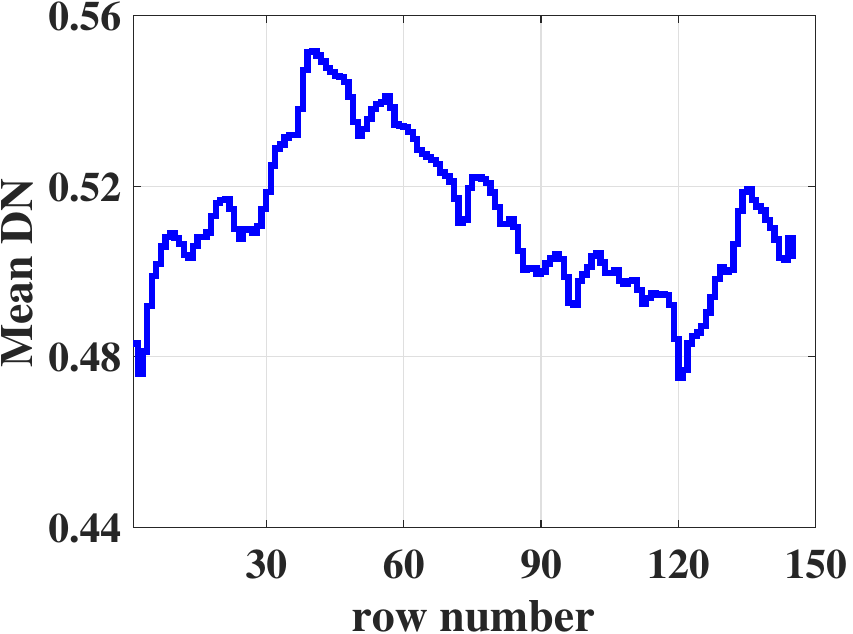}
	}
	\subfigure[MFWTNN] {
		
		\includegraphics[width=0.25\columnwidth]{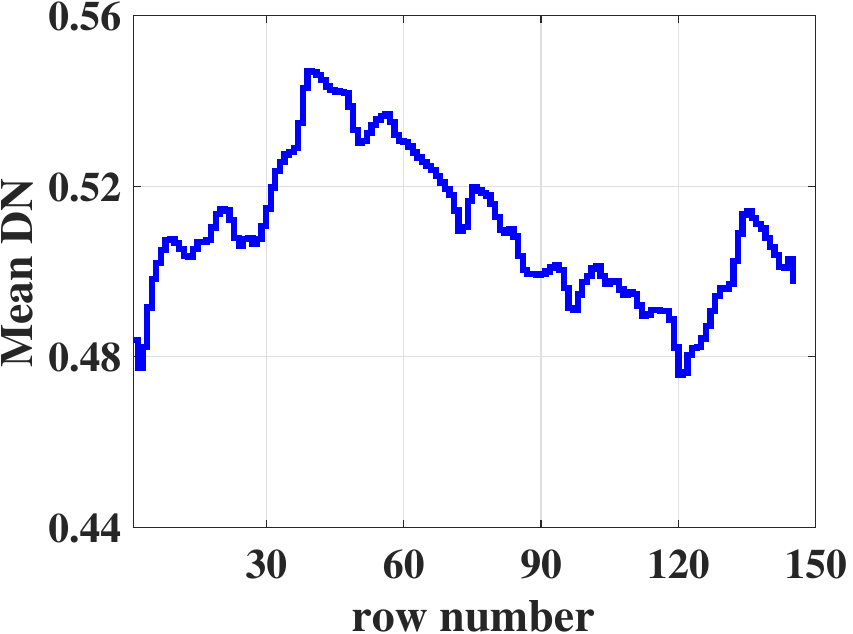}
	}
	\subfigure[NonMFWTNN] {
		
		\includegraphics[width=0.25\columnwidth]{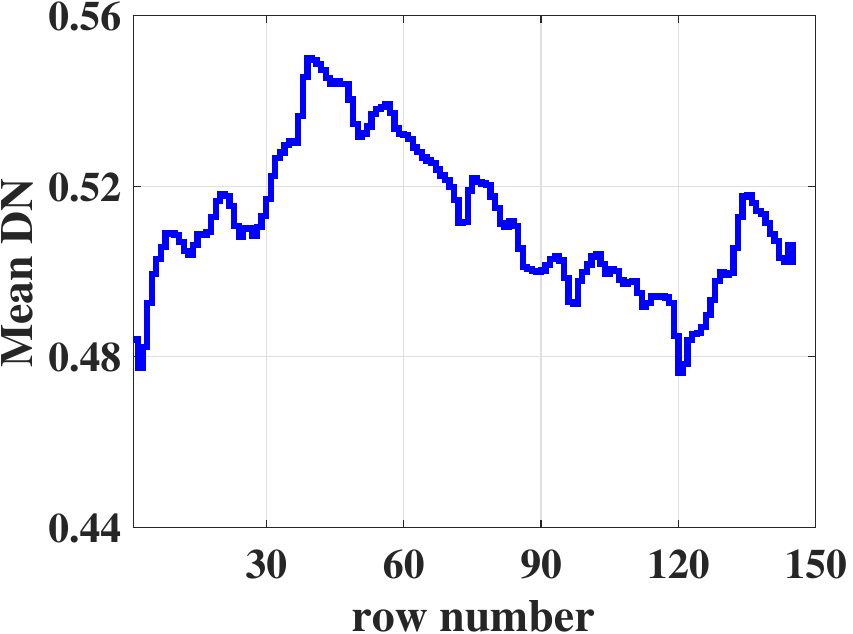}
	}
	\caption{
		The vertical mean profiles of band 50 in Indian dataset.
	}
	\label{real_imshowIndianDN}
\end{figure*}

\subsection{Discussion}\label{Discussion_Parameter}
1) Parameter Analysis:
In the NonMFWTNN model, five parameters need to be discussed, i.e., the modal weight parameter $\alpha$, the regularization parameters $\lambda$, $\tau$ and frequency weight parameter $C_1$, and $C_2$.
All experiments are done under Case 2.

The modal weight $\alpha$ is to control the ratio between different modes.
The low-rank characteristics of the spatial dimensions are similar.
To reduce the parameters, we set the weights of the two spatial dimensions to equal weights.
Therefore, it is set to $\alpha=[1,1,\alpha_3]/(1+1+\alpha_3) $, which makes $\sum_{i=1}^3{\alpha _i}=1$.
We study the impact of $\alpha_3$ which is selected from $\{ 0.1 - 1.0 \} $ on the robustness of the model, and its optimal value is obtained at $0.2$.
It can be observed from Fig. \ref{DCNEW_PSNR_alpha} that when $\alpha_3$ takes a smaller value, the PSNR value of the denoised result is better.
This is because the low-rank characteristic of the spectral mode is stronger than the spatial mode.
The regularization parameter $\tau$ controls the weight of the Gaussian noise term.
Its value is related to the intensity of Gaussian noise.
Therefore, it is set to $\tau=\tau_n/ \sigma $,  where $\sigma$ is Gaussian noise intensity.
We study the impact of $\tau_n$ which is selected from $\{ 10^{-5} - 1.6\times 10^{-4} \} $ on the robustness of the model, and its optimal value is obtained at $10^{-4}$.
It can be observed from Fig. \ref{DCNEW_PSNR_tau} that when $\tau \le 0.0001$, the MPSNR of the denoised result is low and fluctuates greatly.
The main reason is that when $\tau$ is small, the ability to remove Gaussian noise is weak, and there is a lot of remaining Gaussian noise.
The regularization parameter $\lambda$ controls the weight of the sparse noise term.
Defining $\lambda =\lambda _s( \frac{\alpha _1}{\sqrt{\max \left( n_2,n_3 \right) n_1}}+\frac{\alpha _2}{\sqrt{\max \left( n_3,n_1 \right) n_2}}+\frac{\alpha _3}{\sqrt{\max \left( n_1,n_2 \right) n_3}} ) $.
As observed from Fig. \ref{DCNEW_PSNR_lambda}, our model can obtain remarkable results when $\lambda _s$ which is selected from $\{  0.006 - 0.016\} $.
its optimal value is obtained at $0.011$.
Similar to $\tau$, when $\lambda$ is small, the noise removal effect is not good.
The frequency weight parameter $C_1$ reflects the scaling ratio of the weight.
Fig. \ref{DCNEW_PSNR_C1} show the PSNR values when $C_1$ which is selected from $\{ 0.1-1.0 \} $.
It can be seen from these figures that when $C_1$ is small, the frequency weight becomes smaller, and the noise is more difficult to remove; when $C_1$ is large, the frequency weight becomes larger, and the detailed information of tensor is removed.
its optimal value is obtained at $0.6$.
The frequency weight parameter $C_2$ reflects the initial weight of the iteration.
Fig. \ref{DCNEW_PSNR_C2} show the PSNR values when $C_2$ which is selected from $\{ 0.1-1.0 \} $.
its optimal value is obtained at $0.6$.
It shrinks equally of all frequency components.
It inherits the advantages of TNN and ensures that noise can be removed correctly.

2) Convergence Analysis: Fig.\ref{Parameter_analysis03} shows the
MPSNR value, the MSSIM value, and the Error value of the proposed model according to the iterations.
As shown in Fig.\ref{Parameter_analysis03}, the number of iterations increases, the value gradually stabilizes by the proposed method, which justifies the numerical convergence of the proposed method.

3) Shortcoming  Analysis: The MFWTNN and NonMFWTNN can better explore the low-rank structure of HSI, and achieve excellent results in removing the mixed noise of Gaussian and salt and pepper.
But our models do not work for removing dead line noise.
On the one hand, our model does not describe prior information for dead line noise.
On the other hand, the dead line noise is also low rank, and it is easy to be regarded as a part of HSI, which causes dead line noise to remain in the restored HSI.
The following two strategies may solve this shortcoming.
Similar to noise modeling of strip noise \cite{DLB,SSLR_SSTV,chen2017group}, dead line noise can be  modeled more accurately.
Another solution is to use the PnP framework to integrate the BM3D or deep learning frameworks into the model to achieve the purpose of removing dead line noise \cite{zhuang2021hyperspectral, liu2021hyperspectral, ZENG_HSI_tensor}.

\begin{figure}
	\centering
	\subfigure[ ] {
		\label{DCNEW_PSNR_alpha}
		\includegraphics[width=0.28\columnwidth]{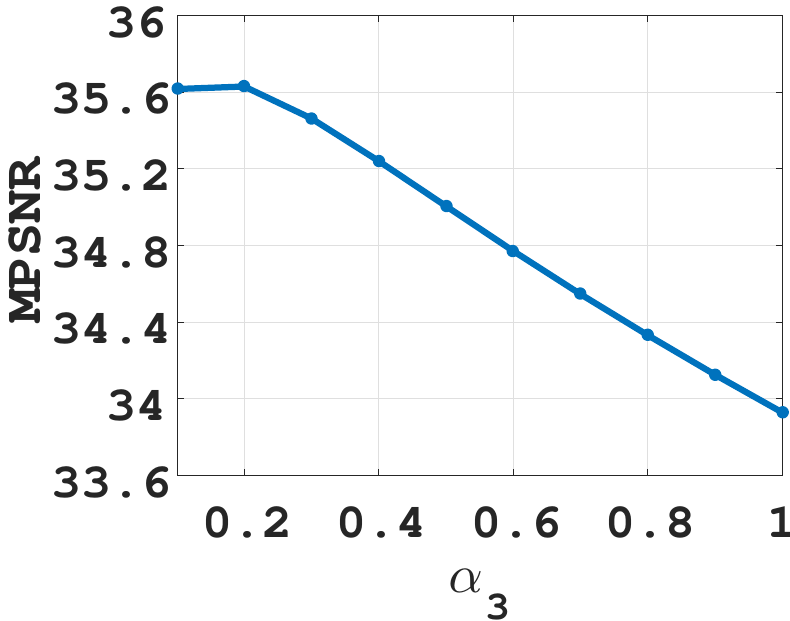}
	}
	\subfigure[ ] {
		\label{DCNEW_PSNR_lambda}
		\includegraphics[width=0.28\columnwidth]{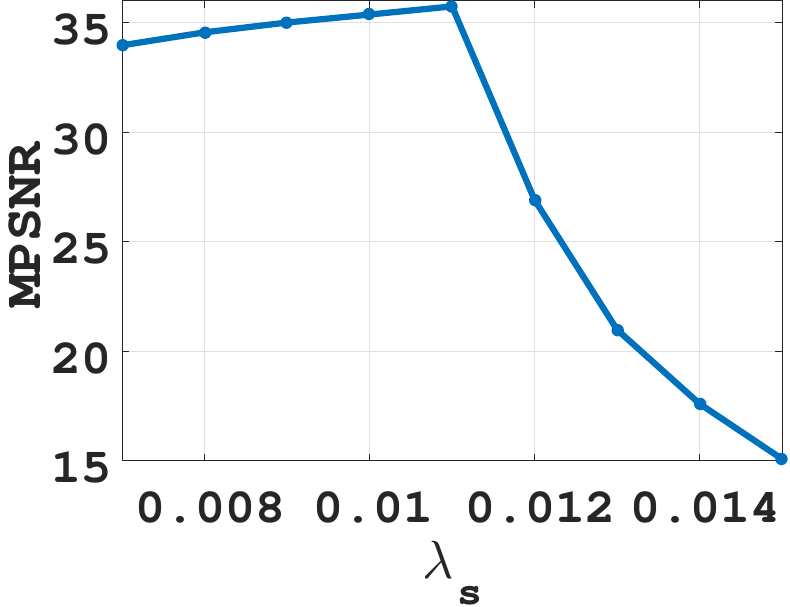}
	}
	\subfigure[ ] {
		\label{DCNEW_PSNR_tau}
		\includegraphics[width=0.28\columnwidth]{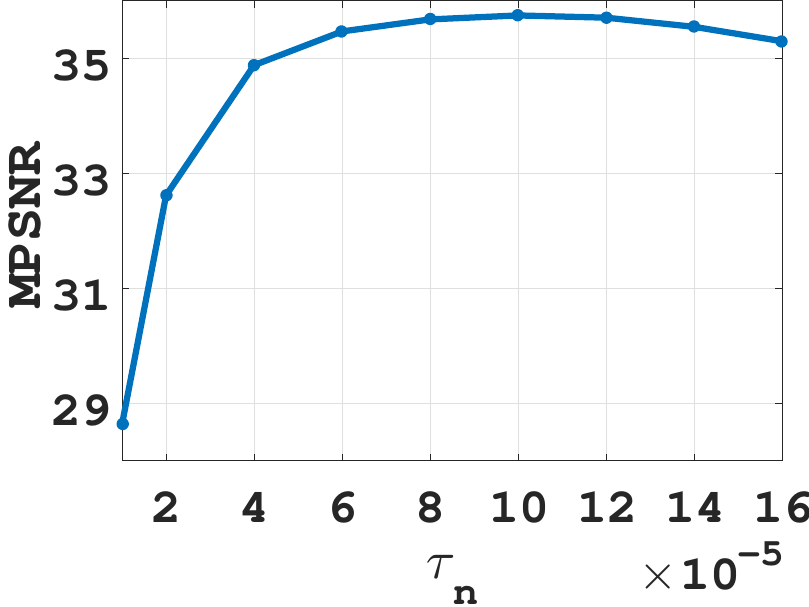}
	}
	\subfigure[ ] {
		\label{DCNEW_PSNR_C1}
		\includegraphics[width=0.28\columnwidth]{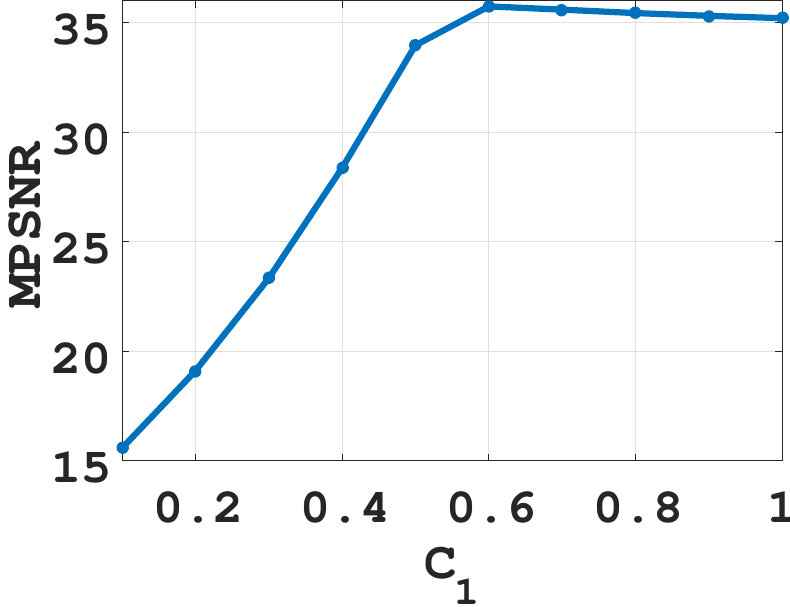}
	}
	\subfigure[ ] {
		\label{DCNEW_PSNR_C2}
		\includegraphics[width=0.28\columnwidth]{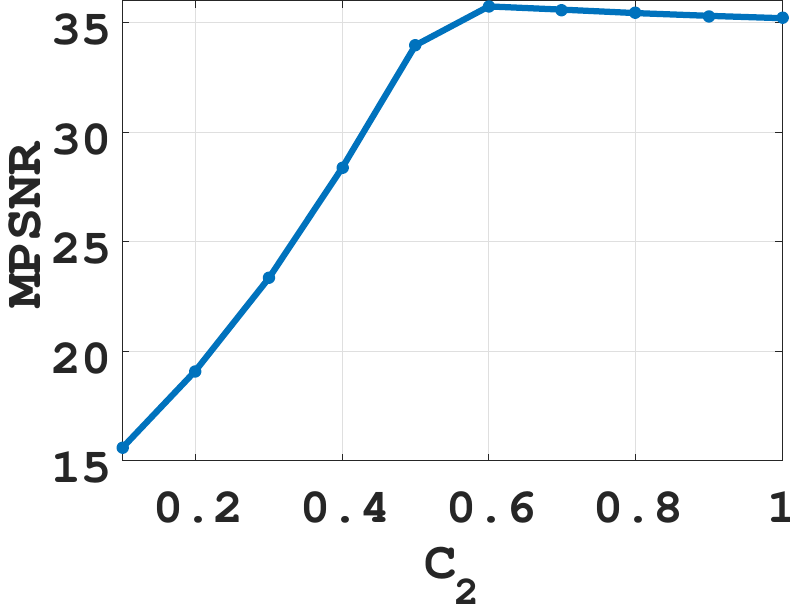}
	}
	\caption{ MPSNR values with respect to different values of parameters.
	}
	\label{Parameter_analysis02}
\end{figure}

\begin{figure}[h] \centering    	
	\subfigure[ ] {
		\label{DCNEW_PSNR_iter}
		\includegraphics[width=0.28\columnwidth]{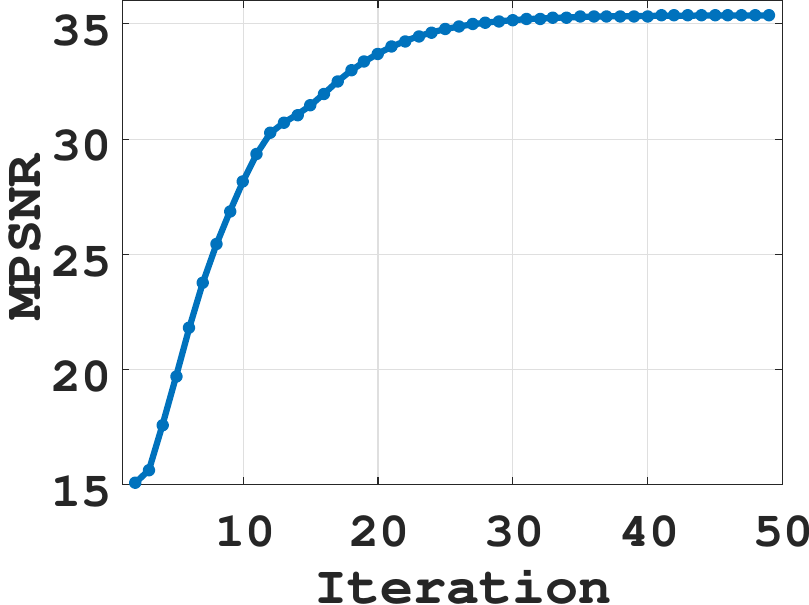}
	}
	\subfigure[  ] {
		\label{DCNEW_SSIM_iter}
		\includegraphics[width=0.28\columnwidth]{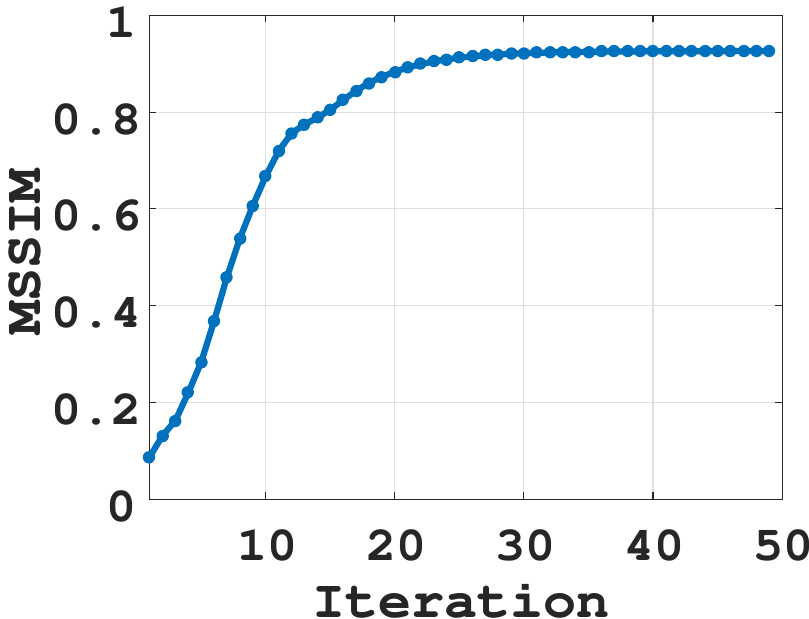}
	}
	\subfigure[ ] {
		\label{DCNEW_PSNR_Error}
		\includegraphics[width=0.28\columnwidth]{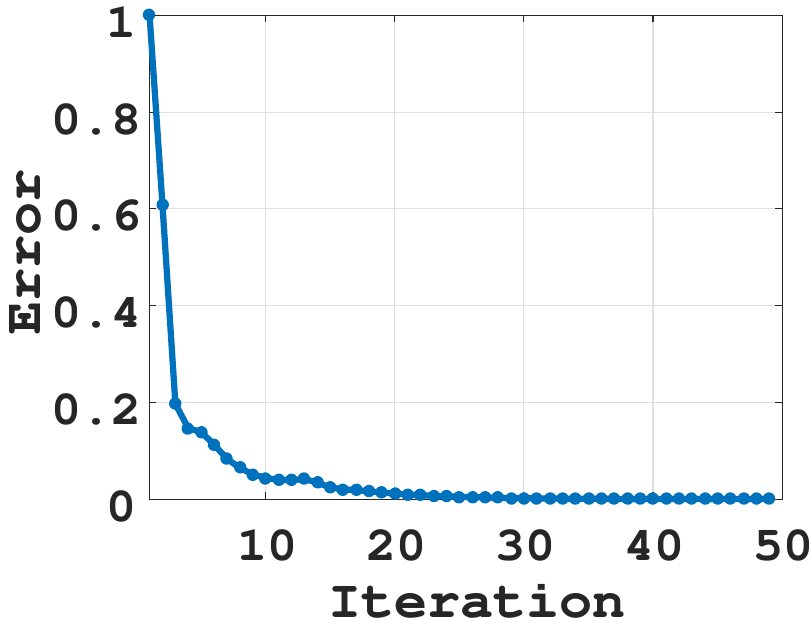}
	}
	\caption{  MPSNR , SSIM and Error  values with respect to the number of iterations.
	}
	\label{Parameter_analysis03}
\end{figure}

\section{Conclusion}
In this paper, we proposed the multi-modal and frequency-weighted tensor nuclear norm (MFWTNN) and Non-convex MFWTNN to approximate the rank function and applied them to the HSI denoising task.
The MFWTNN not only considered the correlation between two spatial modes and one spectral mode but also considered the frequency meaning of HSI in the Fourier domain.
The MFWTNN is the fusion of TNN and SNN, which inherits the advantages of t-SVD and Tucker decomposition.
And we gave the frequency weight adaptive calculation way.
Since the calculation of weights depends on iteratively restored and observed HSI, our models are more robust to different degrees of noise.
It can better characterize the low-rankness of HSI.
Besides, based on the MFWTNN model, we also considered the physical meaning of the internal singular values of the frequency slice and proposed its non-convex approximation.
They powerfully improved capability and flexibility for describing low-rankness in HSIs.
The experiments conducted with both simulated and real HSI datasets showed that our proposed models based HSI denoising model were competitive methods to remove the hybrid noises.

\bibliographystyle{IEEEtran}
\bibliography{IEEEabrv,reference}

\appendices
\section{ Appendix}

\begin{lemma}\cite{logsum} \label{lemma2}
	Let $0<\tau$ and $0<\varepsilon<\min \left(\sqrt{\tau}, \frac{\tau}{y}\right)$, the
	following problem:
	\begin{equation}\label{lognuclearproblem01}
	\min _{x} f(x)=\tau \log (|x|+\varepsilon)+\frac{1}{2}(x-y)^{2}
	\end{equation}
	has a local minimal
	\begin{equation}
	\mathcal{D}_{\tau, \varepsilon}(y)=\left\{\begin{array}{ccc}
	0 & \text { if } & c_{2} \leq 0 \\
	\operatorname{sign}(y)\left(\frac{c_{1}+\sqrt{c_{2}}}{2}\right) & \text { if } & c_{2}>0
	\end{array}\right.
	\end{equation}
	where $c_{1}=|y|-\varepsilon$ and $c_{2}=\left(c_{1}\right)^{2}-4(\tau-\varepsilon|y|) .$
\end{lemma}
\begin{thm}\label{thm2}
	For any positive threshold $\tau>0$ ,  $\mathcal{Y} \in \mathbb{R}^{n_{1} \times n_{2} \times n_{3}}$, and frequency weighting $w $, parameter $\varepsilon >0$, the following problem:
	\begin{equation}\label{thmNonFWTNNproblem01}
	\underset{\mathcal{X}}{\operatorname{argmin}} \ \tau\|\mathcal{X}\|_{\mathrm{FW*,Log}}+\frac{1}{2}\|\mathcal{X}-\mathcal{Y}\|_{F}^{2}
	\end{equation}
	is given by the double-weighted tensor singular value thresholding
	\begin{equation}
	\mathcal{X}^{*}=\mathcal{DW}^{w, \tau , \varepsilon }(\mathcal{Y})=\mathcal{U}*\mathcal{S}^{w, \tau ,\varepsilon}*\mathcal{V}^{H},
	\end{equation}
	where $\mathcal{Y}=\mathcal{U}*\mathcal{S}*\mathcal{V}^H,\mathcal{S}^{w,\tau ,\varepsilon}=\,\,\text{ifft\,\,}\left( \left( \overline{\mathcal{S}} \right) ^{w,\tau ,\varepsilon} \right) $ and
	\begin{equation}
	\left( \overline{\mathcal{S}} \right) ^{w,\tau ,\varepsilon}\left( i,j,s \right) =\left\{ \begin{matrix}
	0,&		\text{if\,\,}c_2\leq 0\\
	sign\left( \overline{\mathcal{S}}\left( i,j,s \right) \right) \left( \frac{c_1+\sqrt{c_2}}{2} \right) ,&		\text{if\,\,}c_2>0\\
	\end{matrix} \right.
	\end{equation}
	where $\overline{\mathcal{S}}=\text{fft}\left( \mathcal{S} \right) ,c_1=\left| \left( \overline{\mathcal{S}}\left( i,j,s \right) \right) \right|-\varepsilon ,\,\,\text{and\,\,}c_2=c_{1}^{2}-4\left( w_k\tau -\varepsilon \left| \left( \overline{\mathcal{S}}\left( i,j,s \right) \right) \right| \right)
	$
	
\end{thm}
\begin{proof}
	According to the properties of the tensor in the Fourier domain, the problem \eqref{thmNonFWTNNproblem01} is equivalent to
	\begin{equation}
	\begin{aligned}
	&\arg\min\text{\ }\tau \lVert \mathcal{X} \rVert _{FW*,Log}+\frac{1}{2}\lVert \mathcal{X}-\mathcal{Y} \rVert _{F}^{2}\\
	\Leftrightarrow &\arg\min\text{\ }\tau \frac{1}{n_3}\sum_{k=1}^{n_3}{w_k\log \left( \lVert \boldsymbol{\bar{\textbf{X}}}^{\left( k \right)} \rVert _* \right)}+\frac{1}{2n_3}\lVert \bar{\mathcal{X}}-\bar{\mathcal{Y}} \rVert _{F}^{2}\\
	\Leftrightarrow &\arg\min\text{\ }\frac{1}{n_3}\sum_{k=1}^{n_3}{\tau w_k\sum_{i=1}^{\min \left( n_1,n_2 \right)}{\log \left( \sigma _i\left( \boldsymbol{\bar{\textbf{X}}}^{\left( k \right)} \right) +\varepsilon \right)}}    \\
	&+\frac{1}{2}\lVert \boldsymbol{\bar{\textbf{X}}}^{\left( k \right)}-\boldsymbol{\bar{\textbf{Y}}}^{\left( k \right)} \rVert _{F}^{2}.
	\end{aligned}
	\end{equation}
	Therefore, the problem in \eqref{thmNonFWTNNproblem01} can be divided into $n_3$ subproblems about problem \eqref{lognuclearproblem01}. From Lemma \ref{lemma2}, for $k$ subproblems, its solution is
	\begin{equation}
	\boldsymbol{\bar{\textbf{X}}}^{\left( k \right)} =\mathcal{D}_{w_k\tau,\varepsilon  }(\boldsymbol{\bar{\textbf{Y}}}^{\left( k \right)})=\left\{\begin{array}{ccc}
	0 & if  & c_{2} \leq 0 \\
	\operatorname{sign}(\boldsymbol{\bar{\textbf{Y}}}^{\left( k \right)})\left(\frac{c_{1}+\sqrt{c_{2}}}{2}\right) &if  & c_{2}>0
	\end{array}\right..
	\end{equation}
\end{proof}

%
%

\ifCLASSOPTIONcaptionsoff
  \newpage
\fi

\end{document}